\documentclass[a4paper,USenglish,cleveref, autoref, thm-restate]{lipics-v2021}
\usepackage{mathtools}
\usepackage[]{todonotes}
\usepackage[]{booktabs}
\usepackage[export]{adjustbox}
\usepackage[ruled,linesnumbered,noend]{algorithm2e}
\usepackage{quiver,makecell}
\usepackage{cite}
\usepackage{subcaption}
\usepackage[T1]{fontenc}
\usepackage{paralist}
\usetikzlibrary{shapes.geometric}
\allowdisplaybreaks

\newcommand{\vx}{\mathbf{x}}
\newcommand{\vy}{\mathbf{y}}

\newcommand{\abs}[1]{\left|#1\right|}
\newcommand{\floor}[1]{\left\lfloor#1\right\rfloor}
\newcommand{\upcl}[1]{\mathop{\uparrow}{#1}}
\newcommand{\dist}{{\cal D}}
\newcommand{\act}{{\sf Act}}

\newcommand{\mat}{{\sf Mat}}
\newcommand{\cpath}{{\sf CPath}}
\newcommand{\stoc}{{\sf Stoc}}

\newcommand{\Pb}{\mathbb{P}}
\DeclareMathOperator*{\argmin}{argmin}
\DeclareMathOperator*{\argmax}{argmax}
\newcommand{\reach}{\mathsf{Reach}}
\newcommand{\ket}[1]{|#1\rangle}
\newcommand{\bigket}[1]{\big|#1\big\rangle}
\newcommand{\biggket}[1]{\Big|#1\Big\rangle}
\newcommand{\lrket}[1]{\left|#1\right\rangle}

\newcommand{\Dsched}{\dist_\mathsf{sched}}
\newcommand{\Dconf}{\dist_\mathsf{conf}}
\newcommand{\Dtrans}{\dist_\mathsf{trans}}
\newcommand{\Dmass}{\dist_\mathsf{mass}}
\newcommand{\Dchan}{\dist_\mathsf{chan}}

\usepackage{comment}
\specialcomment{auxproof}
{\mbox{}\newline\textbf{BEGIN: AUX-PROOF}\dotfill\newline}
{\mbox{}\newline\textbf{END: AUX-PROOF}\dotfill\newline}
\excludecomment{auxproof}

\newcommand{\myparagraph}[1]{\smallskip\noindent \textbf{#1}\;}

\newcommand{\csct}{\ensuremath{(\mathbf{CS,CT})}}
\newcommand{\msct}{\ensuremath{(\mathbf{MS,CT})}}
\newcommand{\csmt}{\ensuremath{(\mathbf{CS,MT})}}
\newcommand{\msmt}{\ensuremath{(\mathbf{MS,MT})}}

\newtheorem{problem}[theorem]{Problem}{\bfseries}{\rmfamily}
\newtheorem{notation}[theorem]{Notation}

\crefname{problem}{Problem}{Problems}
\creflabelformat{enumi}{#2(#1)#3}
\crefname{theorem}{Thm.}{Thms}
\crefname{definition}{Def.}{Defs}
\crefname{proposition}{Prop.}{Props}
\crefname{remark}{Rem.}{Remarks}
\crefname{lemma}{Lem.}{Lemmas}
\crefname{notation}{Notation}{Notations}
\crefname{example}{Example}{Examples}
\crefname{proof}{Proof.}{Proofs}
\crefname{example}{Ex.}{Exs}
\crefname{appendix}{Appendix}{Appendixes}
\crefname{algorithm}{Alg.}{Algs}
\crefformat{section}{Sec.\ #2#1#3}
\crefname{figure}{Fig.}{Figs}
\Crefname{equation}{}{}
\crefname{table}{Table}{Tables}
\crefname{line}{Line}{Lines}

\graphicspath{{./figure/}}%helpful if your graphic files are in another directory

\title{Chance and Mass Interpretations of Probabilities in Markov Decision Processes (Extended Version)}%\thanks{}} %TODO Please add

\titlerunning{Chance and Mass Interpretations in MDPs} %TODO optional, please use if title is longer than one line

\author{Yun Chen Tsai}{National Institute of Informatics, Tokyo, Japan \and SOKENDAI (The Graduate University for Advanced Studies), Kanagawa, Japan}{yctsai@nii.ac.jp}{https://orcid.org/0009-0003-7705-9609}{}
\author{Kittiphon Phalakarn}{National Institute of Informatics, Tokyo, Japan}{kphalakarn@nii.ac.jp}{https://orcid.org/0009-0006-5406-7480}{}
\author{S.\ Akshay}{Department of CSE, Indian Institute of Technology Bombay, Mumbai, India}{akshayss@cse.iitb.ac.in}{https://orcid.org/0000-0002-2471-5997}{}
\author{Ichiro Hasuo}{National Institute of Informatics, Tokyo, Japan \and SOKENDAI (The Graduate University for Advanced Studies), Kanagawa, Japan \and Imiron Co., Ltd., Tokyo, Japan}{hasuo@nii.ac.jp}{https://orcid.org/0000-0002-8300-4650}{}

\authorrunning{Y.C.\ Tsai et. al.} %TODO mandatory. First: Use abbreviated first/middle names. Second (only in severe cases): Use first author plus 'et al.'

\Copyright{Yun Chen Tsai et. al.} %TODO mandatory, please use full first names. LIPIcs license is "CC-BY";  http://creativecommons.org/licenses/by/3.0/


 %TODO mandatory; please add comma-separated list of keywords

\funding{The authors are supported by  
ASPIRE Grant No.\ JPMJAP2301, and 
ERATO HASUO Metamathematics for Systems Design Project (No.\ JPMJER1603), JST. 
S.A.\ is supported by the SBI Foundation Hub for Data Science \& Analytics, IIT Bombay.
I.H.\ is supported by
CREST ZT-IoT Project (No. JPMJCR21M3),   
JST.}%optional, to capture a funding statement, which applies to all authors. Please enter author specific funding statements as fifth argument of the \author macro.

\nolinenumbers %uncomment to disable line numbering

\hideLIPIcs  %uncomment to remove references to LIPIcs series (logo, DOI, ...), e.g. when preparing a pre-final version to be uploaded to arXiv or another public repository

\begin{document}

\maketitle

\begin{abstract}
% \begin{auxproof}
   
% \end{auxproof}

%Markov decision processes (MDPs) are a popular model for decision-making in the presence of uncertainty. In the verification community, the conventional semantics of MDPs takes a state-based view, that is an MDP and a scheduler together induces a distribution over sequences of states. %
%the conventional view of MDPs is as \emph{state transformers} where a configuration is a single active state that makes transitions.An alternative view, namely MDPs as \emph{distribution transformers}, is actively studied recently in the community, too; there an MDP and a scheduler together induces a sequence of distributions over states. %a configuration is a mass distributed over different states. 
Markov decision processes (MDPs) are a popular model for decision-making in the presence of uncertainty. The conventional view of MDPs in verification treats them as state transformers with probabilities defined over sequences of states and with schedulers making random choices. An alternative view, especially well-suited for modeling dynamical systems, defines MDPs as distribution transformers with schedulers distributing probability masses.
Our main contribution is a unified semantical framework that accommodates these two views and two new ones. These four semantics of MDPs arise naturally through identifying different sources of randomness in an MDP (namely \emph{schedulers}, \emph{configurations}, and \emph{transitions}) and providing different ways of interpreting these probabilities (called the \emph{chance} and \emph{mass} interpretations). These semantics are systematically unified through a mathematical construct called \emph{chance-mass (CM) classifier}. As another main contribution, we study a reachability problem in each of the two new semantics, demonstrating their hardness and providing two algorithms for solving them.

\end{abstract}

\section{Introduction}\label{sec:intro}

%\todo{High-level: get rid of Lebesgue}
\myparagraph{MDPs and Conventional Semantics}
Markov decision processes (MDPs) are a classical model combining nondeterminism and randomness, with a rich theory and practical uses in a wide variety of applications~\cite{DBLP:books/wi/Puterman94}. MDPs are often viewed as state transformers where a move from one state to another is defined as follows: from a given state, an action is chosen by a scheduler out of multiple possible actions and then a probabilistic transition is made leading to a resulting state. In other words, a scheduler resolves the nondeterminism in the MDP to give rise to a Markov chain (MC), a purely probabilistic model. In general, a scheduler can also be randomized (also called a mixed scheduler), and thus the actions are chosen according to the probabilities it defines. See~\cite{Baier2008} for more discussion.% Repeatedly applying such moves results in a path of states visited and the probability of such a path can be formalized using cylinder sets as for instance described in~\cite[Chapter 10]{Baier2008}.

% In this context, multiple reachability problems have been considered such as the \emph{threshold reachability} problem, i.e., is there a scheduler such that from a given state if the MDP follows this scheduler then it would reach a given target state with a probability above a given threshold; or the maximization variant, where the goal is to maximize this probability of reaching over the space of all schedulers. All these problems are solvable in polynomial time and even model checking wrt.\ logics such as PCTL turns out to be easily tractable. Hence, state of the art probabilistic model checkers e.g., PRISM~\cite{DBLP:conf/qest/KwiatkowskaNP12} and STORM~\cite{DBLP:journals/sttt/HenselJKQV22} routinely handle large instances and are widely used.

\myparagraph{MDPs as Distribution Transformers}
A different semantics of MDPs (and Markov chains), inspired by dynamical systems, considers them as \emph{distribution transformers}. In this view, we start with a probability distribution or mass over the set of all states of the MDP, and according to the action chosen by a scheduler, this mass gets re-distributed over the states.
% the scheduler distributes this mass according to its choice of actions. 
Therefore, a one-step transition results in a new probability distribution over the set of states, in which sense  the MDP is a linear transformer of distributions. This semantics is studied in multiple recent works~\cite{DBLP:journals/tse/KwonA11,DBLP:conf/qest/KorthikantiVAK10,DBLP:conf/qest/ChadhaKVAK11,DBLP:conf/birthday/AghamovBKNOPV25,LMS14,DBLP:journals/jacm/AgrawalAGT15,DBLP:conf/lics/AkshayGV18,Akshay2023,Akshay2024} with applications ranging from AI, where swarms of robots are modeled, to biological systems, where concentrations of biochemical reagents or populations of cells are modeled. This semantics is also common, e.g., in statistics and the theory of dynamical systems, where mathematicians are often interested in \emph{stationary distributions}, which are essentially fixed points of the distribution transformer. 
%Even the basic reachability problem is significantly harder in this setting; even for Markov chains, it is as hard as the Skolem problem, a long-standing open problem on linear recurrences~\cite{DBLP:journals/siglog/OuaknineW15,DBLP:conf/rp/OuaknineW12,Akshay2015}.

\myparagraph{Reachability Problems in MDPs} The \emph{reachability problem} is fundamental in model checking. In the context of MDPs, it asks whether there exists a scheduler such that the MDP, following this scheduler from a given state, reaches a given \emph{target state} with a certain probability. It is known that the reachability problem for MDPs can be solved in polynomial time, and well-established probabilistic model checkers such as PRISM~\cite{DBLP:conf/qest/KwiatkowskaNP12} and Storm~\cite{DBLP:journals/sttt/HenselJKQV22} can handle large instances efficiently. 
%Furthermore, under this problem setting, the two semantics are simply a matter of preference; that is, both semantics agree on the solution to any instance of the reachability problem.

However, it has been observed that when the reachability problem on MDPs is ``lifted'' and is asked under the distribution transformer semantics, 
% ``lifted'' to the distribution level, 
it becomes significantly harder. The ``lifted'' problem asks for the reachability probability to a \emph{target distribution}, instead of to a target state. The problem is significantly harder even for Markov chains; in particular, it is as hard as the Skolem problem, a long-standing open problem on linear recurrences~\cite{DBLP:journals/siglog/OuaknineW15,DBLP:conf/rp/OuaknineW12,Akshay2015}. 
% Most importantly, if we represent the conventional MDP semantics in terms of distributions, the two semantics no longer agree on the solution to the ``lifted'' reachability problem since the \emph{set of distributions reached} is different in the two semantics. 
%
% Note that the reachability problem in this distribution transformer semantics is very different from that in the conventional semantics: ...

% We note that the reachability problem in the conventional semantics can also be solved  in the distribution transformer semantics---

% a target state can be identified with a Dirac distribution. 

% The opposite reduction is hard, though, since the target distribution does not have a good representation as a single state.

Motivating from these observations, we aim to develop a systematic framework that describes the two semantics and the ``lifted'' reachability problem. This leads to our main contribution: a unified framework for MDPs semantics.

% that captures the above observations; which leads us to our main contribution: a unified framework for MDPs semantics.
%our work proposes a unified framework that captures the two semantics and introduces two new semantics of MDPs; the unified framework allows us to look into the ``lifted'' reachability problem in different semantics and understanding the differences in a systematic way.
% the differences in the two semantics via a unified framework that captures the two semantics.
%Under such setting, the two semantics are just a matter of preference 

\begin{table}[tbp]\centering
  \caption{Examples of the four semantics of MDPs}\vspace{-0.5em}
  \label{table:fourVarExampleIntro}
  \scalebox{0.8}{
    \centering
    \begin{tabular}{ccc}\toprule
    &\csct{} semantics&\csmt{} semantics\\\cmidrule{2-3}
    \multirow[c]{3}{*}[10pt]{\begin{tikzpicture}[scale=0.15]
    \begin{scope}[local bounding box=mdp]
        \tikzstyle{every node}+=[inner sep=0pt]
    \draw [black] (28.6,-25.9) circle (3);
    \draw (28.6,-25.9) node {$q_{0}$};
    \draw [black] (36.7,-18.7) circle (3);
    \draw (36.7,-18.7) node {$q_{1}$};
    \draw [black] (36.7,-33.8) circle (3);
    \draw (36.7,-33.8) node {$q_{2}$};
    \draw [black] (35.706,-21.515) arc (-29.4034:-67.32952:8.635);
    \fill [black] (35.71,-21.51) -- (34.88,-21.97) -- (35.75,-22.46);
    \draw (36.63,-24.22) node [below] {$a,\mbox{ }0.5$};
    \draw [black] (28.942,-22.945) arc (-199.54446:-257.18846:6.638);
    \fill [black] (33.73,-18.69) -- (32.83,-18.38) -- (33.06,-19.36);
    \draw (28.02,-19.72) node [above] {$b,\mbox{ }0.1$};
    \draw [black] (31.444,-26.819) arc (63.35145:28.08087:9.836);
    \fill [black] (35.71,-30.98) -- (35.77,-30.04) -- (34.89,-30.51);
    \draw (36.61,-28.09) node [above] {$a,\mbox{ }0.5$};
    \draw [black] (33.735,-33.51) arc (-107.29404:-161.27364:7.342);
    \fill [black] (33.73,-33.51) -- (33.12,-32.79) -- (32.82,-33.75);
    \draw (28.02,-32.24) node [below] {$b,\mbox{ }0.9$};
    \end{scope}
    \draw (mdp.south)+(0,-2) node[below] {With Scheduler $0.4a + 0.6b$};
\end{tikzpicture}

% \begin{tikzpicture}[scale=0.13]
%     \begin{scope}[local bounding box=mdp]
%         \tikzstyle{every node}+=[inner sep=0pt]
%         \draw [black] (27.9,-30.5) circle (3);
%         \draw (27.9,-30.5) node {$q_0$};
%         \draw [black] (39.1,-23.6) circle (3);
%         \draw (39.1,-23.6) node {$q_1$};
%         \draw [black] (39.7,-36) circle (3);
%         \draw (39.7,-36) node {$q_2$};
%         \draw [black] (30.45,-28.93) -- (36.55,-25.17);
%         \fill [black] (36.55,-25.17) -- (35.6,-25.17) -- (36.13,-26.02);
%         \draw (30.81,-26.55) node [above] {$a,\mbox{ }0.5$};
%         \draw [black] (30.62,-31.77) -- (36.98,-34.73);
%         \fill [black] (36.98,-34.73) -- (36.47,-33.94) -- (36.04,-34.85);
%         \draw (36.45,-32.73) node [above] {$a,\mbox{ }0.5$};
%         \draw [black] (38.125,-38.522) arc (-45.11742:-184.86319:6.547);
%         \fill [black] (38.12,-38.52) -- (37.2,-38.73) -- (37.91,-39.44);
%         \draw (28.03,-40.34) node [below] {$b,\mbox{ }0.9$};
%         \draw [black] (26.473,-27.89) arc (-164.20008:-312.52784:6.677);
%         \fill [black] (37.41,-21.15) -- (37.16,-20.24) -- (36.48,-20.98);
%         \draw (26.65,-19.89) node [above] {$b,\mbox{ }0.1$};
%     \end{scope}
%     \draw (mdp.south)+(0,-2) node[below] {With Scheduler $0.4a + 0.6b$};
% \end{tikzpicture}
}%
    &\includegraphics[width=.3\textwidth,valign=c]{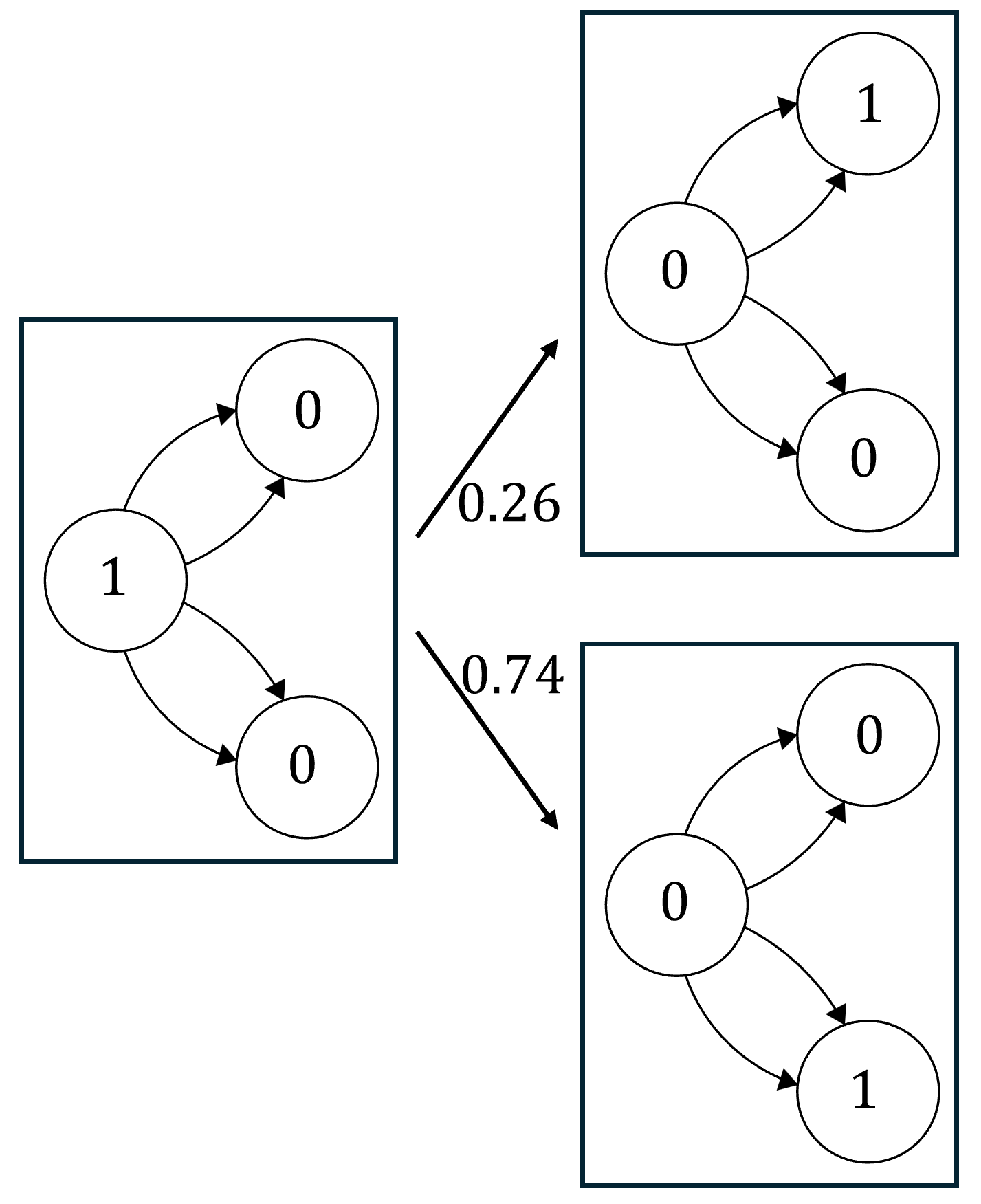}%
    &\includegraphics[width=.3\textwidth,valign=c]{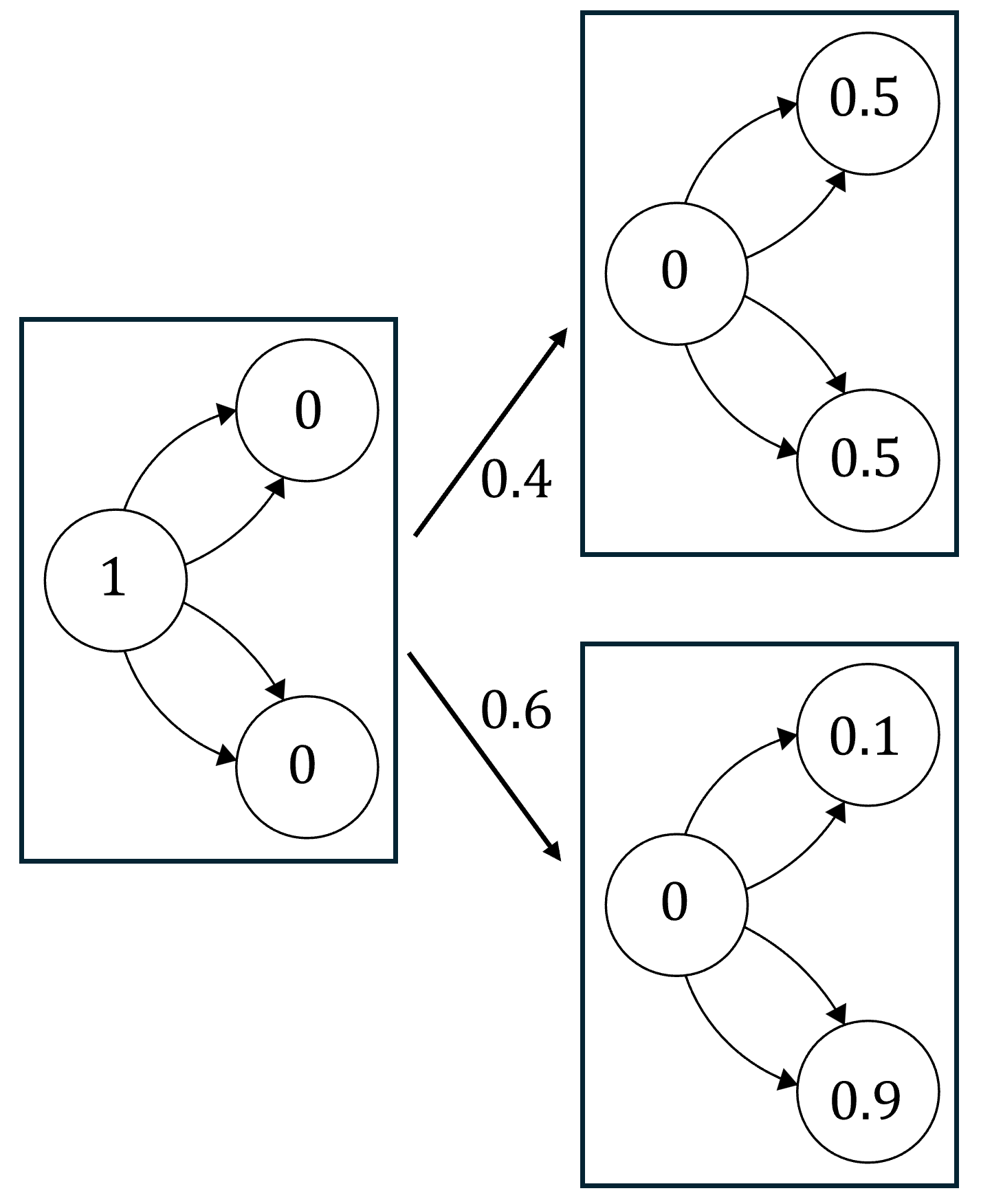}\\\cmidrule{2-3}
    &\msct{} semantics&\msmt{} semantics\\\cmidrule{2-3}
    &\includegraphics[width=.3\textwidth,valign=c]{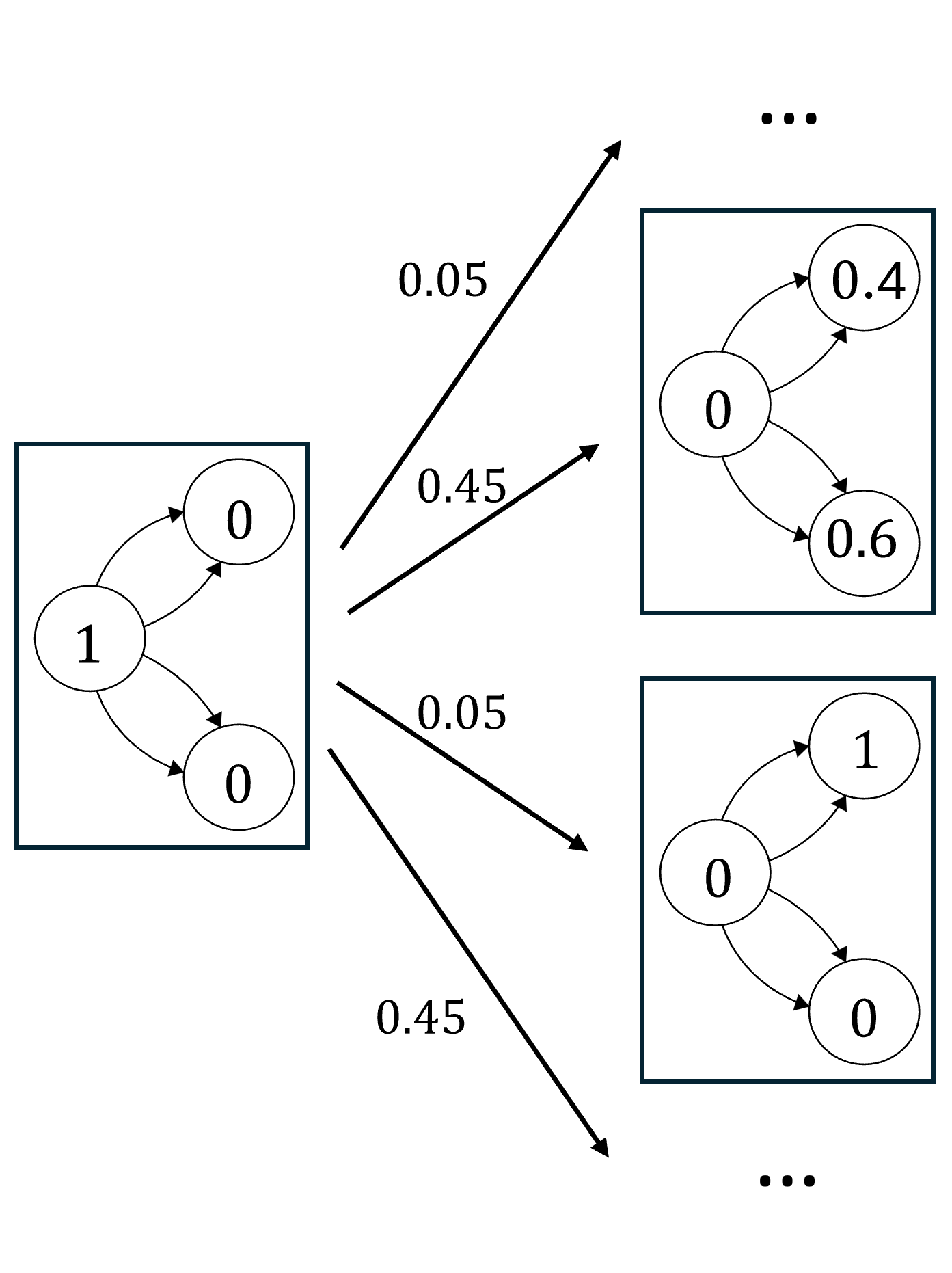}%
    &\includegraphics[width=.3\textwidth,valign=c]{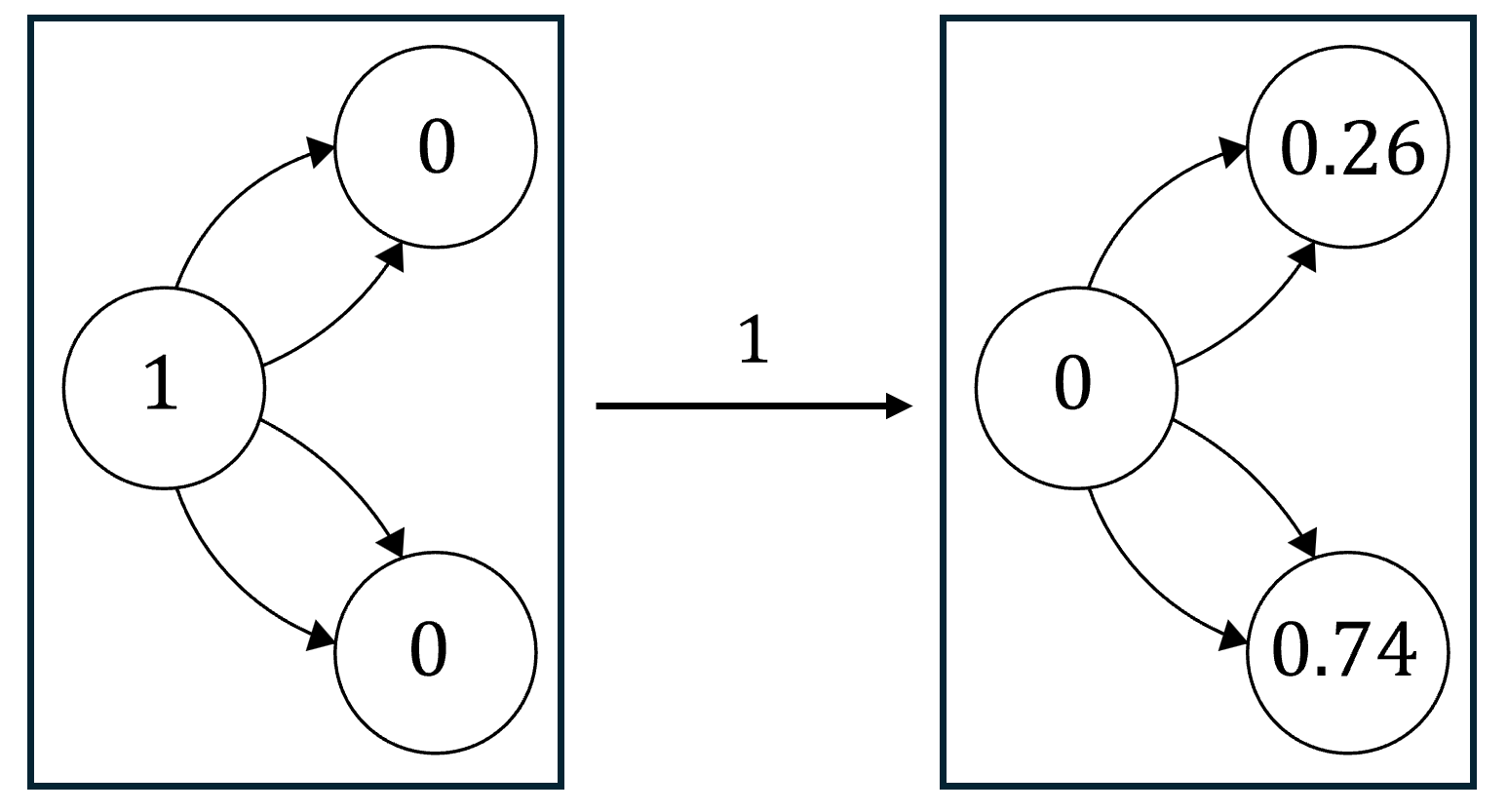}\\\bottomrule
    \end{tabular}
  }
%\vspace{-2em}
\end{table}%\vspace{.2em}

\myparagraph{Four Semantics for MDPs: a Unified Framework}
%In this paper, our goal is to build a general framework that unifies these semantics.
We start by observing that there are two sources of probabilities in MDPs: those arising from randomized/mixed schedulers (such as $0.4$ and $0.6$ in \cref{table:fourVarExampleIntro}) and the inherent transition probabilities (such as $0.1$ and $0.9$ in \cref{table:fourVarExampleIntro}). On the other hand, these probabilities can be interpreted in two ways: the conventional manner, which we call the \emph{chance interpretation}, models probabilities as random events like coin tosses; and the ``distribution transformer'' manner, which we call the \emph{mass interpretation}, models the probabilities as the rate of splitting a mass of particles.

It is then natural to consider four variations of MDP semantics, where each of these two classes of probabilities---coming from the above two sources---is interpreted under either the \emph{chance} or \emph{mass} interpretation (cf.\ \cref{table:fourVarExampleIntro}). For example, the conventional semantics corresponds to applying the \emph{chance} interpretation to both scheduler and transition probabilities, as both schedulers and transitions are modeled as random events (like coin tosses). Similarly, the distribution transformer semantics corresponds to applying the \emph{mass} interpretation to both scheduler and transition probabilities.
% conventional manner or the ``distribution transformer'' manner (cf.\ \cref{table:fourVarExampleIntro}). For brevity, we refer to the two interpretations of probabilities as the \emph{chance} and \emph{mass} interpretations. Thus given an MDP and a randomized scheduler, the choice of semantics determines and changes the Markov chain(s) induced.%that can be obtained.

Let us explain the illustration in \cref{table:fourVarExampleIntro} in more detail. The left side of \cref{table:fourVarExampleIntro} shows a simple MDP with three states and two possible actions from $q_{0}$, where each action leads to a probabilistic transition over states $q_{1}$ and $q_{2}$. A mixed scheduler is fixed, which chooses action $a$ with probability $0.4$ and action $b$ with probability $0.6$. Each of the four Markov chains (MCs) shown on the right side of \cref{table:fourVarExampleIntro}, corresponding to each semantics, is called the \emph{config MC} for the MDP and the specified scheduler. The intuition here is that these config MCs are  Markov chains where each state represents a mass distribution over the state space of the MDP on the left, and they model how mass distributions over the MDP states evolve under the specified interpretation of probabilities.

%which is a Markov chain with a state being a distribution over the state space of the MDP on the left. The basic idea is that the config MCs models 

Let's take the \csmt{}, an abbreviation for \textbf{C}hance-interpreted \textbf{S}cheduler, \linebreak \textbf{M}ass-interpreted \textbf{T}ransitions, semantics as an example. As suggested by its naming, this semantics interprets the probability from the scheduler as a random event, and the probability from the transition as a rate of splitting mass. In one step of the MDP, an (unfair) coin is first tossed to decide the action to be taken by the scheduler: a $0.4$ chance of choosing action $a$ and a $0.6$ chance of choosing action $b$. This corresponds to the two transitions of the config MC as illustrated in \cref{table:fourVarExampleIntro}. Once the action is determined, the mass of $1$ at $q_{0}$ is then split according to the transition probabilities, which yields a distribution of $0.5$ at $q_{1}$ and $0.5$ at $q_{2}$ in the case of action $a$ (the $0.4$ branch). The process for the case of action $b$ is similar. These two distributions are the only distributions that can be reached by the MDP after one step if the probabilities are interpreted as specified.

\cref{table:fourVarExampleIntro} also clearly demonstrates that the reachability problem has different answers under different semantics. For example, the mass distribution ``$0.26$ at $q_1$, $0.74$ at $q_2$'' is reached only under the \msmt{} semantics; the mass distribution ``$1$ at $q_{1}$'' is reached with the probability $0.27$ under the \csct{} semantics, while it is with the probability $0.05$ under the \msct{} semantics.

\myparagraph{Examples}
All four semantics are not only theoretically natural, but also lead to different application scenarios. For the purely random \csct{} semantics, examples are abundant in the CS literature~\cite{Baier2008}, and similarly for the purely distributional \msmt{} semantics, examples are found in the dynamical systems and AI literature~\cite{Akshay2024}. 

%LIPIcs does not support wrapfig
\begin{figure}[tbp]\centering
  \scalebox{.8}{\begin{minipage}{.6\textwidth}
\includegraphics[width=.6\textwidth]{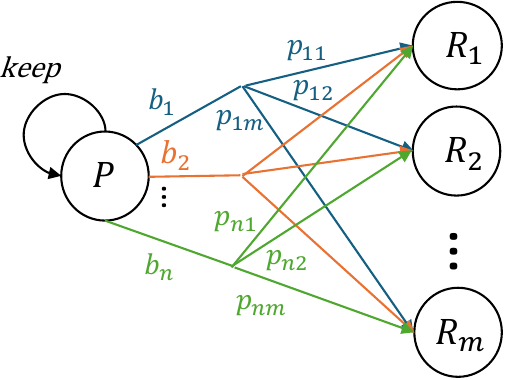}
    \end{minipage}
  }
  \hfill
  \scalebox{.8}{\begin{minipage}{.6\textwidth}
\includegraphics[width=.75\textwidth]{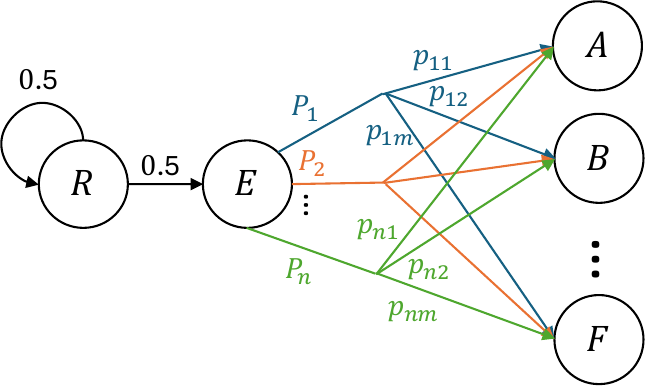}
    \end{minipage}
  }
  \caption{\emph{Casino} MDP for \msct{} sem.\ (left); \emph{exam} MDP for \csmt{} sem.\ (right)}\vspace{-1em}
  \label{fig:casino}
  \label{fig:exam}
\end{figure}

\myparagraph{Casino} An example suited for the \msct{} semantics is the \emph{casino} MDP in \cref{fig:casino}. Here, $P$ is the initial state, from which actions $b_{i}$ ($i\in \{1,\dotsc, n\}$, ``bet at Game $i$'') and $\mathsf{keep}$ (``skip this round'') are offered. Each bet $b_{i}$ can lead to various rewards---represented by states $R_{j}$---and the probabilities $p_{ij}$ are known. What is notable about the \msct{} semantics  is that a mixed scheduler $r_{\mathsf{keep}}\mathsf{keep} + r_{1}b_{1}+\cdots + r_{n}b_{n}$ nicely models a gambler's strategy to split the fund for different games (betting the $r_{1}$ portion of the fund at Game 1, the $r_{2}$ portion at Game 2, etc.). In contrast, the transition probabilities $p_{ij}$ should be interpreted in the chance, not mass, manner. This is because the rewards from each game are random events and only one outcome is produced at each game $i$ in each round (for example you either win or lose with a slot machine, instead of getting a half-and-half outcome). Another point illustrated here is that the scheduler can choose its action depending on the current mass. For example, a possible strategy is to spare half the original fund for the second round (using a strategy $\frac{1}{2}\mathsf{keep}+\cdots$), and to decide how to use it depending on the outcome of the first round (i.e.,\ the mass over the states $R_{1},\dotsc, R_{m}$ after one transition).

%\label{eg:exam} 
\myparagraph{Exam} As an example suited for the \csmt{} semantics, we present the \emph{exam} MDP in \cref{fig:exam}. Here, one chooses a problem set $P_{i}$ from a given repertoire $P_{1},\dotsc, P_{n}$, as commonly done in language tests such as TOEFL and IELTS. Difficulties can differ among problem sets: with the problem set $P_{i}$, the $p_{ij}$ portion of the examinees gets the grade $j\in \{A,B,\dotsc, F\}$, and  this number $p_{ij}$ is known. Exams are repeated, and different problem sets can be used at each time. The state $R$ with its transitions models that the mass of potential examinees decays exponentially over time. Here, a mixed scheduler $r_{1}P_{1}+\cdots+r_{n}P_{n}$ models the strategy that the problem set $P_{i}$ is chosen with the probability $r_{i}$. Each time, one uses the same problem set for all examinees, so $r_{i}$ should be chance-interpreted. In contrast, the transition probabilities $p_{ij}$ should be mass-interpreted: every individual has a chance of getting a certain grade while the overall distribution over grades should follow the distribution indicated by the paper's difficulties.

%otherwise, if it is chance-interpreted, we end up modeling ``a coin is tossed in the exam room and everybody in the room gets the grade A.''

\myparagraph{Contributions}
In this paper, we are interested in 1) formulating a unified framework that captures all four semantics and 2) solving the \emph{threshold reachability} problem under the new semantics we introduce. %The first goal aim to find a common mathematical structure that shares among the four semantics, and the second goal is to understand the complexity of the threshold reachability problem under the new semantics we introduced.
Given an MDP $\mathcal{M}$ and a scheduler $\sigma$, choosing a semantics (out of the four) induces the so-called \emph{config MC} $\mathcal{M}_{\sigma}$; then the threshold reachability problem asks
 % and a target set $H$ of configurations in the config MC $\mathcal{M}_{\sigma}$ (as per each semantics), the threshold reachability problem, asks 
if a target set $H$ can be reached with probability greater than a given threshold $\xi$. As mentioned earlier, this problem is poly-time solvable for the \csct{} semantics, while being Skolem hard for the \msmt{} semantics. Hence, it is natural to ask how hard it is for the new \csmt{} and \msct{} semantics, and if we can develop approaches to solve it.

Our contributions (as summarized in Table~\ref{table:fourVarIntro}) are as follows. 
\begin{table}[tbp]
  \caption{\small Summary of semantics/results: Blue are our contributions, Grey were previously known}\label{table:fourVarIntro}
  \centering
  \begin{tabular}{cc|c|c}
    \toprule
                                     &                               & \multicolumn{2}{c}{{\bf T}ransition}                                                        \\
    %			      \midrule
                                     &                               & {\bf C}hance-interpreted                      & {\bf M}ass-interpreted                      \\
    \midrule
    \multirow{4}{*}{{\bf S}cheduler} & {\bf C}hance-interpreted      & \csct{} semantics                             & \csmt{} semantics                           \\
                                     &                               & \textcolor{gray}{Ptime}                       & \textcolor{blue}{Undecidable, antichain algo} \\
    %\cmidrule(lr){1-1}
    \cmidrule{2-4}
                                     & {\bf M}ass-interpreted        & \msct{} semantics                             & \msmt{} semantics                           \\
                                     &                               & \textcolor{blue}{$\#P$-hard, template algo}
                                     & \textcolor{gray}{Skolem-hard}
    \\
    \bottomrule
  \end{tabular}
\vspace{-2em}
\end{table}
\begin{enumerate}
  \item  A unified framework where different semantics for MDPs are systematically derived
% can be formulated as instances 
% of our framework 
by appropriately specifying the so-called \emph{Chance-Mass (or CM) classifier}. 
%The framework is highly inspired from category theory and some connections in-between are briefly outlined.
  \item For the \csmt{} semantics, we show that this problem is undecidable under a global scheduler and we build an antichain-based backward reachability algorithm for finite union of finitely-generated monotone sets that is sound but necessarily incomplete.
  \item For the \msct{} semantics, we show the problem is $\#P$-hard---the counting counterpart of NP, see e.g.~\cite{AroraB09}---even for a single action in the MDP, and we develop a template-based submartingale which provides a sound 
%and relatively complete
 algorithm for reachability.
  \item We implemented both these algorithms in a prototype. This illustrates the working of our algorithms on a small set of examples including the motivating examples above.
\end{enumerate}

%The structure of the paper is as follows. We start with preliminaries in Section~\ref{sec:prelim}, followed by presenting our unified framework in Section~\ref{sec:unifiedFramework} and its instantiation to four different semantics in Section~\ref{sec:instances}. In Section~\ref{sec:Algo}, we present the complexity results and algorithm on the \msct{} and the \csmt{} semantics. Our experimental evaluation is in Section~\ref{sec:exper} and we end with a brief conclusion in Section~\ref{sec:conclu}.

\myparagraph{Related Work}
Beyond the related work already discussed, here we add some more context and mention a few others. First, the use of distributional \msmt{} semantics for MDPs has been long studied, but mostly from a theoretical point of view. See e.g.,~\cite{DBLP:conf/qest/KorthikantiVAK10}.  Recently, a template-based approach is implemented in~\cite{Akshay2023} for safety constraints and~\cite{Akshay2024} for reach-avoidance, where the authors use Farkas Lemma~\cite{farkas1902theorie} and Handelman's~\cite{handelman1988representing} theorem to develop sound but incomplete algorithms via synthesizing ranking functions. 
%In this paper, we focus on upward-closed and monotone target sets, and
 For the \csmt{} semantics, we use $\gamma$-scaled submartingales for reachability as done e.g., in~\cite{Takisaka2021}. Antichain algorithms have been developed for LTL synthesis e.g., in~\cite{antichain-raskin}, but to the best of our knowledge, we are the first to use them in the context of MDPs.

%%%%%%%%%%%%%%%%%%%%%%%%%%%%%%%%%%%%%%%%%%%%%%%%%%%%%%%%%%%%%%%%%%%%%%
\section{Preliminaries}\label{sec:prelim}
We let $\mathbb{N}, \mathbb{Q}, \mathbb{R}$ denote the sets of natural numbers, rational numbers and real numbers, respectively.
%The set of real numbers between $0$ and $1$, $[n]$ denotes the set of all natural number $k\leq n$. 
The unit interval is denoted by $[0,1]$, and we let $[n]=\{1,2, \dotsc, n \}$.
For sets $A$ and $B$,  $2^{A}$ denotes the powerset of $A$, and $B^{A}$ denotes the set of  functions $f:A\to B$.
%
%
%\todo{Review and omit those which we don't use}
% A finite (or infinite) \emph{path} over a set $A$ is a finite (or infinite) sequence of elements of $A$; 
The collection of all finite (or infinite) sequences over $A$ is denoted by $A^{\ast}$ (or $A^{\omega}$, respectively). $A^{+}\subseteq A^{\ast}$ is the set of nonempty finite sequences.
% For any set $A$, a finite (resp. infinite) path over $A$ is a finite (resp. infinite) sequence of elements over $A$, we denote the set of all finite (resp. infinite) path over $A$ as $A^{\ast}$ (resp. $A^{\omega}$). 
%The length of a sequence $p$ is denoted by $\abs{p}$; it is $\omega$ when $p$ is infinite.
% The length of the path is defined as the number of elements in the sequence, denoted as $\abs{p}$ and we write $\abs{p}=\omega$ when $p$ is a infinite path. 
%For a sequence $p$, $p[i]$ denotes its $i$-th element, and $p[i:j]$ denotes the finite subsequence from the $i$-th to the $j$-th (its length is $j-i+1$). We have $p[i:j]=\varepsilon$ (the empty sequence) if $j<i$.
% over the $i$-th and the $j$-th element over $p$ for any $i\leq j$
% For any path $p\in A^{\ast}\cup A^{\omega}$, denote $p[i]$ the $i$-th element in the path $p$ and $p[i:j]$ the finite sub-path over the $i$-th and the $j$-th element over $p$ for any $i\leq j$.
%% We also abuse the notation to denote $p[i:j]$ as the empty path $\emptyset$. 
%We let $p[i:]$ denote the subsequence starting from the $i$-th element; it is infinite if $p$ is.
% When $p$ is an infinite path, we also denote $p[i:]$ the infinite sub-path that starts at the $i$-th element of $p$. 
% Given a finite sequence $p\in A^{\ast}$, the \emph{cylinder set} induced by $p$ is $\cyl(p):=\{p'\in A^{\omega}\mid p\text{ is a prefix of }p'\}$.
% We naturally extend this definition to define the  cylinder set $\cyl(P)$ induced by a set $P\subseteq A^{\ast}$ of finite sequences.

%\subsection{Probability Theory}
%\paragraph{Probability Theory.}
\myparagraph{Probability Distributions}
% A \emph{measurable space}  $(S,{\cal F})$ consists of a \emph{sample space} $S$
% %is referred to as the sample space and
% and a \emph{$\sigma$-algebra} ${\cal F}\subseteq 2^{S}$.
% % under suitable closure properties. 
% % is a $\sigma$-algebra over $S$.  
% %For a measurable space $(S,{\cal F})$, 
% A \emph{probability measure} $\mu$ over a measurable space $(S,{\cal F})$ is a function $\mu:{\cal F}\to [0,1]$ such that
% % it satisfies the following conditions,
% %\begin{itemize}
% %    \item
% 1) $\mu(E) \geq 0$ for each $ E\in {\cal F}$,
% %    \item
% 2) $\mu(S) = 1$, and
% %    \item
% 3) for any countable family $(E_{k})_{k}$ that is mutually disjoint,
% $\mu\left(\bigcup_{k\in\mathbb{N}}E_{k}\right)=\sum_{k\in\mathbb{N}}\mu(E_{k})$.
% % $\forall \{E_{k}\mid \forall i,j\in\mathbb{N},\, E_{i}\cap E_{j}=\emptyset\}\subset {\cal F}$, $\mu\left(\bigcup_{k\in\mathbb{N}}E_{k}\right)=\sum_{k\in\mathbb{N}}\mu(E_{k})$.
% %\end{itemize}
% A \emph{probability measure space} is a 3-tuple $(S,{\cal F}, \mu)$ where $(S,{\cal F})$ is a measurable space and $\mu$ is a probability measure over it.
%
For any set $X$, we let $\dist(X)=\{d\colon X\to [0,1]\mid \sum_{x}d(x)=1, \mathrm{supp}(d) \text{ is finite}\}$ denote the set of \emph{finitely-supported discrete probability distributions} over $X$. Here, $\mathrm{supp}(d)=\{x\in X\mid d(x)>0\}$ is the \emph{support} of $d$. In this paper, we call $d\in \dist(X)$ simply a \emph{distribution} over $X$. 
%for any $d\in\dist(X)$, the set $\{x\in X\mid d(x)>0\}$ is finite.

% The condition $\sum_{x}d(x)=1$ forces that the support $\{x\mid d(x)>0\}$ is countable. In  measure theory terms, a discrete distribution is nothing but a probability measure over the so-called countable-cocountable algebra~\cite{fremlin2000measure}.

\begin{auxproof}
  We let $\dist(S,{\cal F})$ (or simply $\dist(S)$) denote the set of probability measures over a measurable space $(S,{\cal F})$.
  % For any measurable space $(S,{\cal F})$, we denote the set of all probability measure over $S$ as $\dist(S)$. 
  In this paper, we identify probability measures $\mu$ with functions $d\colon S\to [0,1]$, and say $d\in \dist(S)$, using the equality $\mu(A)=\int_{A}d\, \mathrm{d}\mathcal{L}$ relating them. We can do so since any measurable space we consider is either 1) a countable space (in which case we use the \emph{counting measure} as  $\mathcal{L}$), or 2) a subset of $\mathbb{R}^{n}$ (we use the Lebesgue measure as  $\mathcal{L}$). See e.g.~\cite{AshD99}.
\end{auxproof}

% When $S$ is countable, we can let $\mathcal{F}=2^{S}$, in which case 

% Noted that when $S$ is either countable or subset over $\mathbb{R}^{n}$, for any $\mathbb{M}\in \dist(Q)$, there is a unique (up-to measure equivalence) distribution function $d:S\to [0,1]$ such that $\forall A\in{\cal F}, \mathbb{M}(A)=\int_{A}d\,d{\cal L}$, where ${\cal L}$ dnotes the counting measure (resp. normalized Lebesgue measure) for $S$ is countable (resp. subset of $\mathbb{R}^{n}$)\cite{AshD99}. Since all sample spaces in this paper are either countable or a subset of $\mathbb{R}^{n}$, hence we will abuse the notation to denote ${\cal D}(Q)$ to represent both the set of probability measure and the set of distribution over $S$.

An $n$-dimensional \emph{stochastic matrix} is a $n\times n$ matrix $M\in [0,1]^{n\times n}$  such that $M\mathbf{1} = \mathbf{1}$. We let $\stoc(n)$ denote the set of $n$-dimensional stochastic matrices.

%\subsection{Markov Decision Process}
\myparagraph{Markov Decision Processes} A \emph{Markov decision process} (MDP)
%defined as a 3-tuple of
${\cal M}=(Q,\act,\delta)$ consists of a finite set  $Q$ of \emph{states}, a finite set $\act$  of \emph{actions}, and a \emph{transition function} $\delta:Q\times\act\to\dist(Q)$. We write $M_a$ for the \emph{transition matrix} of an MDP $\mathcal{M}$ for an action $a$; concretely $(M_{a})_{ij} = \delta(q_{i},a)(q_{j})$. Using this, we also write $\mathcal{M}=(Q,\act,(M_{a})_{a\in \act})$ for the same MDP. 

A distribution over $Q$, i.e.\ $d\in \dist(Q)$, shall be called a \emph{configuration} of the MDP $\mathcal{M}$. This terminology is because of their use in \cref{subsec:CMclassConfigMC}.
%\todo{get rid of $\mat$}

\myparagraph{Antichains and Monotone Sets} For a partially-ordered set $(X,\leq)$, a subset $S\subseteq X$ is an \emph{antichain} if for any $x,y\in S$, $x\not\leq y$ and $y\not\leq x$. A subset $S\subseteq X$ is \emph{upward closed} if for any $x\in S$ and $y\in X$, $x\leq y$ implies $y\in S$. For a upward-closed set $S\subseteq X$, 
 the \emph{bottom} of $S$
is defined by  $\floor{S}=\{x\in S\mid \forall y\in S,\,y\leq x\implies x=y\}$; this collects all $\leq$-minimal elements of $S$ and is clearly an antichain.
Given an arbitrary subset $T\subseteq X$, its upward closure $\upcl{T}=\{x\in X\mid \exists y\in T, y\leq x\}$ is upward-closed. 
We say an upward-closed set $S\subseteq X$ is \emph{finitely generated} 
%if its bottom $\floor{S}$ is finite. 
if $S=\upcl{T}$ for some finite set $T$. We do not lose generality if we restrict this $T$ to be an antichain.  In this case, we have $\floor{S}=T$. %\todo{Yun Chen, please prove this. DONE, it's correct.}
The dual notions, namely of \emph{downward-closed set} and its \emph{top}, are defined similarly. A set is \emph{monotone} if it is either upward-closed or downward-closed.

% We identify transition functions $\delta$ with sets of stochatic matrices $\mat_{\act}\subseteq \stoc(\abs{Q})$, and say ${\cal M}=(Q,\act,\mat_{\act})$ is a MDP, with equality $\mat_{\act}:=\{M_{a}\in \stoc(\abs{Q})\mid (M_{a})_{ij} = \delta(q_{i},a)(q_{j})\}$.

%\todo{say this later: `` A configuration $d\in\dist(Q)$ of a MDP ${\cal M}$ is a distribution over its state space $Q$.''}

% Note that for each action $a\in\act$, one can represent the transition function $\delta(q,a)$ with a stochastic matrix $M_{a}$ such that $(M_{a})_{ij} = \delta(q_{i},a)(q_{j})$, and we denote $\mat$ the set of all stochastic matrices. Hence, we will switch between the notion of ${\cal M}=(Q,\act,\delta)$ and ${\cal M}=(Q,\act,\mat)$ whenever one is more convenient. A configuration $d$ of ${\cal M}$ is a density function over its state space $Q$, i.e. $\mu\in \dist(Q)$.

\section{A Unified Framework for Mass and Chance Interpretations}\label{sec:unifiedFramework}
%We present our unified framework. The framework aims to allow separate different sources of randomness and provide an interpretation on each source. 
%\todo{revise}
\begin{auxproof}
 Intuitively, our framework decomposes the transitions of an MDP into two steps. In the first step, the transition is prepared by transforming the configuration into a representation called a \emph{pre-configuration}, this step is formally characterized as $\delta_{\sigma}$ (\cref{def:deltaSigma}) and it is worth noting that this step is universal, that is a pre-configuration does not depend on the semantics of the MDP. In the second step, which characterize the semantics of an MDP, a distribution over configurations is obtained from transforming the pre-configuration with a function parameter $\mathbb{X}$ in our framework, called a \emph{chance-mass classifier} (\cref{def:CMClassifier}). With the two steps combined, mathematically $\mathbb{X}\circ\delta_{\sigma}$, a \emph{configuration Markov chain} (\cref{def:ConfMC}) can be obtained.
\end{auxproof}

Our unified framework translates an MDP $\mathcal{M}$ and a scheduler $\sigma$ to the so-called \emph{configuration Markov chain (config MC)} $\mathcal{M}_{\sigma}$; the last is the semantic basis on which we can ask problems such as reachability. The transition function of the config MC $\mathcal{M}_{\sigma}$ is defined systematically as the composition $\mathbb{X}\circ\delta_{\sigma}$ of a \emph{chance-mass classifier (CM classifier)} $\mathbb{X}$  and a function $\delta_{\sigma}$;   see \cref{def:ConfMC}. There is a notable separation of concerns here:  a CM classifier $\mathbb{X}$ specifies which semantics we use (cf.\ \cref{table:fourVarIntro}) but is independent of  $\mathcal{M}$ or $\sigma$ (\cref{def:CMClassifier}); the function $\delta_{\sigma}$ (\cref{def:deltaSigma}), in contrast, is defined by $\mathcal{M}$ (whose transition function is $\delta$) and a scheduler $\sigma$, but is independent of the choice of the semantics.

% which is a general representation that captures and distinguishes all the randomness in the MDP. This step is universal, that is a pre-configuration does not depend on the semantics of the MDP. 
%involves a universal process of lifting the current \emph{configuration} (a distribution over the states) of the MDP into a \emph{pre-configuration}, which is a general form that contains all the information of the randomness in the MDP; 
% The second step, which characterize the semantics of an MDP, specifies the interpretation for each source of randomness so that a pre-configuration can be transformed as a distribution over configuration. The first step is formally characterized as a function $\delta_{\sigma}$ (\cref{def:deltaSigma}) and the second step, which we call it a \emph{chance-mass classifier}, is treated as a function parameter $\mathbb{X}$ of the framework (\cref{def:CMClassifier}) that specifies the semantics. Composing the two steps, that is mathematically $\mathbb{X}\circ\delta_{\sigma}$, gives a \emph{configuration Markov chain} (\cref{def:ConfMC}).

%\todo[inline]{some writing here. Already hint the intuition of Config MC. Show that it's $X\circ \delta_{\sigma(d_0 ... d_{n})}$. }

%We now present our unified framework. 
The following basic constructs and notations will be used throughout the construction. The  notation $\dist g$---inspired by category theory---may be nonstandard but is useful in \cref{sec:instances}.

\begin{notation}[currying and pushforward]\label{def:curryingPushFwdDist}
  Let $X,Y,Z$ be sets. The \emph{currying} of $f\colon X\times Y\to Z$ is denoted by $f^{\wedge}\colon X\to Z^Y$. 
For  $g\colon X\to Y$, the function 
$\dist g\colon \dist(X)\to\dist(Y)$ carries a (discrete) distribution
 $d\in \dist(X)$ to the \emph{pushforward measure} $\bigl(\dist g(d)\bigr)(y)$ given by
 $\bigl(\dist g(d)\bigr)(y)=\sum_{x\in X, g(x)=y}d(x)$.

% $d\in \dist(X)$ be a (discrete) distribution and $f\colon X\times Y\to Z$, $g\colon X\to Y$ be two functions. The \emph{currying} of $f$ is denoted by $f^{\wedge}\colon X\to Z^Y$ and the \emph{pushforward measure} $\dist g\colon \dist(X)\to\dist(Y)$ is defined by $\bigl(\dist g(d)\bigr)(y)=\sum_{x\in X, g(x)=y}d(x)$.
\end{notation}

\begin{notation}[ket notation] For distributions, we use the following \emph{ket notation}:
  for $d\in \dist(X)$, we write $d = \sum_{x}a_x\ket{x}$ with $a_x=d(x)\in [0,1]$. Intuitively, if we consider $d = \sum_{x}a_x\ket{x}$  as a random variable,  $a_{x}$ is the probability of obtaining the value $x$. More generally, we may use an arbitrary index set $I$ and a function $f\colon I\to X$ and write $\sum_{i\in I}a_{i}\ket{f(i)}$. This distribution assigns to each $x\in X$ the probability $\sum_{i\in I, f(i)=x} a_{i}$.
\end{notation}

We note two special cases of the usages of the ket notation.
% \begin{itemize}
%   \item 
Firstly, when $I=X'\subseteq X$ is a subset of $X$, the distribution $d=\sum_{x\in X'}a_{x}\ket{x}$ assigns $0$ by case distinction:
\begin{math}
  d(x) = a_{x}
\end{math}
if $x\in X'$, and
\begin{math}
  d(x)=0
\end{math}
otherwise.
% \begin{displaymath}
%   d(x)
%   \;=\;
%   \begin{cases}
%     a_{x} &\text{if $x\in X'$, and}
%     \\
%     0 & \text{otherwise.}
%   \end{cases}
% \end{displaymath}
%
% \item 
Secondly,
the pushforward measure $\dist f(d)$, with $d=\sum_{x\in X}a_{x}\ket{x}$ and $f\colon X\to Y$, can be conveniently denoted by $\dist f(d)=\sum_{x\in X}a_{x}\ket{f(x)}$.
% \end{itemize}

\subsection{CM Classifiers, Pre-configurations, and Config Markov Chains}\label{subsec:CMclassConfigMC}
We introduce the central constructs of our framework, under two different views:
\begin{compactitem}
  \item the \emph{concrete} and \emph{element-level} view, describing the constructs as functions; and
  \item the \emph{abstract} and \emph{type-level} view, providing high-level intuitions of the nature of each construct using its type.
\end{compactitem}
The latter view is in fact that of \emph{category theory}~\cite{MacLane71,Awodey06,Jacobs16coalgBook}, but its precise treatment is beyond the scope of this paper. 

We highlight the role of the type-level view through a preview of one of our main constructs, namely a \emph{CM classifier}. It is formulated as a function $\mathbb{X}_{Q}\colon\dist(\dist(\dist(Q)))\to\dist(\dist(Q))$. The intuition of a CM classifier $\mathbb{X}_{Q}$, in the abstract type-level view, is as in the following preview. (The preview focuses on type-level intuitions and is not meant to be understood at the first look---it is really a \emph{preview}. Readers can skim through it and come back.)
\begin{compactitem}
  \item In our framework for MDP semantics, there are three sources of probability distributions (namely \emph{schedulers}, \emph{configurations}, and \emph{transitions}); they are represented by the three occurrences of the distribution operator $\dist$ in the domain of $\mathbb{X}_{Q}$.
  \item Each of the three kinds of probabilities are interpreted either in a \emph{chance} or \emph{mass} manner; these two manners are represented by the two occurrences of $\dist$ in the codomain of $\mathbb{X}_{Q}$.
  \item
        % In order to provide intuitions to the above abstract, type-level and symbolic arguments, 
        For intuition,
        we often annotate $\dist$'s with their intentions, namely $\Dsched,\Dconf,\Dtrans$ (for probabilities of different sources) and $\Dchan,\Dmass$ (for how they are interpreted). In particular, following the above intuitions, a CM classifier $\mathbb{X}_{Q}$ will have the type
        \begin{equation}\label{eq:CMClassifierAnnotated}
          \mathbb{X}_{Q}\colon\Dsched(\Dconf(\Dtrans(Q)))
          %\to
          \longrightarrow
          \Dchan(\Dmass(Q)).
        \end{equation}
        Note that these $\dist$'s have the same mathematical content regardless of the annotations.
  \item The function $\mathbb{X}_{Q}$ specifies which kind of probability (on the left) is interpreted in which manner (on the right). One combinatorial way of presenting such a specification  $\mathbb{X}_{Q}$ is as a mapping from the three $\dist$'s in the domain to the two $\dist$'s in the codomain.
  \item Furthermore, such a mapping can be described by
        1) swapping the order of $\dist$'s, 2) suppressing two $\dist$'s into one, and 3) inserting one $\dist$ where there was none.  See \cref{table:symbX}.  We introduce functions $\lambda, \mu, \eta$ for conducting these operations (\cref{def:mu});  examples of $\mathbb{X}_{Q}$ are thus described as combinations of $\lambda,\mu,\eta$.
\end{compactitem}

\begin{table}[tbp]\centering
  %   \scalebox{0.9}{
  %     \centering
  %     \begin{tabular}{cc}\toprule
  %       \csct{} semantics & \csmt{} semantics  \\\midrule
  %       \makecell{
  %       \input{figure/CMclassCSCT}
  %      }
  %  & \makecell{\input{figure/CMclassCSMT}} \\\midrule
  %       \msct{} semantics  & \msmt{} semantics \\\midrule
  %       \makecell{\input{figure/CMclassMSCT}} & \makecell{\input{figure/CMclassMSMT}} \\\bottomrule
  %     \end{tabular}
  %     }
  %   \caption{combinatorial presentation of CM classifiers $\mathbb{X}$}
  %   \label{table:symbX}
  % =============================================
  \caption{Combinatorial presentation of CM classifiers $\mathbb{X}$}
  \scalebox{0.9}{
    \centering
    \begin{tabular}{cccc}\toprule
      $\mathbb{X}^{\csct{}}$& $\mathbb{X}^{\csmt{}}$& $\mathbb{X}^{\msct{}}$ & $\mathbb{X}^{\msmt{}}$ \\\midrule
      \begin{minipage}{0.25\textwidth}
        \[\begin{tikzcd}[ampersand replacement=\&,column sep=tiny,row sep=small]
	{{\cal D}_{\sf sched}}\& {{\cal D}_{\sf conf}}\& {{\cal D}_{\sf trans}} \\
	{{\cal D}_{\sf sched}}\& {{\cal D}_{\sf trans}}\& {{\cal D}_{\sf conf}} \\
	{{\cal D}_{\sf chance}}\&\&{{\cal D}_{\sf mass}}
	\arrow[from=1-1, to=2-1]
	\arrow[from=1-2, to=2-3]
	\arrow["{\quad\lambda}"{pos=0.4}, from=1-3, to=2-2]
	\arrow[from=2-1, to=3-1]
	\arrow[from=2-2, to=3-1]
	\arrow["{\enspace\mu}"{pos=0.4}, from=2-1, to=3-1]
	\arrow[from=2-3, to=3-3]
\end{tikzcd}\]
      \end{minipage}
                                                       &
      \begin{minipage}{0.25\textwidth}
        \[\begin{tikzcd}[ampersand replacement=\&,column sep=tiny]
	{{\cal D}_{\sf sched}}\& {{\cal D}_{\sf conf}}\& {{\cal D}_{\sf trans}} \\
	{{\cal D}_{\sf chance}}\&\&{{\cal D}_{\sf mass}}
	\arrow[from=1-1, to=2-1]
	\arrow[from=1-2, to=2-3]
	\arrow["{\mu~}"', from=1-3, to=2-3]
\end{tikzcd}\]
      \end{minipage}
                                                       &
      \begin{minipage}{0.25\textwidth}
        \[\begin{tikzcd}[ampersand replacement=\&,column sep=tiny,row sep=small]
	{{\cal D}_{\sf sched}}\& {{\cal D}_{\sf conf}}\& {{\cal D}_{\sf trans}} \\
	{{\cal D}_{\sf sched}}\& {{\cal D}_{\sf trans}}\& {{\cal D}_{\sf conf}} \\
	{{\cal D}_{\sf trans}}\& {{\cal D}_{\sf sched}}\& {{\cal D}_{\sf conf}} \\
	{{\cal D}_{\sf chance}}\&\&{{\cal D}_{\sf mass}}
	\arrow[from=1-1, to=2-1]
	\arrow[from=1-2, to=2-3]
	\arrow["{\quad \lambda}"{pos=0.4}, from=1-3, to=2-2]
	\arrow["{\lambda \quad}"'{pos=0.4}, from=2-1, to=3-2]
	\arrow[from=2-2, to=3-1]
	\arrow[from=2-3, to=3-3]
	\arrow[from=3-1, to=4-1]
	\arrow[from=3-2, to=4-3]
	\arrow["{\mu~}"', from=3-3, to=4-3]
\end{tikzcd}\]
      \end{minipage}
                                                       &
      \begin{minipage}{0.25\textwidth}
        \[\begin{tikzcd}[ampersand replacement=\&,column sep=tiny,row sep=small]
	{{\cal D}_{\sf sched}}\& {{\cal D}_{\sf conf}}\& {{\cal D}_{\sf trans}} \\
	{{\rm Id}}\& {{\cal D}}\& {{\cal D}_{\sf trans}} \\
	{{\cal D}_{\sf chance}}\&\&{{\cal D}_{\sf mass}}
	\arrow[from=1-1, to=2-2]
	\arrow["{\mu~}"', from=1-2, to=2-2]
	\arrow[from=1-3, to=2-3]
	\arrow["{\eta}"'{pos=0.4},from=2-1, to=3-1]
	\arrow[from=2-2, to=3-3]
	\arrow["{\mu~}"', from=2-3, to=3-3]
\end{tikzcd}\]
      \end{minipage}
      \\\bottomrule
    \end{tabular}
  }
  \vspace{-2em}
  \label{table:symbX}
\end{table}

%\todo{Yun Chen, please clean up \cref{table:symbX}. Use Tikz; look at \url{https://q.uiver.app/} Note that braces can be replaced by underlines}
%\todo{YC: the figure is there, I couldn't figure out how to fix the overflow though.}

We move on to the precise technical development, which we mostly do in the concrete, element-level view, with some type-level intuitions supplementing it. An example of each component using the MDP in \cref{table:fourVarExampleIntro} is given in Appendix \ref{sec:UniFrameworkExample}. %\todo{Add it in appendix.} DONE
%

% Our notion of scheduler depends on a distribution of states, rather than on a state.
% We define a scheduler as depending on a distribution---or their history in the memoryful case, rather than a state.
In our general framework, a configuration of an MDP is a distribution over states.
\begin{definition}[configuration, scheduler]\label{def:scheduler}
  %  \begin{itemize}
  %   \item general scheduler (potentially memoryful)
  %   \item memoryless semantics
  %  \end{itemize}
  %Let $\act$ be a finite set of action. A \emph{superposition action} $\bar{a}\in\dist(\act)$ is a distribution over actions.
A \emph{configuration} of an MDP ${\cal M}=(Q,\act,\delta)$ is a distribution $d\in \dist(Q)$. 
  A \emph{(configuration-based) scheduler} for
%  ${\cal M}=(Q,\act,\delta)$ 
$\mathcal{M}$
is a function $\sigma\colon\dist(Q)^{+}\to \dist(\act)$ that maps a sequence of configurations to a distribution over actions.  A scheduler is \emph{memoryless} if its value only depends on the last element of its input, that is for any sequence $d_{1}d_{2}...d_{n}\in\dist(Q)^{+}$, $\sigma(d_{1}d_{2}...d_{n})=\sigma(d_{n})$. Therefore, a memoryless scheduler is a function of the type $\sigma\colon\dist(Q)\to \dist(\act)$. A scheduler is \emph{pure} if it always gives a Dirac distribution, that is for any $p\in\dist(Q)^{+}$, $\sigma(p)(a)=1$ for some $a\in\act$. Otherwise, the scheduler is \emph{mixed}.
\end{definition}

\begin{remark}\label{rem:globalVsLocalActions}
 Note that the scheduler notion in \cref{def:scheduler} uses \emph{global} actions: 
it chooses 
%given a history $d_{1}d_{2}...d_{n}\in\dist(Q)^{+}$, 
an action $a$ is picked from the distribution $\sigma(d_{1}d_{2}...d_{n})\in \dist(\act)$ and \emph{this action $a$ is used at every state}  $q\in Q$. The more common notion of \emph{local} action---allowing different actions at different states---can be encoded as global actions by bloating the action set.
\end{remark}

A CM classifier (as presented in \cref{eq:CMClassifierAnnotated}) takes, as input, three-fold probability distributions, and its layers correspond to the three sources of probability in MDPs. We shall call such data a pre-configuration.

% Our construction considers three sources of probabilities distributions over the state space of an MDP (namely schedulers, configurations and transitions). Mathematically, we formalize it as a \emph{pre-configuration} defined as follows.

\begin{definition}[pre-configuration]\label{def:preconfig}
  Let ${\cal M}=(Q,\act,\delta)$ be an MDP.  A \emph{pre-configuration} of $\mathcal{M}$ is an element $t\in \dist(\dist(\dist(Q)))$, or $t\in \Dsched(\Dconf(\Dtrans(Q)))$ with annotations.
\end{definition}

The following function $\delta_{\sigma}$, as sketched in the beginning of the section, extracts a pre-configuration from an MDP $\mathcal{M}$ and a scheduler $\sigma$, exposing the three sources of probabilities.

% As hinted in the start of this section, the first step is to prepare a transition by transforming a configuration to a pre-configuration. This is done by the function $\delta_{\sigma}$ defined as follows.

\begin{restatable}[$\delta_\sigma$]{definition}{DeltaSigma}\label{def:deltaSigma}
  Let ${\cal M}=(Q,\act,\delta)$ be an MDP and $\sigma$ be a scheduler.
  We define a function $\delta_{\sigma}\colon\dist(Q)^{+}\to \dist(\dist(\dist(Q)))$ as follows: for any $d_{0}d_{1}\dots d_{k}\in\dist(Q)^{+}$,
  \begin{equation}\label{eq:preConfig}
    \delta_{\sigma}(d_{0}d_{1}\dots d_{k})=\sum_{a\in\act}\sigma(d_{0}d_{1}...d_{k})(a)\,\biggket{\sum_{q\in Q}d_{k}(q)\,\bigket{\,\delta(q,a)\,}}.
  \end{equation}
\end{restatable}
%
% \begin{definition}[$\delta_{\sigma}$]\label{def:deltaSigma}
%   Let ${\cal M}=(Q,\act,\delta)$ be an MDP and $\sigma$ be a scheduler.
%   We define a function $\delta_{\sigma}\colon\dist(Q)^{+}\to \dist(\dist(\dist(Q)))$ as follows: for any $d_{0}d_{1}\dots d_{k}\in\dist(Q)^{+}$,
%   \begin{equation}\label{eq:preConfig}
%     \delta_{\sigma}(d_{0}d_{1}\dots d_{k})=\sum_{a\in\act}\sigma(d_{0}d_{1}...d_{k})(a)\,\biggket{\sum_{q\in Q}d_{k}(q)\,\bigket{\,\delta(q,a)\,}}.
%   \end{equation}
% \end{definition}
The annotations in both~\cref{eq:CMClassifierAnnotated} and \cref{def:preconfig} can now be explained using~\cref{eq:preConfig}.

\begin{compactitem}
  \item The innermost distribution $\delta(q,a)$ on the right-hand side of~\cref{eq:preConfig} comes from the transition function $\delta\colon Q\times\act\to \dist(Q)$ of an MDP; so we write $\delta(q,a)\in \Dtrans(Q)$.
  \item The distribution $\biggket{\sum_{q\in Q}d_{k}(q)\,\bigket{\,\delta(q,a)\,}}$, that is one-level higher, has probabilities $d_{k}(q)$ that come from a ``configuration'' $d_{k}$ (this terminology is justified by the memoryless case of \cref{def:ConfMC} later). Thus we say it belongs to $\Dconf(\Dtrans(Q))$.
  \item Finally, the whole right-hand side of each of~\cref{eq:preConfig} is a probability distribution over elements of $\Dconf(\Dtrans(Q))$, where the probabilities (such as $\sigma(d)(a)$) come from the scheduler $\sigma$. Therefore, the right-hand side is an element of $\Dsched(\Dconf(\Dtrans(Q)))$.
\end{compactitem}

%As hinted in~\cref{eq:CMClassifierAnnotated}, a pre-configuration over a state space $Q$ is meant to be an input to a CM classifier $\mathbb{X}_{Q}$. In general, we define the \emph{CM classifier} as follows.
% \noindent
% Now, for the second step, we define the chance-mass classifier as follows. 

\noindent
The following defines CM classifiers, whose intuitions have been discussed earlier in \cref{subsec:CMclassConfigMC}. 

\begin{definition}[chance-mass (CM) classifier]\label{def:CMClassifier}
  A \emph{chance-mass (CM) classifier} is a family $\mathbb{X}=(\mathbb{X}_{Q})_{Q}$ of functions indexed by set $Q$; its $Q$-component has type $\mathbb{X}_{Q}:\dist(\dist(\dist(Q)))\to\dist(\dist(Q))$. When the index $Q$ is obvious we drop it,  writing $\mathbb{X}\colon\dist(\dist(\dist(Q)))\to \dist(\dist(Q))$.
\end{definition}

\begin{remark}\label{rem:naturality}
  For our theoretical development,  in \cref{def:CMClassifier},  it is enough to fix an MDP ${\cal M}=(Q,\act,\delta)$ and only consider the $Q$-component $\mathbb{X}_{Q}\colon\dist(\dist(\dist(Q)))\to \dist(\dist(Q))$ of $\mathbb{X}$. We made the above general definition---it is independent of the choice of $Q$---to emphasize the \emph{uniformity} of $\mathbb{X}$. This uniformity should be formalized as the \emph{naturality in $Q$}, a central notion from category theory~\cite{MacLane71}. We will do so in our future work.
\end{remark}

%We let an MDP $\mathcal{M}$ and a scheduler $\sigma$ induce the \emph{configuration Markov chain (config MC)} $\mathcal{M}_{\mathbb{X}, \sigma}$ under a choice $\mathbb{X}$ of a CM classifier. The config MC is used for the precise definition of MDP semantics, e.g.\ for reachability probabilities.

%We prepare some constructs for the definition of config MC (\cref{def:ConfMC}). For any MDP ${\cal M}$ and scheduler $\sigma$, the following function $\delta_{\sigma}$ yields a pre-configuration; it induces a config MC (\cref{def:ConfMC}) when composed with a CM classifier $\mathbb{X}$.

%\noindent
We are ready to define a config MC $\mathcal{M}_{\mathbb{X}, \sigma}$, a semantic basis on which we formalize e.g.\ reachability problems. As announced, its transition function is a composition $\mathbb{X}\circ \delta_{\sigma}$ of 1) a CM classifier $\mathbb{X}$ that chooses the type of semantics (cf.\ \cref{table:fourVarExampleIntro}) and 2) the function $\delta_{\sigma}$ that exposes three sources of probability in an MDP and forms a pre-configuration.

% Now we have the two steps of transition defined and we are ready to present their composition. We let an MDP $\mathcal{M}$ and a scheduler $\sigma$ induce the \emph{configuration Markov chain (config MC)} $\mathcal{M}_{\mathbb{X}, \sigma}$ under a choice $\mathbb{X}$ of a CM classifier. The config MC is used for the precise definition of MDP semantics, e.g.\ for reachability probabilities. As announced, composing $\delta_{\sigma}$ with a CM classifier $\mathbb{X}$ and a scheduler $\sigma$ gives a config MC.

\begin{definition}[config MC ${\cal M}_{\mathbb{X},\sigma}$]\label{def:ConfMC}
%\todo{Changed the notion of configuration, check subsequent text: checking DONE}
  Let ${\cal M}=(Q,\act,
% (M_{a})_{a\in\act}
\delta
)$ be an MDP, $\mathbb{X}$ be a CM classifier and $\sigma$ be a scheduler. The \emph{configuration Markov chain (config MC)} ${\cal M}_{\mathbb{X},\sigma}=\bigl(\dist(Q)^{+},\delta_{\mathbb{X},\sigma}\bigr)$ is a Markov chain carried by the set $\dist(Q)^{+}$. 
% of \emph{configurations}. 
Its transition function $\delta_{\mathbb{X},\sigma}\colon \dist(Q)^{+}\to\dist(\dist(Q)^{+})$ is defined as follows: for $d_{0}d_{1}\dots d_{k}\in\dist(Q)^{+}$,
  \begin{displaymath}
    \delta_{\mathbb{X},\sigma}(d_{0}d_{1}\dots d_{k})=\sum_{d_{k+1}\in \dist(Q)} \bigg(\mathbb{X}\bigl(\delta_{\sigma}(d_{0}d_{1}\dots d_{k})\bigr)(d_{k+1})\bigg)\,\ket{d_{0}d_{1}...d_{k}d_{k+1}},
  \end{displaymath}
  that is, a config MC state $d_{0}d_{1}\dots d_{k}$ evolves into $d_{0}d_{1}\dots d_{k} d_{k+1}$ with the probability \linebreak $\left(\mathbb{X}\bigl(\delta_{\sigma}(d_{0}d_{1}\dots d_{k})\bigr)\right)(d_{k+1})$. Note that, in this memoryful setting, config MC states record the whole history as sequences and  get extended along transitions. 

  For a memoryless scheduler $\sigma\colon\dist(Q)\to\dist(\act)$, the config MC ${\cal M}_{\mathbb{X},\sigma}=\bigl(\dist(Q),\delta_{\mathbb{X},\sigma}\bigr)$ is carried by the set $\dist(Q)$ of configurations (\cref{def:scheduler}). The function
  \begin{math}
    \delta_{\mathbb{X},\sigma}\colon
    \dist(Q)\to\dist(\dist (Q))
  \end{math}
  is
  \begin{equation}\label{eq:configMCMemoryless}
    \delta_{\mathbb{X},\sigma}(d)=\sum_{d'\in \dist(Q)} \bigg(\mathbb{X}\bigl(\delta_{\sigma}(d)\bigr)(d')\bigg)\,\ket{d'},
  \end{equation}
  that is, a configuration $d\in \dist(Q)$ evolves into $d'\in \dist(Q)$ with the probability $\left(\mathbb{X}\bigl(\delta_{\sigma}(d)\bigr)\right)(d')$. Note that by the definition of $\dist$, the CM classifier $\mathbb{X}\bigl(\delta_{\sigma}(d_{0}d_{1}\dots d_{k})\bigr)$ and $\mathbb{X}\bigl(\delta_{\sigma}(d)\bigr)$ are finitely-supported, hence the use of summation notation here is well-defined.
\end{definition}

% Note that the notion of configuration differs between the memoryful case (history of distributions) and the memoryless case (a distribution). In this paper, we mostly focus on the memoryless setting.
%
We formally state the reachability problem. More details are in Appendix \ref{sec:omittedDef}.
\begin{problem}[threshold reachability problem]\label{prob:GeneralThreshold}
  Assume the setting of \cref{def:ConfMC}. Let $d_{0}\in \dist(Q)$ and  $H\subseteq \dist(Q)$, be an \emph{initial configuration} and a \emph{target set}, respectively. We let $\reach_{\mathbb{X},\sigma}(d_{0},H)$ denote the set of paths that start from $d_{0}$ and eventually reach $H$.
    %\todo{Add in appendix later} DONE
%We define
% \begin{displaymath}
%   \reach_{\mathbb{X},\sigma}(d_{0},H)
%   \;=\;
%   \{d_{0}d_{1}d_{2}\dotsc\in \cpath^{\omega}_{\mathbb{X},\sigma}\mid \text{$d_{i}\in H$ for some $i$}\},
% \end{displaymath}
% that is, it collects the paths that start from $d_{0}$ and eventually reach $H$.
The \emph{threshold reachability problem} for $\xi$ asks if there exists a scheduler $\sigma$ such that
\begin{math}
  \Pb_{\mathbb{X},\sigma}\bigl(\reach_{\mathbb{X},\sigma}(d_{0},H)\bigr)\geq \xi
\end{math}.
\end{problem}

\section{Instances of the Unified Framework and Four Semantics}
\label{sec:instances}
Now we instantiate the framework and derive 
 % demonstrate the instantiation of the unified framework we just presented and show how they lead to 
the four semantics of MDPs in \cref{table:fourVarExampleIntro,table:fourVarIntro}. We first introduce the three operators $\lambda,\mu,\eta$ for constructing the CM classifiers $\mathbb{X}$.

%\subsection{Three Combinatorial Operators $\lambda,\mu,\eta$ for $\dist$}\label{subsec:operators}

%We will use the following ``combinatorial operators'' in the four instances of $\mathbb{X}$  (\cref{table:symbX}).
\begin{definition}[%
    MC2CM $\lambda$,
    suppression $\mu$,
    Dirac $\eta$%
  ]\label{def:mu}
  Let $X$ be a set.
  \begin{compactitem}
    \item The \emph{MC2CM operator} $\lambda_{X}\colon \dist(\dist(X))\to\dist(\dist(X))$ on $X$ is defined as follows:
          \begin{equation}\label{eq:MC2CM}
            \lambda_{X}\biggl(\sum_{i\in I}a_{i}\,\biggket{\sum_{x\in X}b_{ix}\ket{x}}\biggr)=\sum_{f\in X^{I}}\Bigl(\prod_{i\in I}b_{i,f(i)}\Bigr)\,\biggket{\,\sum_{i\in I}a_{i}\,\bigket{f(i)}\,}.
          \end{equation}
    \item
          The \emph{suppression operator} $\mu_{X}\colon \dist(\dist(X))\to\dist(X)$ on $X$ is defined as follows:
          \begin{equation}
            \mu_{X}\biggl(\sum_{i\in I}a_{i}\,\biggket{\sum_{x\in X}b_{ix}\ket{x}}\biggr)
            =\sum_{x\in X} \Bigl(\sum_{i\in I}a_{i}b_{ix}\Bigr)\ket{x}.
          \end{equation}
    \item The \emph{Dirac operator} $\eta_{X}\colon X\to\dist(X)$ is defined by
          \begin{math}
            \eta_{X}(x)=1\,\ket{x}
          \end{math}, that is, $\eta_{X}(x)(x')=1$ if $x=x'$ and $0$ otherwise.
  \end{compactitem}
  Much like in \cref{def:CMClassifier}, we often drop the subscript $X$ and simply write $\lambda,\mu,\eta$.
\end{definition}
\begin{figure*}[tbp]\centering
 \begin{minipage}{0.5\textwidth}
  \centering
  \includegraphics[height=7em,trim={0 3cm 0 2cm},clip]{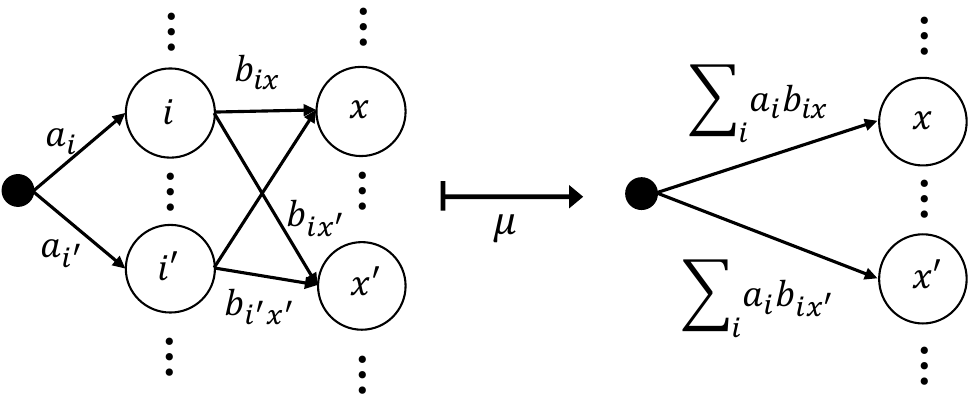}
  \caption{The suppression operator $\mu$}
  \label{fig:mu}
 \end{minipage}
 \hfill
 \begin{minipage}{0.47\textwidth}
  \centering
  \includegraphics[height=7em]{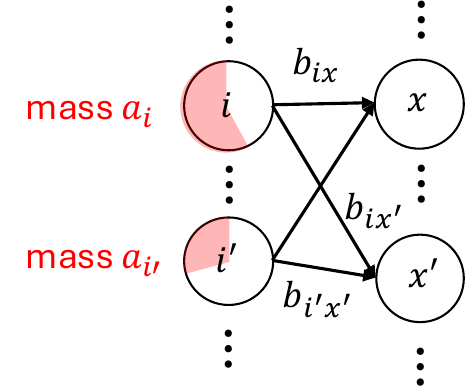}
  \caption{The MC2CM operator $\lambda$, its input}
  \label{fig:lambda}
 \end{minipage}
\end{figure*}

Here are some intuitions (Appendix \ref{subsec:ThreeOpAppen} for more). 
The suppression operator $\mu$, as in \cref{fig:mu}, suppresses two successive probabilistic branching into one: \emph{tossing coins altogether, instead of one-by-one}.
The Dirac operator $\eta$ turns an element $x$ into a trivial distribution.

The MC2CM operator $\lambda$ is much less known. Its type with annotations is $\lambda\colon\Dmass(\Dchan(X))\to\Dchan(\Dmass(X))$; hence the name (``mass-chance to chance-mass''). 
 See \cref{fig:lambda}.
This is how we interpret the input $\sum_{i\in I}a_{i}\,\bigket{\sum_{x\in X}b_{ix}\ket{x}}$ of $\lambda$; more specifically, we are looking at the  mass $a_{i}$ at each $i\in I$, and each $i\in I$ has the transition probability $b_{ix}$ to each $x\in X$.

Here is the intention of $\lambda$:
  the destination of the mass $a_{i}$ at each $i$ is chosen by chance (coin toss), and the whole mass $a_{i}$ goes to the chosen destination. This happens at each $i\in I$, and as a result, each $x\in X$ will acquire the sum of the masses whose destination is $x$.

  This whole process can be modeled as follows.
  \begin{itemize}
    \item Using functions $f\colon I\to X$, we can enumerate all combinations of the destinations $f(i)$ of the mass at each $i\in I$. That is, $f$ is the event ``the destination of $i$ is chosen as $f(i)$.''
    \item For each such event $f\in X^{I}$, the probability for $f$  is
$\prod_{i\in I}b_{i,f(i)}$, where $b_{i,f(i)}$ is the probability of $f(i)$ being chosen as the destination of $i$.
 % $\prod_{i\in I}(\text{$f(i)$ is chosen as the destintion for $i$})$.
    \item Once an event $f$ is chosen,  the mass at each $i\in I$ is moved accordingly. The resulting mass at $x\in X$ is $\sum_{i\in I \text{ s.t.\ } f(i)=x} a_{i}$, and the induced distribution of masses is conveniently represented in the ket notation by $\sum_{i\in I}a_{i}\ket{f(i)}$.
  \end{itemize}
  This explains the right-hand side of~\cref{eq:MC2CM}: it is the distribution over the mass distributions $\sum_{i\in I}a_{i}\ket{f(i)}$ (where $f$ varies), and the probability for each $f$ given by $\prod_{i\in I}b_{i,f(i)}$.

\subsection{Concrete Instances}

Now we present, in the concrete level of elements, the four CM classifiers outlined in \cref{table:fourVarIntro}. We also present explicit formulas for  the config MCs that are derived using them.

% Now we demonstrate how the choice $\mathbb{X}$ instantiate each of the semantics illustrated in \cref{table:fourVarIntro}. 

%  For each instance,
%  \begin{itemize}
%  \item Define $\mathbb{X}$ elementwise,
%  \item refer to Fig. 1, and
%  \item present the construction of the config MC.
%  \item YC: Should we also instantiate the threshold reachability problem here?
%  \end{itemize}

% Add a remark about why we don't do $2^3$ variations.

%\todo{Be consistent about the use of $:=$. }
\begin{definition}[CM classifiers, concrete instances] 
\label{def:CMclassifiersInstances}
We define four CM classifiers 
\linebreak
$
\mathbb{X}^{\csct{}},
\mathbb{X}^{\msct{}},
\mathbb{X}^{\csmt{}},
\mathbb{X}^{\msmt{}}
\colon
 \Dsched(\Dconf(\Dtrans(Q)))\to\Dchan(\Dmass(Q))
$:
\begin{align*}
 &\mathbb{X}^{\csct{}}
 = \mu_{\dist(Q)}\circ\dist\lambda_{Q}
, \quad\text{that is concretely},
\\
&\quad
\mathbb{X}^{\csct{}}\left(\sum_{i\in I}a_{i}\lrket{\sum_{j\in J}b_{i,j}\biggket{\sum_{q\in Q}p_{i,j,q}\ket{q}}}\right)=\sum_{i\in I}\sum_{f\in Q^{J}}\left(a_{i}\prod_{j\in J}p_{i,j,f(j)}\right)\biggket{\sum_{j\in J}b_{i,j}\ket{f(j)}};
\\
&
\mathbb{X}^{\msct{}}=\dist\mu_{(Q)}\circ\lambda_{\dist(Q)}\circ\dist\lambda_{Q}
, \quad\text{that is concretely},
\\
&\quad\resizebox{.95\textwidth}{!}{$\displaystyle{\mathbb{X}^{\msct{}}\left(\sum_{i\in I}a_{i}\lrket{\sum_{j\in J}b_{i,j}\biggket{\sum_{q\in Q}p_{i,j,q}\ket{q}}}\right)=\sum_{f\in Q^{I\times J}}\prod_{(i,j)\in I\times J}p_{i,j,f(i,j)}\biggket{\sum_{(i,j)\in I\times J}a_{i}b_{i,j}\ket{f(i,j)}}}$};
\\
&\mathbb{X}^{\csmt{}}=\dist\mu_{(Q)}
, \quad\text{that is concretely},
\\
&\quad
\mathbb{X}^{\csmt{}}\left(\sum_{i\in I}a_{i}\lrket{\sum_{j\in J}b_{i,j}\biggket{\sum_{q\in Q}p_{i,j,q}\ket{q}}}\right)=\sum_{i\in I}a_{i}\biggket{\sum_{j\in J}\sum_{q\in Q}b_{i,j}p_{i,j,q}\ket{q}}; \text{ and}
\\
&
\mathbb{X}^{\msmt{}}=\eta_{\dist(Q)}\circ\dist\mu_{Q}\circ\mu_{\dist(Q)}
, \quad\text{that is concretely},
\\
&\quad
\mathbb{X}^{\msmt{}}\left(\sum_{i\in I}a_{i}\lrket{\sum_{j\in J}b_{i,j}\biggket{\sum_{q\in Q}p_{i,j,q}\ket{q}}}\right)=\biggket{\sum_{i\in I}\sum_{j\in J}\sum_{q\in Q}a_{i} b_{i,j} p_{i,j,q}\ket{q}}.
\end{align*}
\end{definition}

Here is the illustration of the definition of the $\csct{}$ case (the other cases are similar). As demonstrated in~\cref{table:symbX}, the CM classifier $\mathbb{X}^{\csct{}}$ first applies $\dist\lambda_{Q}$ to the pre-configuration to rearrange the order of the distribution operators and obtain an element in $\Dsched(\Dtrans(\Dconf(Q)))$. Since both the scheduler and the transition distributions are interpreted in the chance manner, we apply the suppression operator $\mu_{\dist(Q)}$ to suppress $\Dsched$ and $\Dtrans$ to $\Dchan$ and leave $\Dconf$ to $\Dmass$.

Those concrete CM classifiers in \cref{def:CMclassifiersInstances} instantiate the general definition of config MC (\cref{def:ConfMC}) as follows. 
To avoid complex super- and subscripts, a  config MC ${\cal M}_{\mathbb{X}^{\csct{}},\sigma}$ shall be denoted by 
${\cal M}^{\csct{}}_{\sigma}$, and similarly for other semantics. 
Here we focus on memoryless schedulers for readability (cf.\ \cref{eq:configMCMemoryless}); the general case is deferred to Appendix \ref{sec:omittedDef}.

\begin{definition}[config MC under concrete semantics, memoryless]
\label{def:ConfigMCsInstances}
  Let ${\cal M}=(Q,\act,\delta)$ be an MDP, and  $\sigma\colon \dist(Q)\to \dist(\act)$ be a (memoryless) scheduler.
  The four config MCs induced by ${\cal M}$ and $\sigma$, under  the four CM classifiers in \cref{def:CMclassifiersInstances}, are concretely described as follows. 
\begin{itemize}
 \item Those config MCs are denoted by
${\cal M}^{\csct{}}_{\sigma}, 
{\cal M}^{\msct{}}_{\sigma}, 
{\cal M}^{\csmt{}}_{\sigma}, 
{\cal M}^{\msmt{}}_{\sigma}$.
 \item They all 
  have $\dist(Q)$ as a state space.
 \item Their transition functions, all of type $\dist(Q)\to\dist(\dist(Q))$, are given as follows.
\begin{align}
 \delta_{\mathbb{X}^{\csct{}},\sigma}(d)(d')
&=
\sum_{a\in\act}\sum_{f\in Q^{Q} \text{ s.t. } \dist f(d)=d'}\sigma(d)(a)\cdot\bigg(\prod_{q\in Q}\delta(q,a)\bigl(f(q)\bigr)\bigg),
\nonumber
\\
    \delta_{\mathbb{X}^{\msct{}},\sigma}(d)(d')
&=
\sum_{f\in Q^{Q\times\act}\text{ s.t. } d'=\dist f(d\otimes\sigma(d))}\prod_{(q,a)\in Q\times\act}\delta(q,a)(f(q,a)),
\nonumber
\\
    \delta_{\mathbb{X}^{\csmt{}},\sigma}(d)(d')
&=
\sum_{a\in\act\text{ s.t. }d'=\mu_{Q}(\dist(\delta^{\wedge}(a))(d))}\sigma(d)(a),
\label{eq:configMCConcreteCSMT}
\\
    \delta_{\mathbb{X}^{\msmt{}},\sigma}(d)(d')
&=\begin{cases}
      1 & \quad \text{if }d'=\sum_{a\in\act}\sum_{q'\in Q}\sum_{q\in Q}\sigma(d)(a) d(q)\delta(q,a)(q')\ket{q'}, \\
      0 & \quad \text{otherwise.}
    \end{cases}
\nonumber
\end{align}
In the second line, $d\otimes \sigma(d)\in \dist(Q\times\act)$ denotes the product, which is a coupling of the configuration $d\in\dist(Q)$ and the distribution  $\sigma(d)\in \dist(\act)$ over actions, that is $(d\otimes \sigma(d))(q,a)=d(q)\cdot\sigma(d)(a)$.
\end{itemize}

\end{definition}

\begin{remark}\label{rem:CSCTConv}
 As we discussed in \cref{sec:intro}, the \csct{} semantics should coincide with the conventional semantics of MDPs, but \cref{def:ConfigMCsInstances}  may not look  that way.  We can show that 1)
 the transition function $\delta_{\mathbb{X}^{\csct{}},\sigma}$ \emph{preserves Dirac-ness} (i.e.\ if  $d$ is Dirac, then $\delta_{\mathbb{X}^{\csct{}},\sigma}(d)(d')>0$ implies $d'$ is Dirac as well), and 2) the definition in \cref{def:ConfigMCsInstances} indeed is the usual semantics when $d$ and $d'$ are Dirac (meaning that these  configurations are single states with mass 1). A proof is in Appendix~\ref{sec:proof}.
\end{remark}

\begin{auxproof}
  \paragraph*{\csct{} Semantics}\label{subsubsec:csct}
 \begin{definition}[$\mathbb{X}^{\csct{}}$]\label{def:CMclassCSCT}
  % Let $Q$ and $\act$ be two finite sets, 
  We define a CM classifier $\mathbb{X}^{\csct{}}\colon\Dsched(\Dconf(\Dtrans(Q)))\to\Dchan(\Dmass(Q))$ by $\mathbb{X}^{\csct{}}=\mu_{\dist(Q)}\circ\dist\lambda_{Q}$.
  That is, concretely,
  % \begin{displaymath}
  %  \mathbb{X}^{\csct{}}(
  % (\text{a general input})
  % )
  % = ... (\text{fix this})  \qquad\sum_{a\in\act}\sum_{f\in Q^{Q}}\rho_{a}\cdot\bigg(\prod_{q\in Q}\delta(q,a)(f(q))\bigg)\biggket{\sum_{q\in Q}d(q)\ket{f(q)}}.
  % \end{displaymath}
  $$\mathbb{X}^{\csct{}}\left(\sum_{i\in I}a_{i}\lrket{\sum_{j\in J}b_{i,j}\biggket{\sum_{q\in Q}p_{i,j,q}\ket{q}}}\right)=\sum_{i\in I}\sum_{f\in Q^{J}}\left(a_{i}\prod_{j\in J}p_{i,j,f(j)}\right)\biggket{\sum_{j\in J}b_{i,j}\ket{f(j)}}.$$
 \end{definition}
 % \begin{definition}[$\mathbb{X}^{\csct{}}$]\label{def:CMclassCSCT}
 %   Let $Q$ and $\act$ be two finite sets, we define a CM classfier \linebreak$\mathbb{X}^{\csct{}}\colon\Dsched(\Dconf(\Dtrans(Q)))\to\Dchan(\Dmass(Q))$ as $\mathbb{X}^{\csct{}}=\mu_{\dist(Q)}\circ\dist\lambda_{Q}$. Concretely, for any superposition action $\bar{a}=\sum_{a\in\act}\rho_{a}\ket{a}\in \dist(\act)$ and any configuration $d\in\dist(Q)$, the CM classifier $\mathbb{X}^{\csct{}}$ acting on the pre-configuration \linebreak$\delta_{\bar{a}}(d)=\sum_{a\in\act}\rho_{a}\biggket{\sum_{q\in Q}d(q)\ket{\delta(q,a)}}$ can be evaluated as follows:
 %   $$\mathbb{X}^{\csct{}}(\delta_{\bar{a}}(d))=\sum_{a\in\act}\sum_{f\in Q^{Q}}\rho_{a}\cdot\bigg(\prod_{q\in Q}\delta(q,a)(f(q))\bigg)\biggket{\sum_{q\in Q}d(q)\ket{f(q)}}.$$
 % \end{definition}

 As demonstrated in~\cref{table:symbX}, the CM classifier $\mathbb{X}^{\csct{}}$ first applies $\dist\lambda_{Q}$ to the pre-configuration to rearrange the order of the distribution operators and obtain an element in $\Dsched(\Dtrans(\Dconf(Q)))$. Since both the scheduler and the transition distributions are interpreted in the chance manner, we apply the suppression operator $\mu_{\dist(Q)}$ to suppress $\Dsched$ and $\Dtrans$ to $\Dchan$ and leave $\Dconf$ to $\Dmass$.

 % To provide some intuition, the purpose of the CM classfier $\mathbb{X}^{\csct{}}$ is to transform the pre-configutation so that the probability distributions $\Dtrans$ over the transition and $\Dsched$ over the scheduler are interpreted in chance, that is ``squeezing'' $\Dtrans$ and $\Dsched$ to $\Dchan$ as illustrated in \cref{table:symbX}. To understand the squeezing, one can think of the distribution $\Dsched(\Dtrans(X))$ as a process of throwing two coins one at a time, the first coin decides the action to be taken and the second decides the transition. Now instead of throwing them one by one, one can throw the two coins at once and decide the action and transition simultanenously, as demonstrated in \cref{fig:suppression}. The mathematical essence of this ``throwing two coins at once'' is exactly captured by the supression operator defined in \cref{def:mu}.\todo{YC: Shall we explain it once only? As it is the same for all four semantics. IH: yes! when we introduce $\mu$. Let's do it later together} This explains the use of $\mu_{\dist(X)}$ in $\mathbb{X}^{\csct{}}$, while the use of $\dist\lambda_{Q}$ is to rearrange the order of distribution operators $\dist$ in the pre-configuration so that $\mu$ can be applied in a desired way ($\Dtrans(\Dsched(X))\to\Dchan(X)$).

 The general definition of config MCs (\cref{def:ConfMC}) is now instantiated as follows.
 %now presented in the following concrete manner. 
 For  readability, we present the case of memoryless schedulers (cf.\ \cref{eq:configMCMemoryless}); the general case is deferred to \cref{def:ConfMCCSCTGeneral}.

 % Next we instantiate the config MC (\cref{def:ConfMC}) using the CM classifier $\mathbb{X}^{\csct{}}$.

 \begin{definition}[config MC under \csct{} semantics, memoryless]\label{def:ConfMCCSCT}
  Let ${\cal M}=(Q,\act,\delta)$ be an MDP, and  $\sigma\colon \dist(Q)\to \dist(\act)$ be a (memoryless) scheduler.
  % (We focus on the CM classifier $\mathbb{X}^{\csct{}}$ from \cref{def:CMclassCSCT}.) 
  The config MC induced by ${\cal M}, \mathbb{X}^{\csct{}}, \sigma$,
  %, defined in \cref{def:ConfMC}, 
  denoted by
  \begin{math}
    {\cal M}^{\csct{}}_{\sigma}=(\dist(Q),\delta_{\mathbb{X}^{\csct{}},\sigma})
  \end{math}, is given  as follows.

  The config MC's state space is the set $\dist(Q)$ of configurations of the MDP $\mathcal{M}$.
  Its transition function $\delta_{\mathbb{X}^{\csct{}},\sigma}\colon\dist(Q)\to\dist(\dist(Q))$
  is defined as follows:
  \begin{equation}\label{eq:configMCTransCSCT}
    \delta_{\mathbb{X}^{\csct{}},\sigma}(d)(d')=\sum_{a\in\act}\sum_{f\in Q^{Q} \text{ s.t. } \dist f(d)=d'}\sigma(d)(a)\cdot\bigg(\prod_{q\in Q}\delta(q,a)\bigl(f(q)\bigr)\bigg).
  \end{equation}

 \end{definition}

 %\todo{Please revise this part}
 As we discussed in \cref{sec:intro}, the \csct{} semantics should coincide with the conventional semantics of MDPs, but \cref{def:ConfMCCSCT}  may not look  that way. In fact, we can show that 1)
 the transition function $\delta_{\mathbb{X}^{\csct{}},\sigma}$ \emph{preserves Dirac-ness} (i.e.\ if  $d$ is Dirac, then $\delta_{\mathbb{X}^{\csct{}},\sigma}(d)(d')>0$ implies $d'$ is Dirac as well), and 2) the definition \cref{eq:configMCTransCSCT} indeed is the usual semantics when $d$ and $d'$ are Dirac (meaning that these  configurations are single states with mass 1).
\end{auxproof}

\begin{auxproof}
  While we motivate the \csct{} semantics as a correspondence of the conventional MDP semantics under our framework, it is not trivial at first sight that \cref{def:ConfMCCSCT} behaves as described in \cref{table:fourVarExampleIntro}. In fact, the conventional MDP semantics is only a special case under such semantics, where we restrict the initial configuration to be a Dirac distribution. Under such restriction, the relation between \csct{} semantics and the conventional MDP semantics can be clearly shown by the following proposition.

  \begin{proposition}[\csct{} semantics with Dirac initial configuration]\label{prop:CSCTDiracInit}
    Assuming the setting of \cref{def:ConfMCCSCT}, if the configuration $d$ is a Dirac distribution, then the transition function $\delta_{\mathbb{X}^{\csct{}},\sigma}$ acting on $d$ is a distribution over Dirac distribution. That is, concretely, for any Dirac distributions $d,d'\in\dist(Q)$ with supports $q,q'\in Q$ respectively, the transition function yields a distribution over Dirac distribution as follows,
    $$\delta_{\mathbb{X}^{\csct{}},\sigma}(d)(d')=\sum_{a\in\act}\sigma(d)(a)\cdot\delta(q,a)(q').$$

    % Further let $p=d_{0}d_{1}...d_{n}$ and fix any $d_{n+1}\in\eta(Q)$, let $q_{j}\in Q$ denote the support of $d_{j}$, that is $d_{j}(q_{j})=1$ for any $j\in [n+1]$, then
    % $$\delta_{\mathbb{X}^{\csct{}},\sigma}(p)(p\cdot d_{n+1})=\sum_{a\in\act}\sigma(p)(a)\cdot\delta(q_{n},a)(q_{n+1}).$$
    % For the case that $\sigma$ being memoryless is defined similarly.
  \end{proposition}

  By replacing each Dirac configuration $d$ with their corresponding support $q$, \cref{prop:CSCTDiracInit} recovers the transition function that we commonly define for conventional MC and MDP.
\end{auxproof}
% Finally, we instantiate the threshold reachability problem stated in \cref{prob:GeneralThreshold} under the restriction to Dirac configuration. The problem statement is as follows.

% \begin{problem}[threshold \csct{} reachability problem with restriction to Dirac configuration]\label{prob:ThresholdCSCT}
%   Assuming the setting of \cref{def:ConfMCCSCT} and let $d_{0}\in\eta(Q)^{+}$ be the initial configuration, $H\subset \eta(Q)$ be the target set and $\lambda\in [0,1]$ be the threshold. The \emph{threshold \csct{}-reachability problem} asks if there exists a scheduler $\sigma$ such that 
%   \begin{math}
%     \Pb_{\mathbb{X}^{\csct{}},\sigma}\bigl(\reach_{\mathbb{X}^{\csct{}},\sigma}(d_{0},H)\bigr)\geq \lambda
% \end{math}.
% \end{problem}

% \begin{proposition}\label{prop:CSCTPtime}
%   \cref{prob:ThresholdCSCT} is solvable in {\sf Ptime}.
% \end{proposition}

% See~\cite{Baier2008} for the proof and details. Noting that it is possible to instantiate this semantics without putting the restriction to Dirac configuration, but its detail study is beyond the scope of this paper. 

% The other semantics are instantiated similarly and we briefly outline each of them below.

\begin{auxproof}
 
 \paragraph*{\msct{} Semantics}\label{subsubsec:msct}
 \begin{definition}[$\mathbb{X}^{\msct{}}$]\label{def:CMclassMSCT}
  We define a CM classifier $\mathbb{X}^{\msct{}}\colon\Dsched(\Dconf(\Dtrans(Q)))\to\Dchan(\Dmass(Q))$ by $\mathbb{X}^{\msct{}}=\dist\mu_{(Q)}\circ\lambda_{\dist(Q)}\circ\dist\lambda_{Q}$. That is, concretely,
  $$\mathbb{X}^{\msct{}}\left(\sum_{i\in I}a_{i}\lrket{\sum_{j\in J}b_{i,j}\biggket{\sum_{q\in Q}p_{i,j,q}\ket{q}}}\right)=\sum_{f\in Q^{I\times J}}\prod_{(i,j)\in I\times J}p_{i,j,f(i,j)}\biggket{\sum_{(i,j)\in I\times J}a_{i}b_{i,j}\ket{f(i,j)}}.$$
  % Let $Q$ and $\act$ be two finite sets, we define a CM classfier \linebreak$\mathbb{X}^{\msct{}}\colon\Dsched(\Dconf(\Dtrans(Q)))\to\Dchan(\Dmass(Q))$ as $\mathbb{X}^{\msct{}}=\dist\mu_{(Q)}\circ\lambda_{\dist(Q)}\circ\dist\lambda_{Q}$. Concretely, for any superposition action $\bar{a}=\sum_{a\in\act}\rho_{a}\ket{a}\in \dist(\act)$ and configuration $d\in\dist(Q)$, the CM classifier $\mathbb{X}^{\msct{}}$ acting on the pre-configuration \linebreak$\delta_{\bar{a}}(d)=\sum_{a\in\act}\rho_{a}\biggket{\sum_{q\in Q}d(Q)\ket{\delta(q,a)}}$ can be evaluated as follows:
  % $$\mathbb{X}^{\msct{}}(\delta_{\bar{a}}(d))=\sum_{f\in Q^{Q\times\act}}\bigg(\prod_{(q,a)\in Q\times\act}\delta(q,a)(f(q,a))\bigg)\biggket{\sum_{(q,a)\in Q\times\act}d(q)\rho_{a}\ket{f(q,a)}}.$$
 \end{definition}

 \begin{definition}[config MC under \msct{} semantics, memoryless]\label{def:ConfMCMSCT}
  Assume a similar setting as in~\cref{def:ConfMCCSCT} with $\mathbb{X}^{\msct{}}$ as the CM classifier instead. Let ${\cal M}^{\msct{}}_{\sigma}=(\dist(Q),\delta_{\mathbb{X}^{\msct{}},\sigma})$ denotes the config MC induced by ${\cal M}, \mathbb{X}^{\msct{}}, \sigma$. Its transition function $\delta_{\mathbb{X}^{\msct{}},\sigma}\colon\dist(Q)\to\dist(\dist(Q))$ is concretely given as follows: %for any configuration $d, d'\in \dist(Q)$, %let $Q^{Q\times\act}_{\sigma}(d,d')\subseteq Q^{Q\times\act}$ denote the set $\{f\in Q^{Q\times\act}\mid d'=\sum_{(q,a)\in Q\times\act}d(q)\sigma(d)(a)\ket{f(q,a)}\}$, then
  \begin{displaymath}
    \delta_{\mathbb{X}^{\msct{}},\sigma}(d)(d')=\sum_{f\in Q^{Q\times\act}\text{ s.t. } d'=\dist f(d\otimes\sigma(d))}\prod_{(q,a)\in Q\times\act}\delta(q,a)(f(q,a)),
  \end{displaymath}
  where $d\otimes \sigma(d)\in \dist(Q\times\act)$ denotes the coupling of the configuration $d\in\dist(Q)$ and the distribution over actions $\sigma(d)\in \dist(\act)$, that is $(d\otimes \sigma(d))(q,a)=d(q)\cdot\sigma(d)(a)$.
 \end{definition}

 % \begin{problem}[threshold \msct{} reachability problem]\label{prob:ThresholdMSCT}
 %   Assuming the setting of \cref{def:ConfMCMSCT} and let $d_{0}\in\dist(Q)^{+}$ be the initial configuration, $H\subset \dist(Q)$ be the target set and $\lambda\in [0,1]$ be the threshold. The \emph{threshold \msct{}-reachability problem} asks if there exists a scheduler $\sigma$ such that 
 %   \begin{math}
 %     \Pb_{\mathbb{X}^{\msct{}},\sigma}\bigl(\reach_{\mathbb{X}^{\msct{}},\sigma}(d_{0},H)\bigr)\geq \lambda
 % \end{math}.
 % \end{problem}
\end{auxproof}

\begin{auxproof}
 \paragraph*{\csmt{} Semantics}\label{subsubsec:csmt}
 \begin{definition}[$\mathbb{X}^{\csmt{}}$]\label{def:CMclassCSMT}
  We define a CM classifier $\mathbb{X}^{\csmt{}}\colon\Dsched(\Dconf(\Dtrans(Q)))\to\Dchan(\Dmass(Q))$ by $\mathbb{X}^{\csmt{}}=\dist\mu_{(Q)}$. That is, concretely,
  $$\mathbb{X}^{\csmt{}}\left(\sum_{i\in I}a_{i}\lrket{\sum_{j\in J}b_{i,j}\biggket{\sum_{q\in Q}p_{i,j,q}\ket{q}}}\right)=\sum_{i\in I}a_{i}\biggket{\sum_{j\in J}\sum_{q\in Q}b_{i,j}p_{i,j,q}\ket{q}}.$$
  % Let $Q$ and $\act$ be two finite sets, we define a CM classfier \linebreak$\mathbb{X}^{\csmt{}}\colon\Dsched(\Dconf(\Dtrans(Q)))\to\Dchan(\Dmass(Q))$ as $\mathbb{X}^{\csmt{}}=\dist\mu_{(Q)}$. Concretely, for any superposition action $\bar{a}=\sum_{a\in\act}\rho_{a}\ket{a}\in \dist(\act)$ and configuration $d\in\dist(Q)$, the CM classifier $\mathbb{X}^{\csmt{}}$ acting on the pre-configuration $\delta_{\bar{a}}(d)=\sum_{a\in\act}\rho_{a}\biggket{\sum_{q\in Q}d(Q)\ket{\delta(q,a)}}$ can be evaluated as follows:
  % $$\mathbb{X}^{\csmt{}}(\delta_{\bar{a}}(d))=\sum_{a\in\act}\rho_{a}\biggket{\sum_{q'\in Q}\sum_{q\in Q}d(q)\delta(q,a)(q')\ket{q'}}.$$
 \end{definition}

 \begin{definition}[config MC under \csmt{} semantics, memoryless]\label{def:ConfMCCSMT}
  Assume a similar setting in~\cref{def:ConfMCCSCT} with $\mathbb{X}^{\csmt{}}$ as the CM classifier instead. Let ${\cal M}^{\csmt{}}_{\sigma}=(\dist(Q),\delta_{\mathbb{X}^{\csmt{}},\sigma})$ denotes the config MC induced by ${\cal M}, \mathbb{X}^{\csmt{}}, \sigma$. Its transition function $\delta_{\mathbb{X}^{\csmt{}},\sigma}\colon\dist(Q)\to\dist(\dist(Q))$ is concretely given as follows:
  \begin{displaymath}
    \delta_{\mathbb{X}^{\csmt{}},\sigma}(d)(d')=\sum_{a\in\act\text{ s.t. }d'=\mu_{Q}(\dist(\delta^{\wedge}(a))(d))}\sigma(d)(a).
  \end{displaymath}
 \end{definition}

 % \begin{problem}[threshold \csmt{} reachability problem]\label{prob:ThresholdCSMT}
 %   Assuming the setting of \cref{def:ConfMCCSMT} and let $d_{0}\in\dist(Q)^{+}$ be the initial configuration, $H\subset \dist(Q)$ be the target set and $\lambda\in [0,1]$ be the threshold. The \emph{threshold \csmt{}-reachability problem} asks if there exists a scheduler $\sigma$ such that 
 %   \begin{math}
 %     \Pb_{\mathbb{X}^{\csmt{}},\sigma}\bigl(\reach_{\mathbb{X}^{\csmt{}},\sigma}(d_{0},H)\bigr)\geq \lambda
 % \end{math}.
 % \end{problem}
\end{auxproof}

\begin{auxproof}
 
\paragraph*{\msmt{} Semantics}\label{subsubsec:msmt}
\begin{definition}[$\mathbb{X}^{\msmt{}}$]\label{def:CMclassMSMT}
  We define a CM classifier $\mathbb{X}^{\msmt{}}\colon\Dsched(\Dconf(\Dtrans(Q)))\to\Dchan(\Dmass(Q))$ by $\mathbb{X}^{\msmt{}}=\eta_{\dist(Q)}\circ\dist\mu_{Q}\circ\mu_{\dist(Q)}$. That is, concretely,
  $$\mathbb{X}^{\msmt{}}\left(\sum_{i\in I}a_{i}\lrket{\sum_{j\in J}b_{i,j}\biggket{\sum_{q\in Q}p_{i,j,q}\ket{q}}}\right)=\biggket{\sum_{i\in I}\sum_{j\in J}\sum_{q\in Q}a_{i} b_{i,j} p_{i,j,q}\ket{q}}.$$
  % Let $Q$ and $\act$ be two finite sets, we define a CM classfier \linebreak$\mathbb{X}^{\msmt{}}\colon\Dsched(\Dconf(\Dtrans(Q)))\to\Dchan(\Dmass(Q))$ as $\mathbb{X}^{\msmt{}}=\eta_{\dist(Q)}\circ\dist\mu_{Q}\circ\mu_{\dist(Q)}$. Concretely, for any superposition action $\bar{a}=\sum_{a\in\act}\rho_{a}\ket{a}\in \dist(\act)$ and configuration $d\in\dist(Q)$, the CM classifier $\mathbb{X}^{\msmt{}}$ acting on the pre-configuration $\delta_{\bar{a}}(d)=\sum_{a\in\act}\rho_{a}\biggket{\sum_{q\in Q}d(Q)\ket{\delta(q,a)}}$ can be evaluated as follows:
  % $$\mathbb{X}^{\msmt{}}(\delta_{\bar{a}}(d))=\biggket{\sum_{a\in\act}\rho_{a}\sum_{q'\in Q}\sum_{q\in Q}d(q)\delta(q,a)(q')\ket{q'}}.$$
\end{definition}

\begin{definition}[config MC under \msmt{} semantics, memoryless]\label{def:ConfMCMSMT}
  Assume a similar setting in~\cref{def:ConfMCCSCT} with $\mathbb{X}^{\msmt{}}$ as the CM classifier instead. Let ${\cal M}^{\msmt{}}_{\sigma}=(\dist(Q),\delta_{\mathbb{X}^{\msmt{}},\sigma})$ denotes the config MC induced by ${\cal M}, \mathbb{X}^{\msmt{}}, \sigma$. Its transition function $\delta_{\mathbb{X}^{\msmt{}},\sigma}\colon\dist(Q)\to\dist(\dist(Q))$ is concretely given as follows:
  \begin{displaymath}
    \delta_{\mathbb{X}^{\msmt{}},\sigma}(d)(d')=\begin{cases}
      1 & \quad d'=\sum_{a\in\act}\sum_{q'\in Q}\sum_{q\in Q}\sigma(d)(a) d(q)\delta(q,a)(q')\ket{q'} \\
      0 & \quad \text{otherwise}
    \end{cases}.
  \end{displaymath}

  %   Let ${\cal M}=(Q,\act,\delta)$ be an MDP, and  $\sigma\colon \dist(Q)\to \dist(\act)$ be a (memoryless) scheduler. The config MC induced by ${\cal M}, \mathbb{X}^{\msmt{}}, \sigma$, defined in \cref{def:ConfMC}, is denoted by 
  % \begin{math}
  %  {\cal M}^{\msmt{}}_{\sigma}=(\dist(Q),\delta_{\mathbb{X}^{\msmt{}},\sigma})
  % \end{math}. Its transition function $\delta_{\mathbb{X}^{\msmt{}},\sigma}\colon\dist(Q)\to\dist(\dist(Q))$ is concretely given as follows: for any configuration $d, d'\in \dist(Q)$,
  %   $$\delta_{\mathbb{X}^{\msmt{}},\sigma}(d)(d')=\begin{cases}
  %     1&\quad d'=\sum_{a\in\act}\sum_{q'\in Q}\sum_{q\in Q}\sigma(d)(a) d(q)\delta(q,a)(q')\ket{q'}\\
  %     0&\quad \text{otherwise}
  %   \end{cases}.$$

  % Let ${\cal M}=(Q,\act,\delta)$ be an MDP and $\mathbb{X}^{\msmt{}}\in \Dsched(\Dconf(\Dtrans(Q)))\to\Dchan(\Dmass(Q))$ be the CM classifier defined in \cref{def:CMclassMSMT}. For any scheduler $\sigma$, we let ${\cal M}^{\msmt{}}_{\sigma}=(\dist(Q)^{+},\delta_{\mathbb{X}^{\msmt{}},\sigma})$ denote the config MC induced by the MDP ${\cal M}$, the scheduler $\sigma$ and the CM classifier $\mathbb{X}^{\msmt{}}$. The transition function $\delta_{\mathbb{X}^{\msmt{}},\sigma}$ is evaluated as follows: for any $p=d_{0}d_{1}... d_{n}\in \dist(Q)^{+}$ and $d_{n+1}\in\dist(Q)$,
  % $$\delta_{\mathbb{X}^{\msmt{}},\sigma}(p)(p\cdot d_{n+1})=\begin{cases}
  %   1&\quad d_{n+1}=\sum_{a\in\act}\sigma_{a}\sum_{q'\in Q}\sum_{q\in Q}d_{n}(q)\delta(q,a)(q')\ket{q'}\\
  %   0&\quad \text{otherwise}
  % \end{cases}.$$
  % For the case $\sigma$ being memoryless is similarly defined.
\end{definition}

% \begin{problem}[threshold \msmt{} reachability problem]\label{prob:ThresholdMSMT}
%   Assuming the setting of \cref{def:ConfMCMSMT} and let $d_{0}\in\dist(Q)^{+}$ be an initial configuration, $H\subset \dist(Q)$ be the target set and $\lambda\in [0,1]$ be the threshold. The \emph{threshold \msmt{}-reachability problem} asks if there exists a scheduler $\sigma$ such that 
%   \begin{math}
%     \Pb_{\mathbb{X}^{\msmt{}},\sigma}\bigl(\reach_{\mathbb{X}^{\msmt{}},\sigma}(d_{0},H)\bigr)\geq \lambda
% \end{math}.
% \end{problem}

\end{auxproof}

\begin{remark}[combinations other than the four]
 We have presented four instantiations of the CM classifier framework, corresponding
% to the four semantics in
 \cref{table:fourVarIntro}. Combinatorially, there are $2^3=8$ possible ways of assigning $\Dsched$, $\Dconf$ and $\Dtrans$ to $\Dchan$ and $\Dmass$, but we have only four since we always assign $\Dconf$ to $\Dmass$. The other options (where $\Dconf$ is assigned to $\Dchan$) are not considered as they do not conceptually fit the general framework in \cref{sec:unifiedFramework}, where a `state' of a config MC induced by an MDP should be a distribution over the state space of the MDP. 

 %in config MCs, configurations accommodate masses while probabilistic branching accommodates the chance aspects.

 % We have presented the instantiation of the CM classifier framework under the four semantics in \cref{table:fourVarIntro}. From the combinatorial perspective, there are $2^3=8$ possible ways of assigning $\Dsched$, $\Dconf$ and $\Dtrans$ to $\Dchan$ and $\Dmass$. In the four semantics we presented, $\Dconf$ is always assigned to $\Dmass$. While mathematically the four other CM classifiers can be defined through the three combinatorial operators introduced in~\cref{def:CMClassifier}, they do not result in a meaningful semantics in the context of MDP and we leave their further study to future work.

 % and the corresponding threshold reachability problems. 
 %%%%%%%%%%%%%%%%%%%%%%%%%%%%%%%%%%%%%%%%%%%%%%%%%%%%%%%%%%%%%%%%%%%%%%
\end{remark}

\section{Algorithms and Complexity}\label{sec:Algo}
Now we present several complexity results and two algorithms for \cref{prob:GeneralThreshold}, with  $\mathbb{X}^{\csmt{}}$ and  $\mathbb{X}^{\msct{}}$ as  CM classifiers (the other two have been studied before). Due to space constraints, all proofs are in Appendix ~\ref{sec:proof}. For completeness, we start by recalling known results when $\mathbb{X}^{\csct{}}$ and  $\mathbb{X}^{\msmt{}}$ are provided as the input for the CM classifier. 
%To facilitate the discussion, 
In this section,
we refer to \cref{prob:GeneralThreshold} as the \emph{threshold $\mathbf{S}$-reachability problem} with the CM classifier provided as $\mathbb{X}^{\mathbf{S}}$, where $\mathbf{S}$ can be $\csct{}, \msct{}, \csmt{}$ or $\msmt{}$.

\myparagraph{Known results on the \csct{} and the \msmt{} semantics}\label{subsec:KnownAlgoResult}
The \csct{} semantics models the conventional MDP semantics
(\cref{rem:CSCTConv}).
% when it is restricted to the Dirac configuration (proof in Appendix \ref{sec:proof}). 
This implies that the threshold $\csct{}$-reachability problem is in {\sf Ptime} when the initial configuration $d_{0}$ is Dirac. The proof is a simple adaption of those for the conventional semantics and can be found in~\cite{Baier2008}. 

\begin{auxproof}
 For the \msmt{} semantics, it resembles the distribution transformer framework~\cite{DBLP:journals/tse/KwonA11,DBLP:conf/qest/KorthikantiVAK10,DBLP:conf/qest/ChadhaKVAK11,DBLP:conf/birthday/AghamovBKNOPV25,LMS14,DBLP:journals/jacm/AgrawalAGT15,DBLP:conf/lics/AkshayGV18,Akshay2023,Akshay2024}. %introduced in a series of recent works~\cite{DBLP:journals/tse/KwonA11,DBLP:conf/qest/KorthikantiVAK10,DBLP:conf/qest/ChadhaKVAK11,DBLP:conf/birthday/AghamovBKNOPV25,LMS14,DBLP:journals/jacm/AgrawalAGT15,DBLP:conf/lics/AkshayGV18,Akshay2023,Akshay2024}. 
 In~\cite{Akshay2015}, it is shown that a variant of~\cref{prob:GeneralThreshold}, which asks if there is a scheduler of type $\sigma\colon \dist(Q)\times Q\to \dist(\act)$ such that the target set $H$ is reachable, is shown to be positivity-hard. In this work, we focus on schedulers of type $\sigma\colon \dist(Q)\to\dist(\act)$, which is a more general type of scheduler as a scheduler of type $\dist(Q)\times Q\to\dist(\act)$ can be expressed as a scheduler of type $\dist(Q)\to \dist(\act^{Q})$ (proof in Appendix \ref{sec:proof}). Hence, this allows us to conclude that the threshold \msmt{}-reachability problem is {\sf Positivity-hard}. The proof is by adapting the proof from [4] via a simple modification to the global scheduler.
\end{auxproof}

  The \msmt{} semantics is about distribution transformers in~\cite{DBLP:journals/tse/KwonA11,DBLP:conf/qest/KorthikantiVAK10,DBLP:conf/qest/ChadhaKVAK11,DBLP:conf/birthday/AghamovBKNOPV25,LMS14,DBLP:journals/jacm/AgrawalAGT15,DBLP:conf/lics/AkshayGV18,Akshay2023,Akshay2024}. %introduced in a series of recent works~\cite{DBLP:journals/tse/KwonA11,DBLP:conf/qest/KorthikantiVAK10,DBLP:conf/qest/ChadhaKVAK11,DBLP:conf/birthday/AghamovBKNOPV25,LMS14,DBLP:journals/jacm/AgrawalAGT15,DBLP:conf/lics/AkshayGV18,Akshay2023,Akshay2024}. 
There, the threshold reachability  is shown to be positivity-hard (see~\cite{Akshay2015}), implying that the threshold \msmt{}-reachability is positivity-hard.
(The result in~\cite{Akshay2015} uses local actions but it is easily translated to our formalization 
%of \cref{prob:GeneralThreshold} 
with global actions. See \cref{rem:globalVsLocalActions}.)

\subsection{Complexity and Algorithm for the \msct{} Semantics}\label{subsec:msctAlgo}
%\input{msct}

%We present two complexity results and a template-based synthesis algorithm for the threshold \msct{}-reachability problem.

%\myparagraph{Complexity Results}
We first a lower bound for the threshold \msct{}-reachability. It is proven by a reduction from the counting variant of the subset sum problem, which is known to be {\sf $\sharp$P-complete}~\cite{book-papa}.

\vspace{.2em}
\noindent
\begin{minipage}{\textwidth}
  \begin{restatable}{theorem}{MSCTSharpP}\label{thm:msct-sharpP-hard}
    The threshold \msct{}-reachability problem is {\sf $\sharp$P-hard}.
  \end{restatable}
  % \begin{theorem}\label{thm:msct-sharpP-hard}
  %   The threshold \msct{}-reachability problem is {\sf $\sharp$P-hard}.
  % \end{theorem}
\end{minipage}

%\vspace{.2em}
% We show, however, that given an explicit scheduler, the problem becomes positivity-hard. Note that unlike in conventional or distributional semantics, here the existence of the scheduler does not imply that this problem is easier than the one above. The proof is done by reduction from the reachability problem for Markov chains under the \msmt{} semantics.

We show, however, that given an explicit scheduler, the problem becomes positivity-hard.  Unlike in conventional \csct{} or distributional \msmt{} semantics, here fixing a scheduler does not ease the problem.  \cref{thm:msct-wScheduler-pos-hard} is proved by reduction from the reachability problem for Markov chains under the \msmt{} semantics (shown positivity-hard in \cite{Akshay2015}).

%\vspace{.2em}
\noindent
\begin{minipage}{\textwidth}
  \begin{problem}[threshold \msct{}-reachability with explicit scheduler]\label{prob:ThresholdMSCT-wScheduler}
  Given an MDP ${\cal M}=(Q,\act,\delta)$, an initial configuration $d_{0}\in\dist(Q)$, a scheduler $\sigma$, a set of configuration $H\subseteq\dist(Q)$ and a threshold $\xi\in [0,1]$, decide if
  $\Pb_{\mathbb{X}^{\msct{}},\sigma}(\reach_{\mathbb{X}^{\msct{}},\sigma}(d_{0},H))\geq\xi$.
  \end{problem}
\end{minipage}

%\vspace{.2em}
\noindent
\begin{minipage}{\textwidth}
  \begin{restatable}[positivity-hardness when scheduler given]{theorem}{MSCTPosHard}\label{thm:msct-wScheduler-pos-hard}
    \cref{prob:ThresholdMSCT-wScheduler} is \text{Positivity-hard}.
  \end{restatable}
  % \begin{theorem}[Positivity-hardness when scheduler given]\label{thm:msct-wScheduler-pos-hard}
  %   \cref{prob:ThresholdMSCT-wScheduler} is \text{Positivity-hard}.
  % \end{theorem}
\end{minipage}

\begin{remark}
  The proof of the above results are presented in Appendix~\ref{sec:proof}.
  Both of the proof does not require the scheduler to be memoryful, hence the hardness results holds even when the problem is restricted to memoryless schedulers.
\end{remark}

\paragraph*{A Template-based Synthesis Algorithm for Finitely-Generated Monotone Set}

Now we present a template-based synthesis algorithm for solving the threshold \msct{}-reachability problem,  under the assumption that the target set is provided as a finitely-generated monotone set.% generated by finitely many sub-distributions. 
%The major motivation for adopting template-based synthesis approach comes from the complexity of the problem. 

We saw that the threshold reachability problem is hard (\cref{thm:msct-sharpP-hard,thm:msct-wScheduler-pos-hard})---this is because  MDPs under the \msct{} semantics can branch into exponentially many  configurations which are expensive to track. Therefore, instead of tracking reachable config's, we
%turn our goal to
directly search for a reachability certificate by template-based synthesis and constraint solving.

% Adoption of the template-based synthesis approach comes from the complexity of the problem (\cref{thm:msct-sharpP-hard,thm:msct-wScheduler-pos-hard})---the problem is hard since MDPs under the \msct{} semantics can branche into exponentially many  configurations and they are expensive to track explicitly.
% % The MDP under this semantics branches into exponentially many possible configurations in just one step, it is therefore unclear that how can the reachability be tracked efficiently, given the hardness results in \cref{thm:msct-sharpP-hard,thm:msct-wScheduler-pos-hard}.
% %
% %classical algorithms for solving reachability can be applied in this framework without having the troubles of the infinite state space
% Hence, we 
% %turn our goal to
% instead search for a certificate by template-based synthesis.
% %and synthesis it whenever it is possible.
 
Our result is an extension of the template-based algorithm in \cite{Akshay2024} with \emph{$\gamma$-scaled submartingales} introduced in \cite{Takisaka2021}. We also make several changes on the assumption of the target set, namely using upward- or downward-closed sets, which involves some technical novelties. We present a version with upward-closed sets; dealing with downward-closed sets is similar.

% For the brevity of presentation, we demonstrate the algorithm using only upward-closed sets; the extension to downward-closed sets is straightforward.  

%Concretely, \cref{def:CMclassMSCT,def:ConfMCMSCT} demonstrated that the set of all transitions in one step is characterized by the set of functions $Q^{Q\times \act}$, which is clearly exponential and thus making it impractical to apply any classical iteration-based algorithm. The complexity results proven in \cref{thm:msct-sharpP-hard,thm:msct-wScheduler-pos-hard} also suggested that this problem is intrinsically hard, hence 

%\todo{YC: Most of the technical details are in appendix though, we should address it?}

We start by defining the $\gamma$-scaled configuration submartingale, adapted from \cite{Takisaka2021}.
%The $\gamma$-scaled configuration submartingale, as an adaptation of the result in \cite{Takisaka2021}, is defined as follows.

\vspace{.3em}
\noindent
\begin{minipage}{\textwidth}
 \begin{definition}[$\gamma$-scaled configuration submartingale]\label{def:gamma-subm}
  Given an MDP ${\cal M}$, a memoryless scheduler $\sigma$, a target set $H$ and a real number $\gamma\in (0,1)$, a $\gamma$-scaled configuration submartingale is a function $R:\dist(Q)\to [0,1]$ such that the following holds:  for each $d\in\dist(Q)\setminus H$,
  \begin{equation}\label{eqt:gamma-subm}
    R(d)\;\leq\; \gamma\cdot \sum_{d'\in\dist(Q)}R(d')\cdot\delta_{\mathbb{X}^{\msct{}},\sigma}(d)(d').
  \end{equation}%\todo{Update the definition}
 \end{definition}
\end{minipage}
Note that, in~\cref{eqt:gamma-subm}, the term $\delta_{\mathbb{X}^{\msct{}},\sigma}(d)$ is of type $\dist(\dist(Q))$ which gives a distribution over successor configurations $d'$  in the config MC under the scheduler $\sigma$. Hence, evaluating it with $d'$ gives the probability of reaching $d'$ from $d$.
% under the scheduler $\sigma$. 
The sum in~\cref{eqt:gamma-subm} can be understood as the expected value of the submartingale $R$ after one-step under the scheduler $\sigma$. Note that the sum is well-defined since all distributions in this paper are countably supported.
% It is safe to use summation here since we assumed all distributions in this paper are finitely-supported. 

\begin{auxproof}
 Furthermore, we shall see later under suitable assumption, \cref{eqt:gamma-subm} can be evaluated efficiently. 
\end{auxproof}

The next theorem is proved much like in~\cite{Takisaka2021}.

\begin{theorem}[$\gamma$-scaled submartingale witness]\label{thm:subm-witness}
  Assume that the input of \cref{prob:GeneralThreshold} is given, with $\mathbb{X}^{\msct{}}$ as the CM classifier.  Let $R$ be a $\gamma$-scaled submartingale $R$ for some $\gamma\in (0,1)$. Then
  \begin{math}
    R(d_{0})\geq\xi
  \end{math}
  implies there exists a memoryless scheduler $\sigma$ such that 
\begin{math}
\Pb(\reach_{\mathbb{X}^{\msct{}},\sigma}(d_{0},H))\geq\xi.
\end{math}
  % , then if there is a $\gamma$-scaled submartingale $R$ for some $\gamma\in (0,1)$ such that $R(d_{0})\geq\xi$, then
  % $\Pb(\reach_{\mathbb{X}^{\msct{}},\sigma}(d_{0},H))\geq\xi$.%
If $R(d_{0})\geq\xi$ holds, we say that $R$ is a witness of reachability.
\end{theorem}

As usual in  template-based synthesis, our algorithm collects constraints---ours are given as a system of inequalities over reals---and solves them. The construction is  three steps.
%Given an MDP ${\cal M}$, initial configuration $d_0$, a target set $H$ and parameters $\xi$ and $\gamma$, 
\begin{compactitem}
  \item In the first step, templates for a 
%configuration-based 
scheduler $\sigma$ and a $\gamma$-scaled configuration submartingale $R$ are constructed.
  \item In the second step, constraints are constructed to ensure that a feasible solution of the constraints corresponds to a valid scheduler and a valid $\gamma$-scaled submartingale that is a witness of reachability.
  \item In the final step, the whole constraints are passed to an off-the-shelf polynomial inequality solver to compute a feasible solution (we use Z3). If a feasible solution is found, the algorithm concludes the target set is reachable and return the submartingale as a witness.
\end{compactitem}

Below we briefly describe each step.

\myparagraph{Step 1: Setup of Template}
The first step is to set up templates for a scheduler and a  $\gamma$-scaled submartingale.  This step is standard and done similarly to other template-based synthesis algorithms; thus we only outline the templates here. For details, see~\cite{Akshay2024} and the references therein.
%The templates are defined similarly as in \cite{Akshay2024}, except a minor difference in the scheduler as discussed earlier in \cref{subsec:KnownAlgoResult}.
\begin{compactitem}
  \item \textbf{Scheduler}:  two sets of template variables $\{\theta^{a}_{q}\mid a\in\act, q\in Q\cup\{0\}\}$ and $\{s_{q}\mid q\in Q\cup\{0\}\}$ are fixed, such that the target scheduler
% $\sigma(d)(a)$ 
$\sigma$  is  as follows:
for each configuration $d\in\dist(Q)$ and action $a\in \act$,
        $\sigma(d)(a)=\frac{\theta^{a}_{0}+\sum_{q\in Q}\theta^{a}_{q}\cdot d(q)}{s_{0}+\sum_{q\in Q}s_{q}\cdot d(q)}.$
  \item \textbf{Submartingale}: a set of template variables $\{r_{q}\mid q\in Q\cup\{0\}\}$ is fixed so that the target submartingale $R\colon \dist(Q)\to [0,1]$ is  %the submartingale $R(d)$ for each configuration $d\in\dist(Q)$ is of the form,
        $R(d)=r_{0}+\sum_{q\in Q}r_{q}\cdot d(q)$.
\end{compactitem}

\myparagraph{Step 2: Constraint Collection}
%After the templates for the scheduler and the submartingale are fixed, 
We now describe how the constraints for the templates are constructed to enforce the validity and witness property of the templates.

%, the detail of setting up these constraints are presented in the appendix due to space constraint.

\myparagraph{Constraints for Schedulers}
To ensure that for any configuration $d\in\dist(Q)$, $\sigma(d)$ is a distribution over actions, 
%i.e. a distribution over actions, 
we require it to satisfy 1) for any action $a\in\act$, $\sigma(d)(a)\geq 0$ and 2) $\sum_{a\in\act}\sigma(d)(a)=1$. 
Concretely, using the template variables, 
% The concrete constraints expressed in terms of the template variables are given as follows.
% It suffices to make sure that $\sigma(d)\in\dist(\act)$, which translates to $\sigma(d)(a)\geq 0\forall a\in\act$ and $\sum_{a\in\act}\sigma(d)(a)=1$, expressing these in terms of the template variable translates to, $\forall d\in\mathbb{R}^{\abs{Q}}$,
\begin{equation}\label{eqt:constraint-scheduler}
  \Phi_{\sf Schedule}:\begin{cases}
    d\in\dist(Q) & \implies \bigwedge_{a\in \act}\left[\theta^{a}_{0}+\sum_{q\in Q}\theta^{a}_{q}\cdot d(q)\geq 0\right]             \\
    d\in\dist(Q) & \implies s_{0}+\sum_{q\in Q}s_{q}\cdot d(q)\geq 1                                                                 \\
    d\in\dist(Q) & \implies \sum_{a\in\act}\theta^{a}_{0}+\sum_{q\in Q}\theta^{a}_{q}\cdot d(q) = s_{0}+\sum_{q\in Q}s_{q}\cdot d(q).
  \end{cases}
\end{equation}

\myparagraph{Constraints for Submartingales}
 For submartingales, we have two constraints: 1) for any configuration $d\in \dist(Q)$, $R(d)\in [0,1]$ and 2) for any configuration $d\in\dist(Q)\setminus H$, $R(d)$ satisfies the submartingale condition in \cref{eqt:gamma-subm}. Concretely, using the template variables, 
%The concrete constraints expressed in terms of the template variables are given as follows:
\begin{align}
  &\Phi_{\sf Bound}:&& \textstyle d\in\dist(Q)\implies 0\leq r_{0}+\sum_{q\in Q}r_{q}\cdot d(q)\leq 1,\label{eqt:constraint-subm-bound}\\
  &\Phi_{\sf Inductive}:&& \textstyle d\in\dist(Q)\setminus H\implies r_{0}+\sum_{q\in Q}r_{q}\cdot d(q)\leq \gamma\sum_{a\in\act}\sigma(q)(a)R(M_a^{T}d).\label{eqt:constraint-subm-inductive}
\end{align}

% \begin{eqnarray}
%   &&\Phi_{\sf Bound}: d\in\dist(Q)\implies 0\leq r_{0}+\sum_{q\in Q}r_{q}\cdot d(q)\leq 1,\label{eqt:constraint-subm-bound}\\
%   &&\Phi_{\sf Inductive}: d\in\dist(Q)\setminus H\implies r_{0}+\sum_{q\in Q}r_{q}\cdot d(q)\leq \gamma\sum_{a\in\act}\sigma(q)(a)R(M_a^{T}d).\label{eqt:constraint-subm-inductive}
% \end{eqnarray}
% \begin{equation}\label{eqt:constraint-subm-bound}
%   \Phi_{\sf Bound}: d\in\dist(Q)\implies 0\leq r_{0}+\sum_{q\in Q}r_{q}\cdot d(q)\leq 1,
% \end{equation}\begin{equation}\label{eqt:constraint-subm-inductive}
%   \Phi_{\sf Inductive}: d\in\dist(Q)\setminus H\implies r_{0}+\sum_{q\in Q}r_{q}\cdot d(q)\leq \gamma\sum_{a\in\act}\sigma(q)(a)R(M_a^{T}d).
% \end{equation}
%mainly need to make sure its co-domain is bounded by $1$, i.e. $\forall d\in\dist(Q), R(d)\in[0,1]$, and satisfying equation \cref{eqt:gamma-subm}. First  the boundedness is simple, $\forall d\in\mathbb{R}^{\abs{Q}}$,

\myparagraph{Constraints for Reachability Certificates} We also impose a constraint so that 
%Finally, we impose one more constraint to ensure that the submartingale
 $R$ is a witness of reachability, i.e. $R(d_{0})\geq \xi$. 
%The concrete constraint expressed in terms of the template variables is given as follows: 
Concretely,
\begin{align}\label{eqt:subm-reach-cert}
&  \Phi_{\sf Reachable}: &&\textstyle R(d_{0})=r_{0}+\sum_{q\in Q}r_{q}d_{0}(q)\geq \xi.
\end{align}

\begin{auxproof}
  Then we proceed to handle the inductive rule in equation \cref{eqt:gamma-subm}, a technical difficulty here is that in general to express right-hand side of \cref{eqt:gamma-subm} one has to enumerate over all possible branches, which leads to an exponential time of translation. Here we exploit the linear structure of the submartingale,
  \begin{lemma}\label{lem:subm-linearity}
    Fix a MDP ${\cal M}$ and a memoryless scheduler $\sigma$, when $R$ is linear, $\forall a\in \act$,
    $$\int_{\dist(Q)}R\,d\mathbb{X}^{\msct{}}(d,\sigma) = \sum_{a\in\act}\sigma(d)(a) R(M_{a}^{T}d)$$
  \end{lemma}
  \begin{proof}
    It follows simply from $R$ being linear, thus by viewing the configuration as some random vector, then expectation of $R$ is equal to $R$ acting on the expected configuration, for detail algebra please refer to the next part of the appendix.
  \end{proof}
  %Thus by \cref{lem:subm-linearity}, we can easily express the constraint for satisfying equation \cref{eqt:gamma-subm} as following,
\end{auxproof}

\myparagraph{Step 2.5: Quantifier Elimination}
It is worth noting that all the constraints we obtained in the system are either a simple inequality or a quantified formula in the form of $\forall \vx\in\mathbb{R}^{n},\, \Phi(x)\geq 0\implies\psi(x)\geq 0$, where $\Phi$ is 
%a conjunction of linear functions
a linear function and $\psi$ is a polynomial. Therefore, the well-known Farkas's lemma~\cite{farkas1902theorie} and Handelman's theorem~\cite{handelman1988representing} can be applied to perform quantifier elimination, which reduce the problem to solving a system of polynomial inequalities over reals. Further details are presented in the appendix.

\begin{auxproof}
  Lastly we briefly describe the approach for performing quantifier elimination, we rely on both Farkas's lemma \cite{farkas1902theorie}, Handelman's theorem~\cite{handelman1988representing} and their extension with non-strict inequalities support \cite{amir23}. We first recall Farkas's lemma and Handelman's theorem.

  \begin{theorem}[Farkas's lemma]
    A system of $m$ linear inequalities $\Phi\vx\succeq 0$ entails a linear inequality $\phi(\vx)\geq 0$ if and only if there exists $\vy\in\mathbb{R}^{m}_{\geq 0}$ such that $\phi=\Phi\vy$, that is,
    $\forall\vx\in\mathbb{R}^{n}, \left(\Phi\vx\succeq 0\implies \phi(\vx)\geq 0\right)\iff\exists \vy\in\mathbb{R}^{m}_{\geq 0},\,\left(\phi=\vy^{T}\Phi\right)$.
  \end{theorem}

  \begin{theorem}[Handelman's theorem]
    Fix a system of $m$ linear inequalities $\Phi\vx\succeq 0$, let ${\sf Prod}(\Phi)=\{\prod_{k=1}^{K}\Psi_{k}\mid K\in \{0\}\cup[m],\Psi_{k}\in\Phi\}$, where the case $K=0$ corresponds to $1$ and the notation $\Psi_{k}\in\Phi$ is abused to represent $\Psi_{k}$ is one of the linear inequalities in $\Phi$. Then if the feasible set of $\Phi$ is non-empty and compact, then $\Phi\vx\succeq 0$ entails a polynomial inequality $\phi(\vx)> 0$ if and only if $\phi$ can be expressed as some convex combination of $\psi\in {\sf Prod}(\Phi)$, i.e. $\forall\vx\in\mathbb{R}^{n}, \left(\Phi\vx\succeq \mathbf{0}\implies \phi(\vx) > 0\right)\iff\exists y_{1},...,y_{t}\geq 0$ and $\psi_{1},...,\psi_{t}\in {\sf Prod}(\Phi)$ such that
    $\psi_{1},...,\psi_{t}\in {\sf Prod}(\Phi),\phi(\vx)=\sum_{i=1}^{t}y_{i}\psi_{i}(\vx)$.
  \end{theorem}

  Just as in \cite{Akshay2024}, we cannot directly apply Handelman's theorem since the polynomial inequality we constructed is non-strict, but the $\impliedby$ direction is still sound trivially.

  Then we note that for constraint in \cref{eqt:constraint-scheduler,eqt:constraint-subm-bound}, all implications can be expressed as linear inequalities, in particular $d\in\dist(Q)$ can be equivalently expressed as $\bigwedge_{q\in Q} d(q)\geq 0\land 1-\sum_{q\in Q}d(q)\geq 0\land \sum_{q\in Q}d(q)-1\geq 0$.
  Hence, we can safely apply Farkas lemma to transform them into existential theory over reals. Next we are left to handle equation \cref{eqt:constraint-subm-inductive}. For the left-hand side, since we assume $H$ is a finitely generated upward-closed set, hence there always exists some $\vx_{1},...,\vx_{n}\in [0,1]^{\abs{Q}}$ such that $H=\uparrow \{\vx_{1},...,\vx_{n}\}$. To express $d\in\dist(Q)\setminus H$, it follows from simple algebra that
  \begin{equation}\label{eqt:Antichain-Conjunction}\small
    \begin{array}{l}\textstyle
      d\in {\cal D}(Q)\setminus H \iff\bigwedge_{{\bf q}\in Q^{n}}\left[(d\in {\cal D}(Q))\land \bigwedge_{i=1}^{n}(d({\bf q}_{i})<x_{i}({\bf q}_{i}))\implies\phi\right].
    \end{array}
  \end{equation}

  Hence, we can first translate the constraint in \cref{eqt:constraint-subm-inductive} to conjunction of $\abs{Q}^{n}$ constraints and apply quantifier elimination respectively. While for the right-hand side, it follows from \cite{Akshay2024} that we has to multiply by $s_{0}+\sum_{q\in Q}s_{q}d(q)$ over the inequality, hence resulting in a quadratic equation over $d$. Finally, we apply Handelman's theorem to reduce each of the constraint to the existential theory of reals.
\end{auxproof}

\myparagraph{Step 3: Constraint Solving}
Once all constraints are translated to an existential problem over reals, the system is then passed to a polynomial inequality solver for finding a feasible solution. If a feasible solution is found, then we claim that the target set $H$ is reachable.
\begin{theorem}
  The algorithm is sound, with a runtime in {\sf 2EXPTIME}.
\end{theorem}%\todo{proof to appendix} Too short, put below instead.
The soundness proof follows directly from \cref{thm:subm-witness}. 
% The notion of relative completeness means that whenever a linear $\gamma$-scaled submartingale that witnesses the reachability exists, our algorithm is able to synthesize it, see~\cite{Akshay2023} for more details. 
For the time complexity, it takes exponential time (in the number of generators of $H$) to construct all the constraints, and since existential theory of reals is in {\sf PSPACE}, hence we obtain {\sf 2EXPTIME} in total.

\begin{auxproof}
  \begin{proof}
    Soundness follows from that the constraints guarantee $R$ is a witness of reachability for MDP ${\cal M}$, initial configuration $d_{0}$ and target set $H$, hence by theorem \ref{thm:subm-witness} our result is sound.
  
    For the runtime, it takes $\widetilde{O}(\abs{Q}^{n})$ to construct all the constraints, and since existential theory of reals is in {\sf PSPACE}, hence giving {\sf 2EXPTIME} in total.
  \end{proof}
\end{auxproof}

%%%%%%%%%%%%%%%%%%%%%%%%%%%%%%%%%%%%%%%%%%%%%%%%%%%%%%%%%%%%%%%%%%%%%%
\subsection{Complexity and Algorithm for the \csmt{} Semantics}\label{subsec:csmtAlgo}
\myparagraph{Complexity Results}
In the \csmt{} semantics, similarly to the conventional \csct{} semantics, 
%when a scheduler is chance-interpreted, 
one can replace a mixed scheduler with a pure scheduler without decreasing the reachability probability. Intuitively, this is achieved by inductively choosing the action with the largest probability of reaching the target set. %and the same trick can be applied in the \csmt{} semantics.
\begin{restatable}{lemma}{CSMTpureSuffice} \label{lem:csmt_pure_suffice}
  Assume the setting of \cref{prob:GeneralThreshold}. If there is a mixed scheduler $\sigma\colon\dist(Q)^{+}\to\dist(\act)$ such that $\Pb(\reach_{\mathbb{X}^{\csmt{}},\sigma}(d_{0},H))\geq \xi$, then there is a pure scheduler $\sigma'\colon\dist(Q)^{+}\to\act$ (interpreted as a scheduler that only consists of Dirac distributions over $\act$ in our framework), such that $\Pb(\reach_{\mathbb{X}^{\csmt{}},\sigma}(d_{0},H))\geq \xi$.
\end{restatable}
% \begin{lemma} \label{lem:csmt_pure_suffice}
%   Assume the setting of \cref{prob:GeneralThreshold}. If there is a mixed memoryless scheduler $\sigma\colon\dist(Q)^{+}\to\dist(\act)$ such that $\Pb(\reach_{\mathbb{X}^{\csmt{}},\sigma}(d_{0},H))\geq \xi$, then there is a pure scheduler $\sigma'\colon\dist(Q)\to\act$ (interpreted as a scheduler that only consists of Dirac distributions over $\act$ in our framework), such that $\Pb(\reach_{\mathbb{X}^{\csmt{}},\sigma}(d_{0},H))\geq \xi$. %(A proof is by induction on a path $p\in \cpath^{\omega}_{\mathbb{X}^{\sf CS,MT}}$.)\todo{Update the definition}

% \end{lemma}
\noindent
Therefore we can limit our search to pure (instead of mixed) schedulers. Furthermore we notice the following: in the concrete description of~\cref{eq:configMCConcreteCSMT}, if $\sigma$ is a pure scheduler, the probability amplitude $\sigma(d)(a)$ is either $0$ or $1$. All these reduce, in the case of the \csmt{} semantics, the original problem (\cref{prob:GeneralThreshold}) to the following \emph{deterministic} problem.

% By \cref{lem:csmt_pure_suffice}, 

% and noting that the probability amplitude $\sigma(d)(a)$ in~\cref{eq:configMCConcreteCSMT} depends only on the scheduler $\sigma$, 
%  % the probabilistic branching described in \cref{eq:configMCConcreteCSMT} depends only on the scheduler, 
% it suffices for us to consider only pure schedulers and therefore, the threshold \csmt{}-reachability problem can be equivalently expressed as the following deterministic \csmt{}-reachability problem.

%it is easy to see that the transition becomes deterministic when the scheduler is pure, that is reachability probability to any target set $H$ is either $0$ or $1$. 

%, it suffices for us to only consider pure scheduler when computing the threshold reachability problem. Further noting that from the concrete transition function given in \cref{def:ConfMCCSMT}, it is easy to see that the transition becomes deterministic when the scheduler is pure, that is reachability probability to any target set $H$ is either $0$ or $1$. Therefore, the threshold \csmt{}-reachability problem can be equivalently expressed as the following deterministic \csmt{}-reachability problem.

% Then with a pure scheduler the one-step semantics becomes
% \begin{equation}\label{eqt:one-step-csmt-pure}
%   \mathbb{X}^{\sf (CS,MT)}_{\sf pure}(d_{0}...d_{n},\sigma) = \Delta(M_{\sigma(d_{0}...d_{n})}^{T}d_{n})
% \end{equation}
% where $\Delta(d)$ is the Dirac distrubtion on $d$ over $\dist(\dist(Q))$. Hence \cref{prob:csmt-reach} is equivalent to the following deterministic problem.
\begin{problem}[deterministic \csmt{}-reachability]\label{prob:csmt-reach-deter}
In the setting of \cref{prob:GeneralThreshold}, for a finite sequence $\bar{a}=a_{1} a_{2} ... a_{n}\in\act^{*}$ of actions, let $d_{\bar{a},1},...d_{\bar{a},n}\in\dist(Q)^{*}$ be a sequence of configurations and $\sigma_{\bar{a}}$ be a pure scheduler such that for any $j\in [n]$, the following is satisfied: $\sigma_{\bar(a)}(d_{0}d_{\bar{a},1}...d_{\bar{a},j-1})(a_{j})=1$ and $\delta_{\mathbb{X}^{\csmt{}},\sigma_{\bar{a}}}(d_{0}d_{\bar{a},1}...d_{\bar{a},j-1})(d_{\bar{a},j})=1$. Note that the scheduler is potentially memoryful.

The \emph{deterministic \csmt{}-reachability problem} asks if there exists a finite sequence $\bar{a}=a_{1} a_{2} ... a_{n}\in\act^{*}$ of actions and a natural number $n\in\mathbb{N}$ such that $d_{\bar{a},n}\in H$. 
% and a finite sequence  $d_{1} ... d_{n}\in\dist(Q)^{*}$\todo{Add a remark about the relation to scheduler and memory} of configurations, such that for any $j\in [n]$,
% \begin{math}
%   \delta_{\mathbb{X}^{\csmt{}},
%  \sigma_{\bar{a}}   
% }(d_{j-1})(d_{j})=1\text{ and }d_{n}\in H,
% \end{math}
% where $\sigma_{\bar{a}}$ is a pure scheduler that satisfies $\sigma_{\bar{a}}(d_{j-1})(a_{j})=1$ for any $j\in [n]$.
% Given an MDP ${\cal M}=(Q,\act,\mat)$ and a target set $H$, decide if there exists a finite sequence  $a_1 a_2 a_3 ... a_n\in \act^{\ast}$ such that
% $d_0 \xrightarrow{a_1} d_1 \xrightarrow{a_1} d_2 \xrightarrow{a_2} \dots \xrightarrow{a_n} d_{n}\in H$,
% where $d\xrightarrow{a}d'\iff d' = M_{a}^{T}d$, and $M_{a}^{T}$ is the transpose of $M_{a}$.\todo{Update the definition}
\end{problem}
Now the following result is proved by reduction from the emptiness problem of probabilistic finite automaton (PFA). The latter is undecidable; see e.g.~\cite{GimbertO10}. Further, we also show that the problem remains undecidable even when the target set $H$ is finitely generated and monotone. Both results only hold for searching a memoryful scheduler.

\begin{restatable}[undecidability of \csmt{}-reachability]{theorem}{CSMTundec}\label{thm:csmt-reach-undec}
  %\cref{prob:csmt-reach-deter} 
The threshold \csmt{}-reachability 
is undecidable even when the target set $H$ is restricted to a finitely-generated monotone set.% generated by finitely many subdistributions.
\end{restatable}
% \begin{theorem}[undecidability of \csmt{}-reachability]\label{thm:csmt-reach-undec}
%   \cref{prob:csmt-reach-deter} are undecidable even when the target set $H$ is restricted to finitely-generated and monotone.
% \end{theorem}
% Noting that the emptiness problem of PFA remains undecidable even when it contains only one accepting state (see e.g.~\cite{rote2024probabilisticfiniteautomatonemptiness}), and that $e_{j}^{T}d \geq \xi\iff d\succeq \xi e_{j}$, hence the set of accepting configuration can be encoded as an upward-closed set generated by $\xi e_{j}$. As a direct consequence, we can prove the following.
% \begin{corollary}[undecidability of \csmt{}-reachability with monotone set]\label{cor:csmt-reach-undec-upC}
%   \cref{prob:csmt-reach-deter} remains undecidable when the target set $H$ is finitely generated and monotone.
% \end{corollary}

\paragraph*{An Antichain-based Algorithm for Finitely-generated Monotone Set}
%\todo{Shall we change ``upward-closed'' into ``monotone'' everywhere? }
Despite the complexity, here we attempt to provide an algorithm for solving \cref{prob:csmt-reach-deter} with a target set $H$ being finitely generated and upward-closed (similarly to \cref{subsec:msctAlgo}, the adaption to downward-closed sets is straightforward). We motivate the use of antichain---used e.g.\ in~\cite{antichain-raskin}---by a simple observation that all stochastic matrices are monotone with respect to the element-wise order. That is, let $M\in\stoc(n)$ be a stochastic matrix, then for any $\vx,\vy\in [0,1]^{n}$, 
\begin{math}
  \vx\preceq \vy
\end{math} implies
\begin{math}
 M^{T}\vx\preceq M^{T}\vy
\end{math}.
We come to the following simple lemma.
\begin{restatable}[pullback of upward-closed set]{lemma}{PullbackUpCset}\label{lem:pullback-upC}
  Let $H$ be an upward-closed set and $M\in\stoc(n)$ be a stochastic matrix. Then the pullback of $H$ with respect to $M$, denoted as $M^{\sharp} H$ and defined as
  $M^{\sharp} H=\{\vx\in[0,1]^{n}\mid M^{T}\vx\in H\}$,
  is upward closed.  
\end{restatable}
% \begin{lemma}[Pullback of upward-closed set]\label{lem:pullback-upC}
%   Let $H$ be an upward-closed set and $M\in\stoc(n)$ be a stochastic matrix, then the pullback of $H$ with respect to $M$, denoted as $M^{\sharp} H$ and defined as
%   $M^{\sharp} H:=\{\vx\in[0,1]^{n}\mid M^{T}\vx\in H\}$,
%   is upward closed.
% \end{lemma}
Hence inductively, the set of all initial configurations that are reachable to $H$ is also upward-closed, which motivates the usage of antichains since they are a canonical representation of upward-closed sets. Formally, we connect this observation to \cref{prob:csmt-reach-deter} with the following theorem, which characterizes the set of configurations reaching $H$ as the least fixed point of an operator that preserves upward-closeness.

%We derive  our antichain-based algorithm from general lattice-theoretic observations. We start with a standard general characterization of reachable sets as a suitable fixed point (\cref{prop:latticeReachable}). When we use it later,  the parameters will be  as follows:  $X$ is the configuration space, and $f$ maps a configuration to the set of successor configurations taken over all schedulers (which we can assume to be pure, see \cref{lem:csmt_pure_suffice}).

\begin{auxproof}
  \begin{proposition}%\label{prop:latticeReachable}
    Let $X$ be  a set,  $f\colon X\to 2^X$ be a function (modeling forward and nondeterministic \emph{dynamics}), and  $H\subset X$ be a target set.
    Consider the (complete) subset lattice $2^{X}$ ordered by inclusion $\subset$, and the operator
    \begin{math}
      f^{*}\colon 2^{X} \to 2^{X},
      f^{*}Q = T\cup\{x\mid Q\cap f(x)\neq\emptyset \}.
    \end{math}
    Intuitively, $x\in f^{*}(Q)$ means that $x\in T$ or $x$ is one-step reachable to $Q$.
    Then we have the following.
    \begin{enumerate}
      \item The operator $f^{*}$ is monotone. Moreover, it preserves supremums (i.e.\ unions): $f^{*}(\bigcup_{i}Q_{i})=\bigcup_{i}f^{*}(Q_i)$.
      \item The repeated application of $f^{*}$ to the minimum $\bot=\emptyset$ gives rise to a chain
            \begin{math}
              \bot
              \;\subset\;
              f^{*}(\bot)
              \;\subset\;
              (f^{*})^{2}(\bot)
              \;\subset\;
              \cdots.
            \end{math}
      \item The last chain
            stabilizes at $\omega$, meaning that
            \begin{math}
              f^{*}\bigl((f^{*})^{\omega}(\bot)\bigr)
              = (f^{*})^{\omega}(\bot),
            \end{math}, where
            \begin{math}
              (f^{*})^{\omega}(\bot)
              =
              \textstyle\bigcup_{n\in\mathbb{N}}(f^{*})^{n}(\bot).
            \end{math}
            Moreover, $(f^{*})^{\omega}(\bot)$ is the least fixed point
            $\mu f^{*}$
            of $f^{*}$.
      \item The least fixed point $\mu f^{*}$ coincides with the reachable sets, that is,
            \begin{math}
              \mu f^{*}
              =
              \{x\mid \text{$x$ has an $f$-path  to $T$}\}.
            \end{math}
            Here an $f$-path is defined to be a sequence $x_{0}x_{1}\dotsc x_{n}$ where $x_{i+1}\in f(x_{i})$ for each $i\in [0,n-1]$.
  
    \end{enumerate}
  \end{proposition}
\end{auxproof}

\begin{restatable}[reachable configuration as a least fixed point]{theorem}{ReachConfLfp}\label{thm:ReachConfiglfp}
  Given an MDP ${\cal M}=(Q,\act,\mat)$, an initial configuration $d_{0}\in\dist(Q)$ and a finitely-generated upward-closed set $H$, let $B:2^{[0,1]^{Q}}\to 2^{[0,1]^{Q}}$ be a function defined by
  \begin{math}
    B(S)= H\cup\bigcup_{M\in\mat}M^{\sharp} S.
  \end{math}
  %Let ${\cal M}$ be  a MDP, $d_{0}$ be an initial configuration, $H$ be a finitely-generated upward-closed target set,  and $B$ be as in \cref{def:BPop}.
  Then the sequence $\bigl(B^{j}(\bot)\bigr)_{j\in \mathbb{N}}$ converges to its least fixed point $\mu B$ and $H$ is reachable from $d_{0}$ if and only if $d_{0}\in \mu B$.
\end{restatable}

% \begin{theorem}[Reachable Configuration as a Least Fixed Point]\label{thm:ReachConfiglfp}
%   Given an MDP ${\cal M}=(Q,\act,\mat)$, an initial configuration $d_{0}\in\dist(Q)$ and a finitely-generated upward-closed set $H$. Let $B:2^{[0,1]^{Q}}\to 2^{[0,1]^{Q}}$ be a function defined by
%   \begin{math}
%     B(S):= H\cup\bigcup_{M\in\mat}M^{\sharp} S.
%   \end{math}
%   %Let ${\cal M}$ be  a MDP, $d_{0}$ be an initial configuration, $H$ be a finitely-generated upward-closed target set,  and $B$ be as in \cref{def:BPop}.
%   The sequence $\bigl(B^{j}(\bot)\bigr)_{j}$ converges to its least fixed point $\mu B$ and $H$ is reachable if and only if $d_{0}\in \mu B$.
% \end{theorem}%\todo{proof to appendix} DONE

The \csmt{}-reachability problem is undecidable in general (\cref{thm:csmt-reach-undec}); thus computing the least fixed point of $B$ is not always possible. Nevertheless, we can approximate it and provide a sound but necessarily incomplete algorithm. Since $\bot$ is by definition upward closed, it follows that all approximants $B^{k}(\bot)$ are upward closed. Hence algorithmically, it suffices for us to keep track of the minimal elements of each $B^{k}(\bot)$, which can be maintained efficiently through the antichain data structure. 

The only challenge that remains is to compute the bottom (the subset of $\preceq$-minimal elements) of an upward-closed set.
 %(we let $\floor{X}$ denote the set of minimal elements of $X$). 
% $\floor{M^{\sharp} S}$, 
Specifically, in the current setting, we need to compute
\begin{math}%\label{eqt:BottomSet}
\floor{M^{\sharp} S}=
  \floor{\{\vy\in[0,1]^{\abs{Q}}\mid M^{T}\vy\succeq\vx
%  \text{ and } 
  \text{ for some } 
  \vx\in S
 \}}
\end{math}
for $M\in\mat$. Since we assume an antichain structure is always maintained in the algorithm, this problem can be reduced to computing the set $\floor{\{\vy\in[0,1]^{\abs{Q}}\mid M^{T}\vy\succeq\vx \}}$ for each $\vx\in \floor{S}$ and merging them.
% and since we assume we are maintaining an antichain structure, 
% the problem becomes as follows: for all $ \vx\in S, M\in\mat$, compute
% \begin{math}%\label{eqt:BottomSet}
%   \floor{\{\vy\in[0,1]^{\abs{Q}}\mid M\vy\succeq\vx \}}.
% \end{math}

%Now, 
% fixing $\vx\in\floor{S}$, our algorithm approximates the set $\floor{\{\vy\in[0,1]^{\abs{Q}}\mid M^{T}\vy\succeq\vx \}}$ (wrt.\ set inclusion) 
Our algorithm solves the last problem (for each $\vx\in\floor{S}$)
% under 
% two  hyper-parameters $K,L$ and
by iteratively solving the following sequence of minimization problems, increasing $k$ over time:
%for each $k\in [K]$, 
we randomly sample  $({\bf q}_{j})_{j\in [k-1]}\in [\,\abs{Q}\,]^{k-1}$ (recall the notation $[n]$ from \cref{sec:prelim}),  %\todo{this sampling part I don't really understand}
and solve
\begin{equation}\label{eqt:anti-min}
  \begin{aligned}
    \vy_{k}^{\ast}=\argmin_{\vy} \mathbf{1}^{T}\vy
    \quad\text{subject to }\quad &
    M^{T}\vy\succeq\vx, \quad  \mathbf{1}^{T}\vy\leq 1,                           \\[-.5em]
                                 & \vy\succeq\mathbf{0}, \quad\text{and}\quad
    \forall j<k,\,(\vy)_{{\bf q}_{j}}<(\vy^{*}_{j})_{{\bf q}_{j}}.
  \end{aligned}
\end{equation}
The intuition is as follows. Our goal is to iteratively obtain elements $\vy^{*}_{1},\vy^{*}_{2},\dotsc$ of the set 
$S=\floor{\{\vy\in[0,1]^{\abs{Q}}\mid M^{T}\vy\succeq\vx \}}$. 
%$\{\vy\in[0,1]^{\abs{Q}}\mid M^{T}\vy\succeq\vx \}$, aiming at 
For the first element $\vy^{*}_{1}$, the last constraint in \cref{eqt:anti-min} is vacuous, and we obtain an element of $S$ that is minimal wrt.\ $\preceq$ (here it is important that $\vy\preceq\vy'$ implies $\mathbf{1}^{T}\vy\le\mathbf{1}^{T}\vy'$). For the step case 
% of obtaining $\vy^{*}_{k}$ 
with $k>1$, the last  constraint in \cref{eqt:anti-min} enforces that $\vy^{*}_{j}\not\preceq\vy^{*}_{k}$ for each $j\in [k-1]$, where $\mathbf{q}_{j}\in |Q|$ is the ``guessed'' index that witnesses $\vy^{*}_{j}\not\preceq\vy^{*}_{k}$. That $\vy^{*}_{k}\preceq\vy^{*}_{j}$ is prohibited, too, for each $j<k$, since otherwise $\vy^{*}_{k}$ must have been obtained in the $j$-th round instead of  $\vy^{*}_{j}$.

This problem is an instance of linear programming which can be solved by any off-the-shelf LP solver in ${\sf PTime}$. If the solver does not return a solution, another ${\bf q}$ will be sampled and used to solve for $\vy_{k}^{\ast}$. The process terminates when either $K$ solutions $\vy^{*}_{1},\dotsc, \vy^{*}_{K}$ are found, or the solver fails to find  $\vy_{k}^{\ast}$ for some $k$ after $L$ different samples of $({\bf q}_{j})_{j}$'s ($K$ and $L$ are hyperparameters). The set of solutions $\{\vy_{k}^{\ast}\mid k\in [K]\}$ is then used to approximate the set $\floor{\{\vy\in[0,1]^{\abs{Q}}\mid M^{T}\vy\succeq\vx \}}$.
% of minimal elements.
% of $\{\vy\in[0,1]^{\abs{Q}}\mid M^{T}\vy\succeq\vx \}$.

\begin{restatable}{lemma}{ApproxLowerSet}\label{lem:anti-approx}
  Given a sub-distribution $\vx\in [0,1]^{n}$, a stochastic matrix $M\in\stoc(n)$, and parameters $K,L\in\mathbb{N}$, let ${\cal Y}_{K,L}(\vx)$ denotes the set of solutions obtained by the above described procedure labelled by parameters $K$ and $L$, then ${\cal Y}_{K,L}(\vx)\subseteq \floor{\{\vy\in[0,1]^{\abs{Q}}\mid M^{T}\vy\succeq\vx \}}$. %
%   consider the above procedure for computing $\vy_{k}^{\ast}$. 
% Then we have, for any $ k\in [K]$,
%   $\vy_{k}^{\ast}\in  \floor{\{\vy\in[0,1]^{\abs{Q}}\mid M\vy\succeq\vx \}}.$
  Moreover,  if the set $\floor{\{\vy\in[0,1]^{\abs{Q}}\mid M\vy\succeq\vx \}}$ is finitely generated, then we have \linebreak${\cal Y}_{K,L}(\vx)=\floor{\{\vy\in[0,1]^{\abs{Q}}\mid M\vy\succeq\vx \}}$ almost surely for large enough $K$ and $L$.
\end{restatable}
% \begin{lemma}\label{lem:anti-approx}
%   Given a sub-distribution $\vx\in [0,1]^{n}$, a stochastic matrix $M\in\stoc(n)$, and natural numbers $K$ and $L$, we have, for any $ k\in [K]$,
%   $\vy_{k}^{\ast}\in  \floor{\{\vy\in[0,1]^{\abs{Q}}\mid M\vy\succeq\vx \}}.$
%   Moreover,  if $\floor{\{\vy\in[0,1]^{\abs{Q}}\mid M\vy\succeq\vx \}}$ is finitely generated, then there is a $K$ and $L$ such that
%   $\{\vy^{\ast}_{k}\mid k\in [K]\}=\floor{\{\vy\in[0,1]^{\abs{Q}}\mid M\vy\succeq\vx \}}$.
% \end{lemma}%\todo{proof to appendix} DONE
\begin{auxproof}
  \begin{proof}
    For the first part, assume not, then there exists  $\vy_{k}'\neq \vy_{k}^{\ast}\in [0,1]^{n}$ such that $\vy_{k}'\preceq \vy^{\ast}_{k}$. Then it is not hard to check that $\vy_{k}'$ is also a feasible solution to the $k$-th minimization problem and $\mathbf{1}^{T}\vy_{k}'<\mathbf{1}^{T}\vy_{k}^{\ast}$. This contradicts with $\vy^{\ast}_{k}$ being the minimum solution. The second statement is easy: just pick $K$ to be the number of generators and $L=\abs{Q}^{K}$.
  \end{proof}
\end{auxproof}

\begin{algorithm}[H]\footnotesize
  \caption{Backward projection algorithm}\label{algo:anti}
  \DontPrintSemicolon
  \SetKwRepeat{Do}{do}{while}
  \KwIn{an MDP ${\cal M}=(Q,\act,\mat)$, an initial configuration $d_{0}\in\dist(Q)$, a target set antichain $\floor{H}\subset [0,1]^{\abs{Q}}$, parameters $K,L\in\mathbb{N}$}
  \KwOut{{\sf True} if $d_{0}$ is reachable to $H$}
  $S\leftarrow \{\}$\\
  $S'\leftarrow \floor{H}$\\
  \Do{$S\neq S'$}{
    $S\leftarrow S'$\\
    $S'\leftarrow \floor{H}$\\
    \For{$M\in\mat$}{
      \For{$\vx\in S$}{
        $\{\vy_{k}^{\ast}\mid k\in [K]\}\leftarrow$ Solve the problem in \cref{eqt:anti-min} against $M$,$\vx$ with parameters $K,L$\label{line:iter}\\
        $S'\leftarrow \floor{S'\cup\{\vy_{k}^{\ast}\mid k\in [K]\}}$\label{line:floor}
      }
    }
    \lIf{$d_{0}\in S'$}{\Return{{\sf True}}}
  }
  \lIf{$d_{0}\in S'$}{\Return{{\sf True}}}
\end{algorithm}

%Now we are ready to present the algorithm. It is in \cref{algo:anti}.
Finally, we formalize the algorithm and its correctness. The algorithm takes as input an MDP ${\cal M}=(Q,\act,\mat)$, an initial configuration $d_{0}\in\dist(Q)$, a target set $H$ provided as $\floor{H}$ and two parameters $K,L\in\mathbb{N}$, and returns {\sf True} if $d_{0}$ is reachable to $H$. It starts with a set $S=\floor{H}$ and iteratively updates $S$ by computing $\floor{B(S)}$ through solving the optimization problem described above. The algorithm only terminates when it witnesses a reachability or the set of configurations reachable to $H$ stabilizes, that is, $d_{0}\in S$ or $S=S'$. The full pseudocode is provided in~\cref{algo:anti}. %\cref{algo:anti} (in Appendix~\ref{sec:omittedDef}). 
The termination is necessarily not guaranteed since the problem is undecidable (\cref{thm:csmt-reach-undec}). But from \cref{lem:anti-approx} we can conclude that the algorithm always computes an under-approximation of the set of configurations reaching $H$, hence as a direct consequence of \cref{thm:ReachConfiglfp}, we have the following theorem.
%we always maintain an under approximation of the set of configurations reachable to H in the algorithm, as a direct consequence of \cref{thm:ReachConfiglfp,lem:anti-approx}, we have the following theorem.
% \begin{algorithm}[tb]\footnotesize
%   \caption{Backward Projection Algorithm}\label{algo:anti}
%   \DontPrintSemicolon
%   \SetKwRepeat{Do}{do}{while}
%   \KwIn{an MDP ${\cal M}=(Q,\act,\mat)$, an inital configuration $d_{0}\in\dist(Q)$, a target set antichain $\floor{H}\subset [0,1]^{\abs{Q}}$, parameters $K,L\in\mathbb{N}$}
%   \KwOut{{\sf True} if $d_{0}$ is reachable to $H$}
%   $S\leftarrow \{\}$\\
%   $S'\leftarrow \floor{H}$\\
%   \Do{$S\neq S'$}{
%     $S\leftarrow S'$\\
%     $S'\leftarrow \floor{H}$\\
%     \For{$M\in\mat$}{
%       \For{$\vx\in S$}{
%         $\{\vy_{k}^{\ast}\mid k\in [K]\}\leftarrow$ Solve problem in \cref{eqt:anti-min} against $M$,$\vx$ with parameter $K,L$\label{line:iter}\\
%         $S'\leftarrow \floor{S'\cup\{\vy_{k}^{\ast}\mid k\in [K]\}}$\label{line:floor}
%       }
%     }
%     \lIf{$d_{0}\in S'$}{\Return{{\sf True}}}
%   }
%   \lIf{$d_{0}\in S'$}{\Return{{\sf True}}}
% \end{algorithm}

\begin{theorem}\label{thm:antichainAlgoSound}
  \cref{algo:anti} is a sound algorithm for solving the \csmt{}-reachability problem with finitely-generated monotone target set $H$. 
\end{theorem}

%%%%%%%%%%%%%%%%%%%%%%%%%%%%%%%%%%%%%%%%%%%%%%%%%%%%%%%%%%%%%%%%%%%%%%
\section{Prototype Implementation and Experiments}\label{sec:exper}
While the primary focus of this paper is the theoretical development of the unifying framework and algorithms in the two new semantics, we also implemented prototypes of the algorithms in \cref{subsec:csmtAlgo,subsec:msctAlgo}, as a proof of concept. Details of the experiment setup and benchmarks are presented in Appendix \ref{sec:ExperAppendix}. We tested the implementations against various models from~\cite{Akshay2023,DBLP:conf/qest/ChadhaKVAK11} and also our motivating examples \emph{Casino} and \emph{Exam} from~\cref{sec:intro}.% and Both of the algorithms are able to solve these problems efficiently.

Despite the hardness of the problem, the antichain algorithm in \cref{subsec:csmtAlgo} on average solves a problem with three to five states within a few seconds. Most runs terminate before reaching a loop limit and prove the reachability.  The template-based algorithm in \cref{subsec:msctAlgo} on  average solves a problem with 3 to 5 states within the range of a few seconds to a hundred seconds. We were able to find a certificate that witnesses reachability in all our examples by considering only polynomials up to degree 4. Complete results are provided in Appendix \ref{sec:ExperAppendix} due to lack of space. 
In summary, our implementation indicates that the algorithms work on small models, at least, and have the potential to be further developed for real-world applications.

\section{Conclusion}\label{sec:conclu}
%should be not too long, just write here
In this paper, we presented a unified framework for defining semantics for MDPs that distinguishes randomness arising from schedulers from randomness in transitions. By fixing which get interpreted distributionally (\emph{mass}) and which probabilistically (\emph{chance}), we obtain four semantics of which two are new. We proved hardness results for the reachability problem with monotone target states for the two new semantics and developed sound algorithms using completely different techniques. As future work, it would be interesting to use the framework defined here, e.g., for obtaining the semantics for \emph{quantum} MCs and MDPs that have been recently considered in the verification community~\cite{DBLP:journals/ijautcomp/YingFY21, DBLP:conf/cav/GuanFTY24}.

%%
%% Bibliography
%%

%% Please use bibtex, 
%\newpage
\bibliography{ref}

\newpage
\appendix
\section{Examples and Illustrations on Unified Framework}\label{sec:UniFrameworkExample}
In this section, we provide more examples to demonstrate the unified framework concretely.
%We demonstrate the unified framework with a concrete example, we will take the MDP and scheduler defined in \cref{table:fourVarExampleIntro} for demonstration.

\subsection{CM classifier, Pre-configuration and Config MC}
We start by providing a concrete examples for each of the components defined in the construction of the unified framework. We will take the MDP and the scheduler defined in \cref{table:fourVarExampleIntro} as example, and we mainly focus on the \csct{} semantics for demonstrative purpose. Adopting it to other MDP and semantics should be simple.

\myparagraph{Pre-configuration and $\delta_{\sigma}$}
With direct substitution of equation \cref{eq:preConfig}, it is easy to see that the pre-configuration of the MDP in \cref{table:fourVarExampleIntro} is given by
\begin{align*}
  \delta_{\sigma}(d)&=0.4\underbrace{\biggket{d(q_{0})\underbrace{\bigket{0.5\ket{q_{1}}+0.5\ket{q_{2}}}}_{\delta(q_{0},a)} + d(q_{1})\underbrace{\bigket{\ket{q_{1}}}}_{\delta(q_{1},a)} + d(q_{1})\underbrace{\bigket{\ket{q_{2}}}}_{\delta(q_{2},a)}}}_{\text{transition induced by action }a}\\
  &\enspace+0.6\underbrace{\biggket{d(q_{0})\underbrace{\bigket{0.1\ket{q_{1}}+0.9\ket{q_{2}}}}_{\delta(q_{0},b)} + d(q_{1})\underbrace{\bigket{\ket{q_{1}}}}_{\delta(q_{1},b)} + d(q_{1})\underbrace{\bigket{\ket{q_{2}}}}_{\delta(q_{2},b)}}}_{\text{transition induced by action }b}.
\end{align*}

\myparagraph{CM classifier}
 Since we are composing the CM-classifier with the pre-configuration generated by $\delta_{\sigma}$, this forces the set $I$ and $J$ in \cref{def:CMclassifiersInstances} to be $\act$ and $Q$ respectively. Thus by expanding the definition, the CM-classifier $\mathbb{X}^{\csct{}}$ when applying on a pre-configuration $\zeta$ of the form $\zeta=\sum_{a\in\act}\rho_{a}\biggket{\sum_{q\in Q}d(q)\ket{\delta(q,a)}}$ gives the following expression:
\begin{displaymath}
  \mathbb{X}^{\csct{}}(\zeta)=\sum_{a\in\act}\sum_{f\in Q^{Q}}\rho_{a}\cdot\bigg(\prod_{q\in Q}\delta(q,a)(f(q))\bigg)\biggket{\sum_{q\in Q}d(q)\ket{f(q)}}.
\end{displaymath}

\myparagraph{Config MC}
%We further substitute the MDP and scheduler in \cref{table:fourVarExampleIntro} into the above expression, we have
Thus by combining the $\delta_{\sigma}$ and the CM classifier, we should get a config MC as defined in the framework. Below is the concrete transition functions of the corresponding config MC when applying to the example in \cref{table:fourVarExampleIntro} under the \csct{} semantics.

\begin{align*}
  \delta_{\mathbb{X}^{\csct{}},\sigma}(d)&=(\mathbb{X}^{\csct{}}\circ\delta_{\sigma})(d)\\
  &=\left(\underbrace{0.4\cdot\left(0.5\cdot 1\cdot 1\right)}_{\text{action $a$, }q_{0}\to q_{1}}+\underbrace{0.6\cdot \left(0.1\cdot 1\cdot 1\right)}_{\text{action $b$, }q_{0}\to q_{1}}\right)\biggket{(d(q_{0})+d(q_{1}))\ket{q_{1}}+d(q_{2})\ket{q_{2}}}\\
  &\enspace + \left(\underbrace{0.4\cdot\left(0.5\cdot 1\cdot 1\right)}_{\text{action $a$, }q_{0}\to q_{2}}+\underbrace{0.6\cdot \left(0.9\cdot 1\cdot 1\right)}_{\text{action $b$, }q_{0}\to q_{2}}\right)\biggket{d(q_{1})\ket{q_{1}}+(d(q_{0})+d(q_{2}))\ket{q_{2}}}\\
  &=0.26\biggket{(d(q_{0})+d(q_{1}))\ket{q_{1}}+d(q_{2})\ket{q_{2}}} + 0.74\biggket{d(q_{1})\ket{q_{1}}+(d(q_{0})+d(q_{2}))\ket{q_{2}}}.
\end{align*}

The final expression tells us that after one step of MDP under the \csct{} semantics, with $0.26$ probability the system will be in a state where the mass on $q_{1}$ is $d(q_{0})+d(q_{1})$ and the mass on $q_{2}$ is $d(q_{2})$, and with $0.74$ probability the system will be in a state where the mass on $q_{1}$ is $d(q_{1})$ and the mass on $q_{2}$ is $d(q_{0})+d(q_{2})$. This also match the illustration in \cref{table:fourVarExampleIntro} when we plug in $d(q_{0})=1, d(q_{1})=d(q_{2})=0$.
 
\subsection{The Three Operators $\lambda$, $\mu$, $\eta$}\label{subsec:ThreeOpAppen}
\begin{figure}[h]
  \centering
  \includegraphics[width=.9\textwidth]{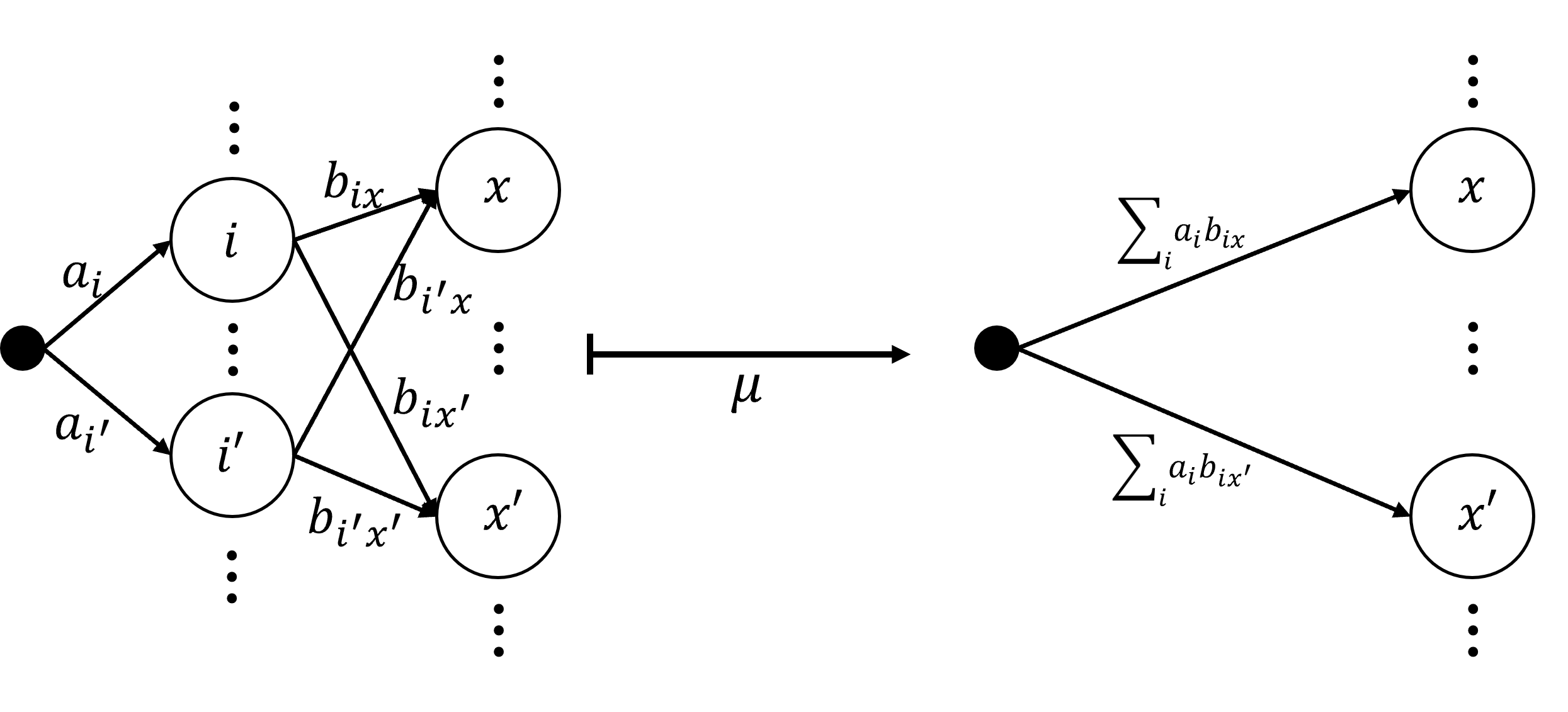}
  \caption{the suppression operator $\mu$ }
  \label{fig:suppression}

  \includegraphics[width=.9\textwidth]{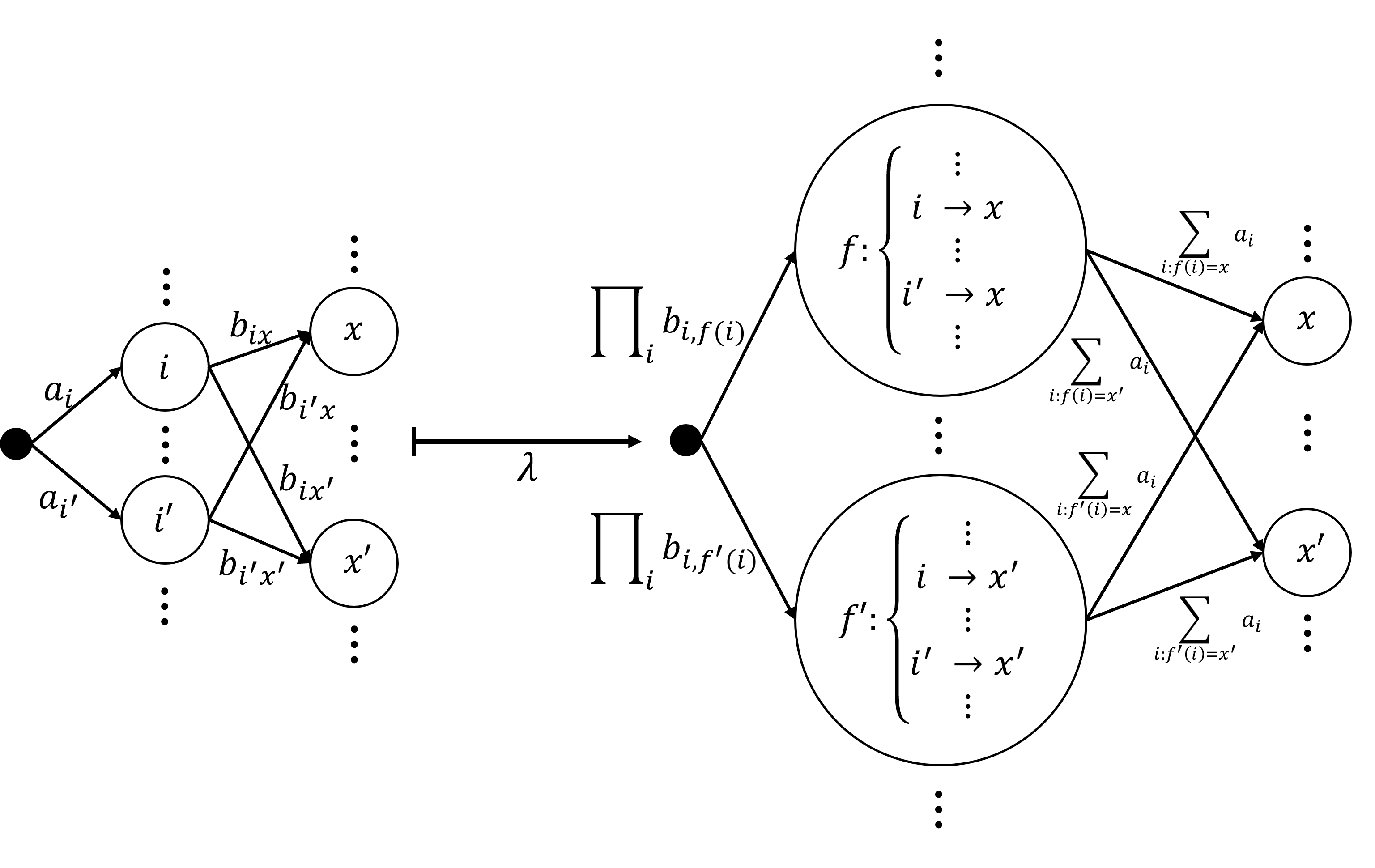}
  \caption{the MC2CM operator $\lambda$ }
  \label{fig:MC2CM}
\end{figure}

In this section we provide more intuition on the three operators. The graphical illustration for $\lambda$ and $\mu$ is provided in Appendix \ref{sec:UniFrameworkExample}. Below we present an analogy of coin tossing to explain the effect of the three operators.

\begin{example}[coin-tossing analogy]\label{ex:coinToss}%\todo{YC: It's probably fine but our coin can produce more than two outcomes in one toss. Should we mention it? Or perhaps calling it dice-toss instead?}
  Let ${\sf A}, {\sf B}$ be sets. We define two types of coin, the \emph{${\sf A}$-coin}, denoted as $C^{\alpha}$, produces an outcome $a\in {\sf A}$ upon tossing; and the \emph{${\sf B}$-coin}, denoted as $C^{\beta}$, produces an outcome $b\in {\sf B}$ upon tossing. Note that a $\ast$-coin generates a distribution over $\ast$, hence we abuse the notation and denote $C^{\alpha}\in\dist({\sf A})$ and $C^{\beta}\in\dist({\sf B})$. Now consider having an ${\sf A}$-coin and a collection of ${\sf B}$-coins, in the form of $(C^{\alpha},(C^{\beta}_{a})_{a\in {\sf A}})$, we construct a random process of tossing two coins one at a time as follows. The process first tosses the coin $C^{\alpha}$ and produces an outcome $a\in{\sf A}$, which naturally induces the second coin $C^{\beta}_{a}$; the coin $C^{\beta}_{a}$ is then tossed to decide a choice $b\in {\sf B}$.

  Let $X$ be a set and fix a function $f:{\sf A}\times {\sf B}\to X$, we denote $\dist_{\sf A}(\dist_{\sf B}(f,X))$ the set of distributions induced by $f$, that is for any distribution $d\in\dist_{\sf A}(\dist_{\sf B}(f,X))$, $d$ is in the form of $\sum_{a\in {\sf A}}p_{a}\biggket{\sum_{b\in {\sf B}}p_{a,b}\ket{f(a,b)}}$. For any such $d$, it induces a unique combination of coins $(C^{d,\alpha},(C^{d,\beta}_{a})_{a\in {\sf A}})$ so that the distribution $d$ is the distribution over the outcomes $f(a,b)$ induced by the process under the mapping $f$.
\end{example}

\myparagraph{Suppression $\mu$ and Dirac $\eta$}
The suppression operator $\mu$ is best described as ``suppression of two successive probabilistic branching into one'', as shown in \cref{fig:suppression}. To understand the suppression, consider the coin-tossing analogy in~\cref{ex:coinToss}. Instead of throwing the coin one by one, one can throw all the coins at once and decide both of the choices simultaneously. Intuitively, a certain pair of outcome $(a,b)\in{\sf A\times B}$ will be produced with probabilities $p_{a}p_{a,b}$, which is the same as the illustration in \cref{fig:suppression}. The mathematical essence of this ``throwing all coins at once'' is exactly captured by the suppression operator defined in \cref{def:mu}.

The Dirac operator $\eta$ turns an element $x$ into the (trivial) Dirac distribution at $x$, which simply corresponds to creating a coin that always gives the same outcome, or a coin with two heads practically, when thinking in the coin-tossing analogy (\cref{ex:coinToss}). Categorically speaking, these $\mu, \eta$ together equip $\dist$ with a structure of a \emph{monad}. See \cite{MacLane71,Awodey06,HasuoJS07b}.

\myparagraph{MC2CM $\lambda$}
The MC2CM operator $\lambda$ is much less known. Its type with annotations is $\lambda\colon\Dmass(\Dchan(X))\to\Dchan(\Dmass(X))$; hence the name (``mass-chance to chance-mass''). Its definition is best described as ``switching the order of two probabilistic branching'', as illustrated in~\cref{fig:MC2CM}. 

Again we consider the coin-tossing analogy in~\cref{ex:coinToss} to provide more intuition. The idea of ``switching'' in this example can be simply understood as swapping the order of tossing the coins, but this cannot be carried out directly since the choice of the second coin $C_{a}^{\beta}$ depends on the outcome of the first coin $C^{\alpha}$. A solution to this is as follows, we toss all the coins $(C_{a}^{\beta})_{a\in {\sf A}}$ at once and record all the outcomes for each $a\in {\sf A}$. Then we toss the coin $C^{\alpha}$ to obtain an outcome $a\in{\sf A}$, a pair of outcome $(a,b)$ and thus $x=f(a,b)\in X$ can be decided by checking the record for $(C_{a}^{\beta})_{a\in{\sf A}}$ we made in the first place. Intuitively, the probability of getting a specific table of records, denoted as $R:{\sf A}\to {\sf B}$, is simply the product of the probability of getting each individual record $R(a)$, that is $\prod_{a\in {\sf A}}p_{a,R(a)}$. This explains the right-hand side of~\cref{eq:MC2CM}, which all mappings $f\in X^{I}$ in the outer summation are enumerated with the product of $b_{i,f(i)}$ as the corresponding probability.
%upon flipping the order of the probabilistic branching

\section{Quantifier Elimination for the template-based synthesis algorithm on the \msct{} Semantics}
We briefly describe the approach for performing quantifier elimination on the quantified constraints obtained in \cref{subsec:msctAlgo}. The main tools we rely on are the Farkas's lemma~\cite{farkas1902theorie} and Handelman's theorem~\cite{handelman1988representing}
%and their extension with non-strict inequalities support \cite{amir23}.
We first recall each of them below.

\begin{theorem}[Farkas's lemma]
  A system of $m$ linear inequalities $\Phi\vx\succeq 0$ entails a linear inequality $\phi(\vx)\geq 0$ if and only if there exists $\vy\in\mathbb{R}^{m}_{\geq 0}$ such that $\phi(\vx)=\vy^{T}\Phi\vx$, that is,
  $\forall\vx\in\mathbb{R}^{n}, \left(\Phi\vx\succeq 0\implies \phi(\vx)\geq 0\right)\iff\exists \vy\in\mathbb{R}^{m}_{\geq 0},\,\left(\phi(\vx)=\vy^{T}\Phi\vx\right)$.
\end{theorem}

\begin{theorem}[Handelman's theorem]
  Let $\Phi(\vx)\succeq 0$ be a system of $m$ linear inequalities over $\vx\in\mathbb{R}^{n}$, that is a conjunction of $m$ inequalities in the form of $\Phi_{1}(\vx)\geq 0, \Phi_{2}(\vx)\geq 0,...,\Phi_{m}(\vx)\geq 0$. We write $\Psi\in\Phi$ to denote $\Psi$ is the left-hand side of one of the linear inequalities in $\Phi(\vx)\succeq 0$. For any system of linear inequalities, the semigroup generated by $\Phi$ under the usual multiplication, denoted as ${\sf Prod}(\Phi)$, is defined by ${\sf Prod}(\Phi)=\{\prod_{k=1}^{K}\Psi_{k}(\vx)\mid K\in \{0\}\cup[m],\Psi_{k}(\vx)\in\Phi(\vx)\}$. Note that the case $K=0$ corresponds to $1$.

  \noindent
  Now further assume the system $\Phi(\vx)\succeq 0$ has a non-empty and compact feasible region, then it entails a {\emph polynomial} inequality $\phi(\vx) > 0$ if and only if $\phi$ can be expressed as a convex combination of elements in the semigroup ${\sf Prod}(\Phi)$. That is, $\forall\vx\in\mathbb{R}^{n}, \left(\Phi\vx\succeq \mathbf{0}\implies \phi(\vx) > 0\right)\iff\exists y_{1},...,y_{t}\geq 0$ and $\psi_{1},...,\psi_{t}\in {\sf Prod}(\Phi)$ such that $\phi(\vx)=\sum_{i=1}^{t}y_{i}\psi_{i}(\vx)$.
  %Fix a system of $m$ linear inequalities $\Phi\vx\succeq 0$, let ${\sf Prod}(\Phi):=\{\prod_{k=1}^{K}\Psi_{k}\mid K\in \{0\}\cup[m],\Psi_{k}\in\Phi\}$ denotes the semigroup generated by $\Phi$ under the usual multiplication, where the case $K=0$ corresponds to $1$ and the notation $\Psi_{k}\in\Phi$ is abused to represent $\Psi_{k}$ is one of the linear inequalities in $\Phi$. If the feasible region of $\Phi$ is non-empty and compact, then $\Phi\vx\succeq 0$ entails a polynomial inequality $\phi(\vx)> 0$ if and only if $\phi$ can be expressed as some convex combination of $\psi\in {\sf Prod}(\Phi)$, i.e. $\forall\vx\in\mathbb{R}^{n}, \left(\Phi\vx\succeq \mathbf{0}\implies \phi(\vx) > 0\right)\iff\exists y_{1},...,y_{t}\geq 0$ and $\psi_{1},...,\psi_{t}\in {\sf Prod}(\Phi)$ such that $\phi(\vx)=\sum_{i=1}^{t}y_{i}\psi_{i}(\vx)$.
\end{theorem}

We require only the Farkas's lemma for most of the constraints obtained in \cref{subsec:msctAlgo}, except for the constraint in \cref{eqt:constraint-subm-inductive}. For the latter, we have a polynomial expression on the right-hand side of the $\implies$ after expanding the scheduler with the template variables, hence we need to apply the Handelman's theorem. However, the Handelman's theorem cannot be applied directly since the polynomial inequality constructed there is non-strict. Nevertheless, if the polynomial $\phi(\vx)$ can be expressed as a convex combination of elements in ${\sf Prod}(\Phi)$, the entailment can still be implied even when the inequality is non-strict.

%Just as in \cite{Akshay2024}, we cannot directly apply Handelman's theorem since the polynomial inequality we constructed is non-strict, but the $\impliedby$ direction is still sound trivially.

% Then we note that for constraint in \cref{eqt:constraint-scheduler,eqt:constraint-subm-bound}, all implications can be expressed as linear inequalities, in particular $d\in\dist(Q)$ can be equivalently expressed as $\bigwedge_{q\in Q} d(q)\geq 0\land 1-\sum_{q\in Q}d(q)\geq 0\land \sum_{q\in Q}d(q)-1\geq 0$.

\myparagraph{Applying the lemma and theorem} The key idea of applying the lemma and theorem is to transform the problem of deciding the satisfiability of a $\forall$-quantified formula into the problem of searching a feasible solution over polynomial inequalities. In both the Farkas's lemma and the Handelman's theorem, the satisfiability can be ensured if the right-hand side can be expressed as a convex combination over some domain induced by the left-hand side. Therefore, we set up an extra set of template variables to search for the coefficients of the convex combination required that imply the satisfiability of the quantified formula.

For Farkas's lemma, a template variable is declared for each linear inequalities on the left-hand side of the implication, that is a template variable $y_{j}$ is declared for each $\Phi_{j}\in \Phi$. As hinted, these variables are intended to serve as the coefficient of the convex combination over the linear inequalities $\Phi(\vx)$ so that it is equivalent to the right-hand side $\phi(\vx)$. Therefore, for each $x_{i}$ of the variable $\vx$, we enforce an equality between the coefficients of each $x_{i}$ in $\phi(\vx)$ and $\sum_{j}y_{j}\Phi_{j}(\vx)$. Solving the set of equalities generated ensures the convex combination and thus the satisfiability of the quantified formula. Further, this translation can be done efficiently by noting that the left-hand side of the constraints obtained in \cref{subsec:msctAlgo} are all the same, namely $d\in\dist(Q)$. The translation can be done once only and reused for other constraints.

The translation is more complicated when applying the Handelman's theorem. Since we need to consider the semigroup generated by the product of the linear inequalities, which is countably infinite in general and thus impossible to compute explicitly. Practically, we put a bound on the degree of the polynomial generated in the semigroup and consider only the convex combination generated among them. Then the same procedure as in the Farkas's lemma can be applied here.

There are two small technical details involves in the translation appears in the constraint \cref{eqt:constraint-subm-inductive}. The left-hand side consists of the expression $d\in\dist(Q)\setminus H$, which is non-trivial that it can be expressed as a conjunction of linear inequalities. With some algebraic manipulation, one can show that the expression $d\in\dist(Q)\setminus H$ can be equivalently expressed as follows,
\begin{equation}\label{eqt:Antichain-Conjunction}\small
  \begin{array}{l}\textstyle
    \left[d\in {\cal D}(Q)\setminus H\implies \phi(d)\right] \iff\bigwedge_{{\bf q}\in Q^{n}}\left[(d\in {\cal D}(Q))\land \bigwedge_{i=1}^{n}(d({\bf q}_{i})<x_{i}({\bf q}_{i}))\implies\phi(d)\right].
  \end{array}
\end{equation}
The translation can therefore be carried out by first rewriting the constraint in \cref{eqt:constraint-subm-inductive} as a conjunction of $\abs{Q}^{n}$ constraints and apply the Handelman's theorem to each of them. While for the right-hand side, noting that we do not follow exactly the definition in \cref{def:gamma-subm}, since that would result in an exponential size of constraints. We show, later in Appendix \ref{sec:proof}, that when the submartingale is linear, then the inductive expression can be simplified to the form we show above.

\section{Omitted Definitions}\label{sec:omittedDef}

\myparagraph{Formal definition of the threshold reachability problem}

First, we define the set of configuration paths induced by a config MC. Intuitively, it is the set of all possible paths that the config MC can take.

\begin{definition}[configuration path]\label{def:configPath}
  In the setting of \cref{def:ConfMC}, we define
  \begin{math}
    \cpath_{\mathbb{X},\sigma}=
    \bigl\{d_{0} d_{1}\dotsc d_{n}\in \dist(Q)^{+}
    \,\big|\,
    \bigl(d_{0},\, d_{0}d_{1},\, \dotsc,\, d_{0}d_{1}\dotsc d_{n}\bigr) \text{ is a (finite) path in $\mathcal{M}_{\mathbb{X},\sigma}$}
    \bigr\}
  \end{math}, where  \emph{path} of a Markov chain is defined as usual (requiring non-zero transition probabilities).
  Elements of
  \begin{math}
    \cpath_{\mathbb{X},\sigma}
  \end{math}
  are called \emph{configuration paths} of $\mathcal{M}$ under $\mathbb{X}, \sigma$. The set $\cpath^{\omega}_{\mathbb{X},\sigma}$ of \emph{infinite configuration paths} is defined similarly.

  For a memoryless scheduler $\sigma\colon\dist(Q)\to \dist(\act)$, the above definition simplifies to the sets of (finite and infinite) paths in the config MC $\mathcal{M}_{\mathbb{X},\sigma}$ (\cref{def:ConfMC}), that is, $\cpath_{\mathbb{X},\sigma}=\{d_{0}d_{1}...d_{n}\in\dist(Q)^{+}\mid (d_{0},d_{1},...,d_{n}) \text{ is a finite path in $\mathcal{M}_{\mathbb{X},\sigma}$}\}$ and similarly for the set of infinite paths.
\end{definition}

It follows from basic measure theory that 1) $\cpath^{\omega}_{\mathbb{X},\sigma}$ comes with a natural $\sigma$-algebra ${\cal F}_{\mathbb{X},\sigma}$ generated by cylindric sets, and 2) there is a unique probability measure $\Pb_{\mathbb{X},\sigma}$ on $(\cpath^{\omega}_{\mathbb{X},\sigma},{\cal F}_{\mathbb{X},\sigma})$ induced by the config MC $\mathcal{M}_{\mathbb{X},\sigma}$. We obtain these by applying  results for general MCs \cite{Baier2008,AshD99} to the config MC $\mathcal{M}_{\mathbb{X},\sigma}$ (\cref{def:ConfMC}).

The formal definition of the threshold reachability problem can thus be defined as follows.

\begin{problem}[threshold reachability problem]\label{prob:GeneralThresholdFormal}
  Assume the setting of \cref{def:ConfMC}. Let $d_{0}\in \dist(Q)$ and  $H\subseteq \dist(Q)$, called an \emph{initial configuration} and a \emph{target set}, respectively. We let $\reach_{\mathbb{X},\sigma}(d_{0},H)$ denote the set of paths that start from $d_{0}$ and eventually reach $H$, it is formally defined as follows:
    \begin{displaymath}
      \reach_{\mathbb{X},\sigma}(d_{0},H)
      \;=\;
      \{d_{0}d_{1}d_{2}\dotsc\in \cpath^{\omega}_{\mathbb{X},\sigma}\mid \text{$d_{i}\in H$ for some $i$}\},
    \end{displaymath}
  The \emph{threshold reachability problem} asks for a scheduler $\sigma$ such that
  \begin{math}
    \Pb_{\mathbb{X},\sigma}\bigl(\reach_{\mathbb{X},\sigma}(d_{0},H)\bigr)\geq \xi
  \end{math}.
  \end{problem}

\myparagraph{General definition of the config MC} 
The following is the general case (allowing memoryful schedulers) of the memoryless definition in \cref{sec:instances}.

\begin{definition}[config MC under \csct{} semantics, the general case]\label{def:ConfMCCSCTGeneral}
  Let ${\cal M}=(Q,\act,\delta)$ be an MDP, and  $\sigma\colon \dist(Q)^{+}\to \dist(\act)$ be a scheduler. The config MC induced by ${\cal M}, \mathbb{X}^{\csct{}}, \sigma$, denoted by
  \begin{math}
    {\cal M}^{\csct{}}_{\sigma}=(\dist(Q)^{+},\delta_{\mathbb{X}^{\csct{}},\sigma})
  \end{math}, is given as follows.

  The config MC's state space is the set $\dist(Q)^{+}$ of sequences of configurations of the MDP $\mathcal{M}$. Its transition function $\delta_{\mathbb{X}^{\csct{}},\sigma}\colon\dist(Q)^{+}\to\dist(\dist(Q)^{+})$ is defined as follows: let $p=d_{0}d_{1}...d_{n}$ be a finite sequence of configurations, and $d_{n+1}$ be the next configuration, then
  \begin{displaymath}
    \delta_{\mathbb{X}^{\csct{}},\sigma}(p)(p\cdot d_{n+1})=\sum_{a\in\act}\sum_{f\in Q^{Q} \text{ s.t. } \dist f(d_{n})=d_{n+1}}\sigma(p)(a)\cdot\bigg(\prod_{q\in Q}\delta(q,a)\bigl(f(q)\bigr)\bigg).
  \end{displaymath}
\end{definition}

\begin{definition}[config MC under \msct{} semantics, the general case]\label{def:ConfMCMSCTGeneral}
  Assume the setting that is the same as in~\cref{def:ConfMCCSCTGeneral} except for $\mathbb{X}^{\msct{}}$ being the CM classifier instead. Let ${\cal M}^{\msct{}}_{\sigma}=(\dist(Q)^{+},\delta_{\mathbb{X}^{\msct{}},\sigma})$ denotes the config MC induced by ${\cal M}, \mathbb{X}^{\msct{}}, \sigma$. Its transition function $\delta_{\mathbb{X}^{\msct{}},\sigma}\colon\dist(Q)^{+}\to\dist(\dist(Q)^{+})$ is concretely given as follows: let $p=d_{0}d_{1}...d_{n}$ be a finite sequence of configurations, and $d_{n+1}$ be the next configuration, then
  \begin{displaymath}
    \delta_{\mathbb{X}^{\msct{}},\sigma}(p)(p\cdot d_{n+1})=\sum_{a\in\act}\sum_{f\in Q^{Q\times\act}\text{ s.t. } d_{n+1}=\dist f(d_{n}\otimes\sigma(p))}\prod_{(q,a)\in Q\times\act}\delta(q,a)(f(q,a)).
  \end{displaymath}
\end{definition}

\begin{definition}[config MC under \csmt{} semantics, the general case]\label{def:ConfMCCSMTGeneral}
  Assume the setting that is the same as in~\cref{def:ConfMCCSCTGeneral} except for $\mathbb{X}^{\csmt{}}$ being the CM classifier instead. Let ${\cal M}^{\csmt{}}_{\sigma}=(\dist(Q)^{+},\delta_{\mathbb{X}^{\csmt{}},\sigma})$ denotes the config MC induced by ${\cal M}, \mathbb{X}^{\csmt{}}, \sigma$. Its transition function $\delta_{\mathbb{X}^{\csmt{}},\sigma}\colon\dist(Q)^{+}\to\dist(\dist(Q)^{+})$ is concretely given as follows: let $p=d_{0}d_{1}...d_{n}$ be a finite sequence of configurations, and $d_{n+1}$ be the next configuration, then
  \begin{displaymath}
    \delta_{\mathbb{X}^{\csmt{}},\sigma}(p)(p\cdot d_{n+1})=\sum_{a\in\act\text{ s.t. }d_{n+1}=\mu_{Q}(\dist(\delta^{\wedge}(a))(d))}\sigma(p)(a).
  \end{displaymath}
\end{definition}

\begin{definition}[config MC under \msmt{} semantics, the general case]\label{def:ConfMCMSMTGeneral}
  Assume the setting that is the same as in~\cref{def:ConfMCCSCTGeneral} except for $\mathbb{X}^{\msmt{}}$ being the CM classifier instead. Let ${\cal M}^{\msmt{}}_{\sigma}=(\dist(Q)^{+},\delta_{\mathbb{X}^{\msmt{}},\sigma})$ denotes the config MC induced by ${\cal M}, \mathbb{X}^{\msmt{}}, \sigma$. Its transition function $\delta_{\mathbb{X}^{\msmt{}},\sigma}\colon\dist(Q)^{+}\to\dist(\dist(Q)^{+})$ is concretely given as follows: let $p=d_{0}d_{1}...d_{n}$ be a finite sequence of configurations, and $d_{n+1}$ be the next configuration, then
  \begin{displaymath}
    \delta_{\mathbb{X}^{\msmt{}},\sigma}(p)(p \cdot d_{n+1})=\begin{cases}
      1 & \quad d_{n+1}=\sum_{a\in\act}\sum_{q'\in Q}\sum_{q\in Q}\sigma(p)(a) d_{n}(q)\delta(q,a)(q')\ket{q'} \\
      0 & \quad \text{otherwise}
    \end{cases}.
  \end{displaymath}
\end{definition}

% \begin{definition}[config MC under \csct{} semantics, the general case]\label{def:ConfMCCSCTGeneral}
% (** to be completed**)

% \end{definition}

\myparagraph{Pseudocode for the antichain-based algorithm on the \csmt{} semantics}

\begin{algorithm}[H]\footnotesize
  \caption{Backward Projection Algorithm}\label{algo:anti}
  \DontPrintSemicolon
  \SetKwRepeat{Do}{do}{while}
  \KwIn{an MDP ${\cal M}=(Q,\act,\mat)$, an initial configuration $d_{0}\in\dist(Q)$, a target set antichain $\floor{H}\subset [0,1]^{\abs{Q}}$, parameters $K,L\in\mathbb{N}$}
  \KwOut{{\sf True} if $d_{0}$ is reachable to $H$}
  $S\leftarrow \{\}$\\
  $S'\leftarrow \floor{H}$\\
  \Do{$S\neq S'$}{
    $S\leftarrow S'$\\
    $S'\leftarrow \floor{H}$\\
    \For{$M\in\mat$}{
      \For{$\vx\in S$}{
        $\{\vy_{k}^{\ast}\mid k\in [K]\}\leftarrow$ Solve problem in \cref{eqt:anti-min} against $M$,$\vx$ with parameter $K,L$\label{line:iter}\\
        $S'\leftarrow \floor{S'\cup\{\vy_{k}^{\ast}\mid k\in [K]\}}$\label{line:floor}
      }
    }
    \lIf{$d_{0}\in S'$}{\Return{{\sf True}}}
  }
  \lIf{$d_{0}\in S'$}{\Return{{\sf True}}}
\end{algorithm}

\section{Omitted Proofs}\label{sec:proof}

This part of the appendix contains the proof for lemma and theorems for the paper.

\myparagraph{Proof of equivalence between \csct{} semantics and conventional semantics}
We claimed in \cref{subsec:KnownAlgoResult} that the \csct{} semantics reduces to the conventional semantics under the assumption that the initial configuration is Dirac. We formally prove it here. It suffices to show that $\delta_{\mathbb{X}^{\csct{}},\sigma}$ maps a Dirac configuration to a distribution of Dirac configuration, the rest follows by a simple induction on path. We show it for memoryless case and its generalization is trivial, let $d\in\dist(Q)$ be Dirac configuration with supports on $s\in Q$, then we have

\begin{align*}
  \delta_{\mathbb{X}^{\csct{}},\sigma}(d)&=\sum_{a\in\act}\sum_{f\in Q^{Q}}\sigma(d)(a)\cdot\bigg(\prod_{q\in Q}\delta(q,a)(f(q))\bigg)\biggket{\sum_{q\in Q}d(q)\ket{f(q)}}\\
  &=\sum_{a\in\act}\sum_{f\in Q^{Q}}\sigma(d)(a)\cdot\bigg(\prod_{q\in Q}\delta(q,a)(f(q))\bigg)\biggket{\ket{f(p)}}\\
  &=\sum_{a\in\act}\sigma(d)(a)\sum_{f\in Q^{Q}}\bigg(\prod_{q\in Q\setminus\{p\}}\delta(q,a)(f(q))\bigg)\cdot \delta(p,a)(f(p))\biggket{\ket{f(p)}}\\
  &=\sum_{a\in\act}\sigma(d)(a)\sum_{q\in Q}\delta(p,a)(q)\biggket{\ket{q}}\sum_{\substack{f\in Q^{Q}\\f(p)=q}}\bigg(\prod_{q\in Q\setminus\{p\}}\delta(q,a)(f(q))\bigg)\\
  &=\sum_{a\in\act}\sum_{q\in Q}\sigma(d)(a)\cdot \delta(p,a)(q)\biggket{\ket{q}}.
\end{align*}
It is not hard to see that $\sigma(d)(a)\cdot\delta(p,a)(q)\biggket{\ket{q}}$ is a Dirac configuration and therefore concluding our proof.

\myparagraph{Proof of global scheduler subsume local scheduler under \msmt{} semantics}
\begin{proof}
  Let ${\cal M}=(Q,\act,\delta)$ be an MDP and $\sigma\colon\dist(Q)\times Q\to \dist(\act)$ be a local scheduler. We can construct an MDP ${\cal M}'=(Q,\act^{Q},\delta')$ with $\delta'(q,f)=\delta(q,f(q))$ and a global scheduler $\sigma'\colon \dist(Q)\to\dist(\act^{Q})$ such that $\sigma'(d)(f)=\prod_{q\in Q}\sigma(d,q)(f(q))$. Then it is not hard to see that for any configurations $d\in\dist(Q)$, we have the transition in the config MC induced by ${\cal M}'$ and $\sigma'$ is
  \begin{align*}
    \delta_{\mathbb{X}^{\msmt{}},\sigma'}(d) &= \biggket{\sum_{f\in\act^{Q}}\sum_{q'\in Q}\sum_{q\in Q}\sigma'(d)(f) d(q)\delta'(q,f)(q')\ket{q'}}\\
    &= \biggket{\sum_{q'\in Q}\sum_{q\in Q}\sum_{f\in\act^{Q}}\left(\prod_{s\in Q}\sigma(d,s)(f(s))\right)d(q)\delta(q,f(q))(q')\ket{q'}}\\
    &= \biggket{\sum_{q'\in Q}\sum_{q\in Q}\sum_{a\in\act}\sum_{\substack{f\in\act^{Q},\\f(q)=a}}\bigg(\prod_{s\in Q}\sigma(d,s)(f(s))\bigg)d(q)\delta(q,f(q))(q')\ket{q'}}\\
    &= \biggket{\sum_{q'\in Q}\sum_{q\in Q}\sum_{a\in\act}d(q)\delta(q,a)(q')\sigma(d,s)(a)\sum_{\substack{f\in\act^{Q},\\f(q)=a}}\bigg(\prod_{s\in Q}\sigma(d,s)(f(s))\bigg)\ket{q'}}\\
    &= \biggket{\sum_{q'\in Q}\sum_{q\in Q}\sum_{a\in\act}d(q)\delta(q,a)(q')\sigma(d,s)(a)\ket{q'}}.
  \end{align*}%
  Since the config MC ${\cal M}'$ has the transition function as the MC induced by ${\cal M}$ and local scheduler $\sigma$ as defined in \cite{Akshay2023}. Hence, they must generate the same distribution of reachable configuration and thus concluding the proof.
  % Using the currying notation, that the local scheduler can be equivalently expressed as $\sigma^{\wedge}\colon\dist(Q)\to \dist(\act)^{Q}$, where $\sigma^{\wedge}(d)(q)=\sigma(d,q)$. Therefore, w
\end{proof}

\myparagraph{Proof of \cref{thm:msct-sharpP-hard}}
\MSCTSharpP*{}
\begin{proof}
  We provide a reduction from the counting variant of the subset-sum problem which is known to be $\sharp$P-hard~\cite{book-papa}. It suffices for us to consider only set of positive integers, fix a set $S=\{a_{1},a_{2},...,a_{n}\in\mathbb{N}\}$ and a target sum $T$, we construct a simple MDP ${\cal M}=(Q=[n]\cup\{\top,\bot\},\act=\{0\},\delta)$, such that the transition acts as follows, $\forall i\in [n]$,
  $$\delta(i,0)(\top)=\delta(i,0)(\bot)=\frac{1}{2}$$
  and
  $$\delta(\bot,0)(\bot)=\delta(\top)(\top)=1.$$
  Finally, we fix the parameter $\xi=\frac{1}{2^{n}}$ and define the initial distribution $d_{0}$ target set $H$ as follows,
  $$d_{0}(n)=\frac{a_{n}}{\sum_{i=1}^{n}a_{i}}$$
  and
  $$H=\left\{d\in \dist(Q)\left| d(\top)=\frac{T}{\sum_{i=1}^{n}a_{i}},d(\bot)=1-d(\top)\right.\right\}.$$
  We start by noting that the system immediate stabilizes, i.e. reaches a configuration that remains unchanged upon transition, after one step. Hence, it suffices for us to consider just a single step of transition from $d_{0}$. Further there is only a single action $0$, hence there is only a trivial scheduler $\sigma$ can be defined, which satisfies $\sigma(d)=0$ for any $d\in\dist(Q)$. The reachability probability for the target set $H$ can be expressed as follows,
  $$\Pb_{\sigma}^{\msct{}}(\reach(d_{0}, H))=\sum_{d\in H}\delta_{\mathbb{X}^{\msct},\sigma}(d_{0})(d).$$
  From the transition function $\delta_{\mathbb{X}^{\msct{}}}$ defined in \cref{def:ConfigMCsInstances}, we can observe that each mapping $f\colon Q\times\act\to Q$ induces a transition in the configuration space. Further, since there is only one action $0$ and all the concentrations are on $[n]$, we can restrict our attention to the set of mappings $f:[n]\to\{\bot,\top\}$ as all other mappings are not possible (i.e. having a zero probability) from $d_{0}$. Hence, by expanding the definition of $\delta_{\mathbb{X}^{\msct{}},\sigma}$ and applying the above observation, we have
  $$\Pb_{\sigma}^{\msct{}}(\reach(d_{0}, H))=\sum_{d\in H}\sum_{\substack{f\in [n]^{\{\bot,\top\}}\colon \\ d=\dist f(d_{0})}}\prod_{(q,a)\in Q\times\act}\delta(q,a)(T(q,a)).$$
  Since we are limiting to the mapping $f\colon [n]\to \{\bot,\top\}$ that occurs with non-zero probability and all transitions from $[n]$ to $\{\bot,\top\}$ occur with probability $\frac{1}{2}$, the expression above can be simplified as follows,
  %the innermost product can thus be simplified as $\frac{1}{2^{n}}$. %Besides, the target set $H$ is a singleton set, thus the whole expression above can be further simplified as follows,
  $$\Pb_{\sigma}^{\msct{}}(\reach(d_{0}, H))=\sum_{\substack{f\in [n]^{\{\bot,\top\}}\colon \\ \dist f(d_{0})\in H}}\frac{1}{2^{n}}.$$
  Lastly, note that each $f$ induces a subset over $[n]$, namely $\{q\in [n]\mid f(q)=\top\}$. By having $\dist f(d_{0})\in H$, it implies the follows,
  $$\sum_{k\in [n]:f(k)=\top}d_{0}(k)=\sum_{k\in [n]:f(k)=\top}\frac{a_{k}}{\sum_{i=1}^{n}a_{i}}=\frac{T}{\sum_{i=1}^{n}a_{i}}.$$
  Simplifying the equality gives,
  $$\sum_{k\in [n]:f(k)=\top}a_{k}=T.$$
  Therefore, for each $f\in [n]^{\{\bot,\top\}}$ such that $\dist f(d_{0})\in H$, it corresponds to a subset of $[n]$ that sums to $T$. Noting that the reverse holds trivially, for any subset $A\subset S$ that sum to $T$, one can define a map $f_{A}\colon [n]\to\{\bot,\top\}$ such that $f(k)=\top\iff a_{k}\in A$. Hence, we have
  $$\Pb_{\sigma}^{\msct{}}(\reach(d_{0}, H))=\frac{\abs{\{A\subseteq S\mid \sum_{a\in A} a = T\}}}{2^{n}}.$$
  Therefore, a binary search over the discrete domain of $\left\{\frac{i}{2^{n}}\mid i\in \mathbb{N}\cup\{0\}\right\}$ can be done to compute $\abs{\{A\subseteq S\mid \sum_{a\in A} a = T\}}$, which is the solution to subset-sum counting problem. The construction of the machine takes $O(n)$ time and the binary search add an extra factor of $O(n)$ to the overall complexity, which shows the reduction is polynomial overall, hence concluding our proof.
\end{proof}

\myparagraph{Proof of \cref{thm:msct-wScheduler-pos-hard}}
\MSCTPosHard*{}
\begin{proof}
  Here we give a brief outline of the proof. The goal is to construct a polynomial time reduction from the problem A stated in \cite{Akshay2015}. To simulate the run of a Markov chain under mass-interpretation, we need to rely on the mass-interpreted scheduler under the \msct{} semantics. Given a Markov Chain ${\cal M}=(Q,M)$ where $M$ denotes its transition matrices, since $Q$ is finite we are safe to assume that $Q$ is ordered by a label, i.e. $Q=\{q_{1},q_{2},...,q_{\abs{Q}}\}$ and the transition matrix $M$ is arranged in the same order.
  
  We construct an MDP $\widetilde{\cal M}=(\widetilde{Q}=Q\times\{0,1,2\}\cup\{\phi,\psi\}, \act=E\cup\{\bot\}, \delta)$, where $E=\{(p,q)\in Q\mid (M)_{p,q}>0\}$. The MDP consists of three copies of $Q$ and two auxiliary states $\phi,\psi$ as state space, every transition in the original Markov chain is now a unique action in the MDP $\widetilde{\cal M}$. The transition function $\delta$ is defined as follows, for any $p,q\in Q$,
  \begin{align*}
    (p,q)\in E\implies \delta((p,0),(p,q))(q,1) &=1,\\
    \delta((p,0),\bot)(\phi)&=1,\\
    e\in E\implies \delta((q,1),e)((q,1)) &= 1,\\
    \delta((q,1),\bot)((q,0)) &= 1,\\
    j\in[\abs{Q}]\land e\in E\implies \delta((q_{j},2),e)((q_{j+1},2)) &= 1,\\
    e\in E\implies \delta((q_{\abs{Q}},2),e)(\psi) &= 1,\\
    j\in[\abs{Q}]\implies \delta((q_{j},2),\bot)(\phi) &= 1,\\
    a\in\act=E\cup\{\bot\}\implies \delta(\phi,a)(\phi) &=1,\\
    a\in\act=E\cup\{\bot\}\implies \delta(\psi,a)((q_{1},2)) &=1.
  \end{align*} 
%  for every $(p,q)\in E$ such that $\delta((p,0),(p,q))(q,1)=\delta((p,0),\bot)(\phi)=\delta((q,1),)=\delta((q,1),\bot)((q,0))=1$. 
  
  The intuition of such construction is to simulate the transition of the Markov chain state-by-state, a transition $(p,q)$ in the Markov chain is simulated by applying the corresponding action $(p,q)\in E$ which maps the state from $(p,0)$ to $(q,1)$ in the MDP $\widetilde{\cal M}$. The set of state $\{(q_{j},2)\}_{j}$ is designed to keep track of the state simulated. Therefore, we can construct a scheduler $\sigma$ such that it choose the corresponding sets of actions based on the transition matrix $M$ and the current state to be simulated. Under suitable choice of initial configuration, the above construction of the MDP $\widetilde{\cal M}$ and scheduler $\sigma$ can simulate the run of the Markov chain ${\cal M}$ under mass-interpretation. For any initial configuration $d_{0}\in\dist(Q)$, we construct the initial configuration $d_{0}'\in\dist(\widetilde{Q})$ as follows, for any $q\in Q$,
  $d_{0}'((q,0))=\frac{d_{0}(q)}{2}, d_{0}'((q_{1},2))=0.5$.
  
  % will be adjusted such that on each action $(p,q)\in E$, it corresponds to a transition from $(p,0)$ to $(q,1)$ with probability $1$. An auxiliary action $\bot$ is defined such that upon $\bot$ is applied, it corresponds to a transition from $(q,1)$ to $(q,0)$ for all $q\in Q$. The intuition of such construction is to simulate the transition of the Markov chain state-by-state, on each step, our scheduler will pick one state $p$ based on the history, and simulate the transition by applying $(p,q)\in E$ with probability according to $(M)_{p,q}$. Then we can simulate one step in the original Markov chain ${\cal M}$ with $\abs{Q}$ steps in $\widetilde{\cal M}$, the auxiliary action $\bot$ will be applied every $\abs{Q}+1$ steps to set the mass back to the $0$-copy of each state so that the next transition can be again simulated. Then since a Markov chain can be simulated with the scheduler constructed as above
  Finally, by inquiring the reachability probability with target set containing all distributions that the mass on $Q\times \{0\}$ satisfies the condition stated in problem A and threshold $\xi=1$, we can solve decide problem A and hence it follows from result in \cite{Akshay2015} that the inequalities version of the problem is Positivity-hard thus concluding our proof. Note that the constructed scheduler $\sigma$ is memoryless, hence the hardness results holds even when restricted to memoryless scheduler.
\end{proof}

\myparagraph{Proof of right-hand side of \cref{eqt:constraint-subm-inductive}}

This proof is to show that when the submartingale $R$ is linear, then the inequality in \cref{eqt:gamma-subm} can be equivalently expressed as the right-hand side of \cref{eqt:constraint-subm-inductive}. Concretely, we aim to show the following:
$$\sum_{d'\in\dist(Q)}R(d')\cdot\delta_{\mathbb{X}^{\msct{}},\sigma}(d)(d') = \sum_{a\in\act}\sigma(d)(a)R(M_{a}^{T}(d)).$$
\begin{proof}
  \begin{align*}
    &\sum_{d'\in\dist(Q)}R(d')\cdot\delta_{\mathbb{X}^{\msct{}},\sigma}(d)(d') \\
    & = \sum_{d'\in\dist(Q)}R(d')\sum_{\substack{f\in Q^{Q\times\act}\\d'=\dist f(d\otimes\sigma(d))}}\prod_{(p,a)\in Q\times\act}\delta(p,a)(f(p,a)) \\
    & =\sum_{f\in Q^{Q\times\act}}\left(r_{0}+\sum_{q\in Q}r_{q}\dist f(d\otimes\sigma(d))(q)\right)\prod_{(p,a)\in Q\times \act}\delta(p,a)(f(p,a))\\
    & =r_{0}\sum_{f\in Q^{Q\times\act}}\prod_{(p,a)\in Q\times \act}\delta(p,a)(f(p,a))\\
    & \quad+\sum_{f\in Q^{Q\times\act}}\left(\prod_{(p,a)\in Q\times \act}\delta(p,a)(f(p,a))\right)\sum_{q\in Q}r_{q}\dist f(d\otimes\sigma(d))(q)
  \end{align*}
  Note that the former summation sum to $1$, hence simplifying it and substituting the pushforward $\dist f$ gives the following,
  \begin{align*}
     & r_{0} + \sum_{f\in Q^{Q\times\act}}\sum_{q\in Q}r_{q}\left(\prod_{(p,a)\in Q\times \act}\delta(p,a)(f(p,a))\right)\sum_{(s,b)\in Q\times\act, f(s,b)=q}\sigma(d)(b)d(s) \\
     & =r_{0} + \sum_{q\in Q}r_{q}\sum_{f\in Q^{Q\times\act}}\sum_{(s,b)\in Q\times\act, f(s,b)=q}\sigma(d)(b)\left(\prod_{(p,a)\in Q\times \act}\delta(p,a)(f(p,a))\right)d(s)      \\
     & =r_{0} + \sum_{q\in Q}r_{q}\sum_{(s,b)\in Q\times\act}\sigma(d)(b)d(s)\delta(s,b,q)\sum_{\substack{f\in Q^{Q\times\act},\\f(s,b)=q}}\left(\prod_{(p,a)\neq (s,b)}\delta(p,a)(f(p,a))\right)\\
     & =r_{0} + \sum_{q\in Q}r_{q}\sum_{(s,b)\in Q\times\act}\sigma(d)(b)d(s)\delta(s,b,q)\sum_{f\in Q^{(Q\times\act)\setminus\{(s,b)\}}}\left(\prod_{(p,a)\in (Q\times\act)\setminus\{(s,b)\}}\delta(p,a)(f(p,a))\right)\\
     & =r_{0} + \sum_{q\in Q}r_{q}\sum_{(s,b)\in Q\times\act}\sigma(d)(b)\delta(s,b,q) d(s)                                                                         \\
     & =r_{0}+ \sum_{b\in\act}\sigma(d)(b)\sum_{q\in Q}r_{q}\sum_{s\in S}\delta(s,b,q) d(s) = \sum_{a\in\act}\sigma(d)(a)R(M_{a}^{T}(d))
  \end{align*}\end{proof}

\myparagraph{Proof of \cref{eqt:Antichain-Conjunction}}

Here we show that the left-hand side of \cref{eqt:Antichain-Conjunction} can be expressed as combination in linear inequalities. Concretely, we aim to show the following:
$$\small\begin{array}{l}\textstyle
  \left[d\in {\cal D}(Q)\setminus H\implies \phi(d)\right] \iff\bigwedge_{{\bf q}\in Q^{n}}\left[(d\in {\cal D}(Q))\land \bigwedge_{i=1}^{n}(d({\bf q}_{i})<x_{i}({\bf q}_{i}))\implies\phi(d)\right].
\end{array}$$
\begin{proof}
  \begin{align*}
    d\in {\cal D}(Q)\setminus H                         & \iff d \in {\cal D}(Q) \land \bigwedge_{i=1}^{n}d\notin \uparrow \{\vx_{i}\}                                                                \\
                                                        & \iff \bigwedge_{i=1}^{n}(d \in {\cal D}(Q))\land\bigvee_{q\in Q}d(q)<\vx_{i}(q)                                                             \\
                                                        & \iff \bigwedge_{i=1}^{n}\bigvee_{q\in Q}(d \in {\cal D}(Q))\land (d(q)<\vx_{i}(q))                                                          \\
                                                        & \iff \bigvee_{{\bf q}\in Q^{n}} (d \in {\cal D}(Q))\land \bigwedge_{i=1}^{n}(d({\bf q}_{i})<\vx_{i}({\bf q}_{i}))                           \\
    \therefore d\in {\cal D}(Q)\setminus H\implies \phi(d) & \iff \bigvee_{{\bf q}\in Q^{n}} \left[(d \in {\cal D}(Q))\land \bigwedge_{i=1}^{n}(d({\bf q}_{i})<\vx_{i}({\bf q}_{i}))\right]\implies \phi(d) \\
                                                        & \iff \phi(d)\lor\bigwedge_{{\bf q}\in Q^{n}}\neg\left((d \in {\cal D}(Q))\land \bigwedge_{i=1}^{n}(d({\bf q}_{i})<\vx_{i}({\bf q}_{i}))\right) \\
                                                        & \iff \bigwedge_{{\bf q}\in Q^{n}}\left[(d\in {\cal D}(Q))\land \bigwedge_{i=1}^{n}(d({\bf q}_{i})<x_{i}({\bf q}_{i}))\implies\phi(d)\right]
  \end{align*}
\end{proof}

%= ===============================================
\myparagraph{Proof of \cref{lem:csmt_pure_suffice}}
\CSMTpureSuffice*{}
\begin{proof}
  We prove by induction on the length of input to the scheduler. %\todo{If we show memoryless suffice, modify this.}

  The base case will be at input of length $1$. %as the machine always has its current configuration being its history. 
  Fix any configuration $d\in\dist(Q)$, in the case that $d\in H$, it is trivial that any schedulers have an equal probability of $1$ to reach $H$ from $d$. Now assume that $d\in\dist(Q)\setminus H$, let ${\sf Supp}(\sigma(d))\subseteq {\sf Act}$ denotes the support of $\sigma(d)$, which is the set of action that the scheduler has a non-zero probability of taking it. Consider that
  \begin{align*}
    \Pb_{\mathbb{X}^{\csmt{}},\sigma}(\reach(d,H)) & = \sum_{a\in {\sf Supp}(\sigma(d))}\Pb_{\mathbb{X}^{\csmt{}},\sigma}(\reach(d',H) | d \xrightarrow{a} d')\sigma(d)(a) \\
                                                   & \leq \max_{a\in {\sf Supp}(\sigma(d_0))}\Pb_{\mathbb{X}^{\csmt{}},\sigma}(\reach(d',H) | d \xrightarrow{a} d'),
  \end{align*}
  where $d\xrightarrow{a} d'$ denotes the event of the first transition take place is the action $a$ which transform $d$ to $d'$. Since the above holds for any $d$, hence we construct an alternative scheduler $\sigma'$ as follows, for any configuration $d\in\dist(Q)$,
  \begin{eqnarray*}
    \sigma'(d)&=&\argmax_{a\in {\sf Supp}(\sigma(d))}\Pb_{\mathbb{X}^{\csmt{}},\sigma}(\reach(d',H) | d \xrightarrow{a} d'),\\
    \sigma'(\ast)&=&\sigma(\ast).
  \end{eqnarray*}
  Then it naturally follows that
  \begin{align*}
    \Pb_{\mathbb{X}^{\csmt{}},\sigma'}(\reach(d,H)) & = \max_{a\in {\sf Supp}(\sigma(d_0))}\Pb_{\mathbb{X}^{\csmt{}},\sigma}(\reach(d',H)\mid d\xrightarrow{a} d') \\
                                                    & \geq \Pb_{\mathbb{X}^{\csmt{}},\sigma}(\reach(d,H)).
  \end{align*}
  The scheduler $\sigma'$ constructed is pure on the first step and achieve a larger reachability probability. The same construction can be applied inductively to obtain a pure scheduler and a larger reachability probability is always achieved.
\end{proof}

\myparagraph{Proof of \cref{thm:csmt-reach-undec}}
\CSMTundec*{}
\begin{proof}
  We construct a reduction from the emptiness problem of probabilistic finite automaton, which is formally stated as follows.
  \begin{problem}[Emptiness Problem of Probabilistic Finite Automaton]\label{prob:PFAEmpt}
  A probabilistic finite automaton (PFA) ${\cal A}=(\Sigma, Q, d_{0},\alpha, \{M_{a}\}_{a\in\Sigma})$ consists of a finite set of alphabets $\Sigma$, a finite set of states $Q$, an initial configuration $q_{0}\in\dist(Q)$, an accepting vector $\alpha\in\{0,1\}^{\abs{Q}}$ and a set of transition matrices $\{M_{a}\}_{a\in\Sigma}\subseteq \stoc(\abs{Q})$, which is a set of stochastic matrices labeled by the set of alphabets $\Sigma$.

  The emptiness problem of PFA asks given a PFA ${\cal A}$ and a cut-point $\xi\in [0,1]$, decide if there exists a word $a_1 a_2 ...a_n \in \Sigma^{*}$ such that
  $$d_{0}^{T} M_{a_1} M_{a_{2}}... M_{a_n} \alpha > \xi.$$
  \end{problem}

  Observe that the mathematical structure of a PFA is almost identical to that of an MDP, we can naturally construct a corresponding MDP ${\cal M}_{\cal A}$ from a PFA ${\cal A}$ by treating the set of alphabets $\Sigma$ as the set of actions. A target set $H=\{d\in {\cal D}(Q)\mid \alpha^{T}d > \xi\}$ can be constructed so that if $H$ is reachable deterministically, then it implies there is a sequence of actions $\bar{a}=a_{1}a_{2}...a_{n}$ and a sequence of configurations $d_{0} d_{1} ... d_{n}$ such that for any $j\in [n]$,
  $$\delta_{\mathbb{X}^{\csmt{}},\sigma_{d_{0}\bar{a}}}(d_{j-1})(d_{j})=1\text{ and }d_{n}\in H.$$
  From \cref{def:ConfigMCsInstances}, it is not hard to see that
  $$\delta_{\mathbb{X}^{\csmt{}},\sigma_{d_{0}\bar{a}}}(d_{j-1})(d_{j})=1\iff d_{j}=d_{j-1}^{T}M_{a_{j}}.$$
  Hence, the reachability can be rephrased as follows,
  $$(d_{0}^{T} M_{a_1} M_{a_{2}}... M_{a_n})^{T}\in H\iff d^{T} M_{a_1} M_{a_{2}}... M_{a_n}\alpha>\xi,$$
  which implies that $a_{1} a_{2} ... a_{n}$ is an accepting word in the original PFA ${\cal A}$. Therefore, one can clearly see that by solving \ref{prob:csmt-reach-deter} with the MDP ${\cal M}_{\cal A}$ is equivalent to finding a counter example for the emptiness problem of the PFA ${\cal A}$. This concludes the proof.
\end{proof}

% \myparagraph{Proof of \cref{prop:latticeReachable}}
% \begin{proof}
%   Item 1 is shown easily, following the definition of $f^{\sharp}$. Item 2 is shown by induction: (base case) the first inclusion holds since $\bot$ is the minimum, and (step case) once we show $(f^{*})^{i-1}(\bot)\subset (f^{*})^{i}(\bot)$ for $i$,  $(f^{*})^{i}(\bot)\subset (f^{*})^{i+1}(\bot)$ for $i+1$ follows immediately by the monotonicity of $f^{*}$.

%   Item 3 is a consequence of the Kleene fixed point theorem (see e.g.~\cite{AbramskyJ94}). Concretely, it is easy to see that $(f^{*})^{\omega}(\bot)$ is a fixed point, using its definition by a union and that $f^{*}$ is union-preserving. To show that it is the \emph{least} fixed point, let $y$ be any fixed point ($y=f^{*}(y)$). One can  show that $(f^{*})^{n}(\bot)\subseteq y$ for any $n\in\mathbb{N}$; this is by induction on $n$. Then $\mu f^{*}$, being the supremum of the left-hand side, is also below $y$.

%   For Item 4, it suffices to show the following (by the above observations): for each $x\in X$ and $n\in\mathbb{N}$, we have
%   \begin{displaymath}
%     (f^{*})^{n}(\bot)
%     \;=\;
%     \{x\mid \text{$x$ has an $f$-path of length $\le n$ to $T$}\}.
%   \end{displaymath}
%   This is easily proved by induction on $n$.

% \end{proof}

\myparagraph{Proof of \cref{lem:pullback-upC}}
\PullbackUpCset*{}
\begin{proof}
  The proof is straightforward. Since $H$ is upward-closed, thus for any $\vx\in H$ and $\vy\in [0,1]^{n}$, $$\vy\succeq \vx\implies \vy\in H.$$ 
  But then since $M$ is a stochastic matrix, meaning that all of its elements are non-negative. Hence, for any $\vx,vy\in [0,1]^{n}$, we have $\vy\succeq \vx\implies M^{T}\vy\succeq M^{T}\vx$. Combing the two observations gives the following: for any $\vx\in M^{\sharp} H$,
  $$\vy\succeq \vx\implies M^{T}\vy\in H.$$
  which shows that the set $M^{\sharp} H$ is upward-closed as well.
\end{proof}

\myparagraph{Proof of \cref{thm:ReachConfiglfp}}
\ReachConfLfp*{}
\begin{proof}
We start by stating the following standard result in lattice theory. When we use it later, the parameters will be as follows:  $X$ is the configuration space, and $f$ maps a configuration to the set of successor configurations taken over all schedulers (which we can assume to be pure, see \cref{lem:csmt_pure_suffice}).

\begin{proposition}\label{prop:latticeReachable}
  Let $X$ be a set, $f\colon X\to 2^X$ be a function (modeling forward and nondeterministic \emph{dynamics}), and $H\subset X$ be a target set.
  Consider the (complete) subset lattice $2^{X}$ ordered by inclusion $\subset$, and the operator
  \begin{math}
    f^{*}\colon 2^{X} \to 2^{X},
    f^{*}Q = T\cup\{x\mid Q\cap f(x)\neq\emptyset \}.
  \end{math}
  Intuitively, $x\in f^{*}(Q)$ means that $x\in T$ or $x$ is one-step reachable to $Q$.
  Then we have the following.
  \begin{enumerate}
    \item The operator $f^{*}$ is monotone. Moreover, it preserves supremums (i.e.\ unions): $f^{*}(\bigcup_{i}Q_{i})=\bigcup_{i}f^{*}(Q_i)$.
    \item The repeated application of $f^{*}$ to the minimum $\bot=\emptyset$ gives rise to a chain
          \begin{math}
            \bot
            \;\subset\;
            f^{*}(\bot)
            \;\subset\;
            (f^{*})^{2}(\bot)
            \;\subset\;
            \cdots.
          \end{math}
    \item The last chain
          stabilizes at $\omega$, meaning that
          \begin{math}
            f^{*}\bigl((f^{*})^{\omega}(\bot)\bigr)
            = (f^{*})^{\omega}(\bot),
          \end{math}, where
          \begin{math}
            (f^{*})^{\omega}(\bot)
            =
            \textstyle\bigcup_{n\in\mathbb{N}}(f^{*})^{n}(\bot).
          \end{math}
          Moreover, $(f^{*})^{\omega}(\bot)$ is the least fixed point
          $\mu f^{*}$
          of $f^{*}$.
    \item The least fixed point $\mu f^{*}$ coincides with the reachable sets, that is,
          \begin{math}
            \mu f^{*}
            =
            \{x\mid \text{$x$ has an $f$-path  to $T$}\}.
          \end{math}
          Here an $f$-path is defined to be a sequence $x_{0}x_{1}\dotsc x_{n}$ where $x_{i+1}\in f(x_{i})$ for each $i\in [0,n-1]$.

  \end{enumerate}

\end{proposition}

  Item 1 is shown easily, following the definition of $f^{\sharp}$. Item 2 is shown by induction: (base case) the first inclusion holds since $\bot$ is the minimum, and (step case) once we show $(f^{*})^{i-1}(\bot)\subset (f^{*})^{i}(\bot)$ for $i$,  $(f^{*})^{i}(\bot)\subset (f^{*})^{i+1}(\bot)$ for $i+1$ follows immediately by the monotonicity of $f^{*}$.

  Item 3 is a consequence of the Kleene fixed point theorem (see e.g.~\cite{AbramskyJ94}). Concretely, it is easy to see that $(f^{*})^{\omega}(\bot)$ is a fixed point, using its definition by a union and that $f^{*}$ is union-preserving. To show that it is the \emph{least} fixed point, let $y$ be any fixed point ($y=f^{*}(y)$). One can  show that $(f^{*})^{n}(\bot)\subseteq y$ for any $n\in\mathbb{N}$; this is by induction on $n$. Then $\mu f^{*}$, being the supremum of the left-hand side, is also below $y$.

  For Item 4, it suffices to show the following (by the above observations): for each $x\in X$ and $n\in\mathbb{N}$, we have
  \begin{displaymath}
    (f^{*})^{n}(\bot)
    \;=\;
    \{x\mid \text{$x$ has an $f$-path of length $\le n$ to $T$}\}.
  \end{displaymath}
  This is easily proved by induction on $n$.

  Now we proceed to prove the \cref{thm:ReachConfiglfp}. Let $S^{i}$ denote the value of $S$ after the $i$-th iteration. It suffices to show that $\uparrow S^{i}\subseteq B^{i}(\bot)$. The base case $i=0$ is trivial as $B^{0}(\bot)=\bot=\uparrow S^{0}$, we now focus on the inductive case.

  \noindent
  Inductively, assume $\uparrow S^{i}\subseteq B^{i}(\bot)$ for some $i\in\mathbb{N}$, then for any $\vx\in \uparrow S^{i+1}$, by construction we have either $\vx\in H$, or that there exists an action $a\in\act$ and two vectors $\vx'\in S^{i},\vy^{\ast}\in[0,1]^{\abs{Q}}$ such that $\vx\succeq\vy^{\ast}$ and $M_{a}\vy^{\ast}\in \uparrow S^{i}$. Combining two gives $M_{a}\vx\in\uparrow S^{i}\subseteq B^{i}(\bot)$. In the former case we have $\vx\in H\subseteq B^{i+1}(\bot)$, while for the latter case $M_{a}\vx\in B^{i}(\bot)\implies \vx\in M_{a}\sharp B^{i}(\bot)\implies \vx\in B^{i+1}(\bot)$. In both case we have $\vx\in B^{i+1}(\bot)$ which finishes the induction. The rest follows from \cref{prop:latticeReachable}.
  %\todo{(IH) The second part which was there is subsumed by \cref{prop:latticeReachable}, in my understanding}    
  % Next we show that $d\in B^{k}(\bot)\implies \exists a_{1},...,a_{j}\in \act$, $j<k$, such that $ M_{a_{j}}^{T}M_{a_{j-1}}^{T}...M_{a_{1}}^{T}d\in H$.\\
  % Induction on $k$, trivial for $k=1$ as $B^{1}=H$. Inductively, for some $d\in B^{k+1}(\bot)$, noted that by \cref{lem:propB} $B^{k}(\bot)\subseteq B^{k+1}(\bot)$. Hence it suffices to consider $d\in B^{k+1}\setminus B^{k}$, since $H=B(\bot)\subseteq B^{k}(\bot)\implies d\notin H$, thus by \cref{def:BPop}, $\exists M\in\mat$ such that $d\in M^{\sharp} B^{k}(\bot)\implies Md\in B^{k}(\bot)$, then the rest follows from inductive hypothesis.
\end{proof}

\myparagraph{Proof of \cref{lem:anti-approx}}
\ApproxLowerSet*{}
\begin{proof}
  We start with the former statement, we aim to show for any $k\in [K]$, $\vy_{k}^{\ast}\in  \floor{\{\vy\in[0,1]^{\abs{Q}}\mid M\vy\succeq\vx \}}.$
On the contrary, assume there exists some $k\in [K]$ and $\vy_{k}'\neq \vy_{k}^{\ast}\in [0,1]^{n}$ such that $\vy_{k}'\preceq \vy^{\ast}_{k}$. Then it is not hard to check that $\vy_{k}'$ is also a feasible solution to the $k$-th minimization problem and $\mathbf{1}^{T}\vy_{k}'<\mathbf{1}^{T}\vy_{k}^{\ast}$. This contradicts with $\vy^{\ast}_{k}$ being the minimum solution. The latter statement is easy: since the set $\{\vy\in[0,1]^{\abs{Q}}\mid M\vy\succeq\vx \}$ is finitely-generated, it admits a finite antichain representation, and since there is a non-zero chance for picking the right $\mathbf{q}$ such that the linear program has a solution, hence the probability of finding all the antichain elements can be as close to $1$ as possible by fixing $K$ to be the number of antichain elements and let $L$ to be as large as possible.
\end{proof}

\section{Experiment Setup and Results}\label{sec:ExperAppendix}
%While the primary focus of this paper is the theoretical formalization of the unifying framework and the two new semantics, we have also implemented prototypes of both algorithms in \cref{subsec:csmtAlgo,subsec:msctAlgo}, as a proof of concept. 
Both algorithms were implemented in Python 3.12, utilizing {\texttt PySMT} 0.9 \cite{pysmt2015} and {\texttt PolyQEnt} 0.8 \cite{polyqent} for constraint collection and quantifier elimination in the template-based algorithm for the \msct{} semantics, as well as {\texttt Pyomo} 6.8 \cite{pyomo} for solving the optimization problem in the antichain-based algorithm for the \csmt{} semantics. Additionally, {\texttt z3} \cite{10.1007/978-3-540-78800-3_24} and {\texttt glpk} \cite{glpk} have been used as backend solvers. Both experiments were conducted on a high-performance server (AMD EPYC 7702P processor with 128 GB of RAM).

\myparagraph{Models} Each algorithm is tested against 3 to 4 models and each of them are tested for 10 trials to obtain the average runtime. The models for testing include the running example from \cite{Akshay2023}, the pharmacokinetics system from \cite{DBLP:conf/qest/ChadhaKVAK11}, and our motivating \emph{casino} and \emph{Exam} example from \cref{sec:intro}, we refer them as {\sf Toy}, {\sf Pharmacokinetics}, {\sf Casino} and {\sf Exam} respectively. For the model {\sf Toy} and {\sf Pharmacokinetics}, the input is chosen as the same as in \cite{Akshay2023} and \cite{DBLP:conf/qest/ChadhaKVAK11} respectively; while for the model {\sf Casino} and {\sf Exam}, the initial configuration and target set are fixed and the transition probabilities are randomly generated.

\myparagraph{Parameters} For the antichain-based algorithm, since the algorithm is not guarantee to terminate, we set a maximum bound of 100 loop to limit the testing time. Also recall that the algorithm involves a randomization process which takes two parameters $K$ and $L$. We set $K=3$ for the {\sf Toy} and {\sf Pharmacokinetics} models, and $K=5$ for the {\sf Casino} model. The parameter $L$ is fixed to $1$ for all models. For the template-based algorithm, we set $\gamma$ for the submartingale to be $1-10^{-5}$. When applying the Handelman's theorem for the quantifier elimination, only polynomials up to degree 4 are considered.

\begin{table}[tbp]\centering
  \caption{Experiment results, for the antichain-based algorithm on \csmt{} semantics (left), and for the template-based algorithm on \msct{} semantics (right). For the antichain algorithm we included the number of runs that successfully show the target is reachable. The average runtime is in seconds.}
  \scalebox{.8}{\begin{minipage}{.6\textwidth}
      \begin{tabular}{ccc}
                       & \msct{} semantics &                \\\cmidrule{2-2}
        Models         & AvgTime(s)        & $\#$ Success \\\midrule
        {\sf Toy}      & 0.4767            & 10             \\
        {\sf Pharmaco} & 0.3066            & 10             \\
        {\sf Exam-2-2} & 5.648             & 7              \\
        {\sf Exam-5-2} & 3.479             & 10             \\\bottomrule
      \end{tabular}
    \end{minipage}
  }
  \hfill
  \scalebox{.8}{\begin{minipage}{.6\textwidth}
      \begin{tabular}{cc}
                         & \csmt{} semantics\\\cmidrule{2-2}
        Models           & AvgTime(s)\\\midrule
        {\sf Toy}        & 50.44\\
        {\sf Casino-5-2} & 1.278\\
        {\sf Casino-2-3} & 1.360\\\bottomrule
      \end{tabular}
    \end{minipage}
  }%\vspace{-2em}
  \label{table:msct-exper}
  \label{table:csmt-exper}
\end{table}

\myparagraph{Results} We summarize our results in \cref{table:csmt-exper,table:msct-exper}. Since our algorithm is currently unoptimized and input-sensitive, runtime is not our primary focus. For the antichain-based algorithm, the {\sf Toy} and {\sf Pharmacokinetics} models serve as simple validation cases, with the latter featuring a more complex transition structure; while the {\sf exam} model is a more challenging test case. Our algorithm efficiently solves both. %The {\sf exam} model is a more challenging test case, typically containing reachable paths of length 60-70. Our algorithm successfully identifies and reports these paths within a reasonable runtime.

For the template-based algorithm, our algorithm requires significant time to solve the {\sf Toy} model, which aligns with observations in \cite{Akshay2023}, where additional hints were necessary for the solver to achieve reasonable performance. Conversely, our algorithm performs exceptionally well on the two {\sf Casino} models. A potential explanation for this discrepancy is that the winning probabilities for each machine are sampled uniformly, randomly, and independently. As a result, a constant uniform scheduler is likely sufficient to reach the desired target set efficiently.
\end{document}